\documentclass[11pt]{article}

\usepackage{mystyle}

\usepackage{cite}
\usepackage{pbox}
\usepackage{pgfplots}
\usepackage{subfigure,color}
\usepackage[update,prepend]{epstopdf}
\usepackage[outercaption]{sidecap}

\newlength{\figsize} 

\newcommand{\RR}{\textsc{ReduceRank}\xspace}
\newcommand{\RRPFD}{\textsc{ReduceRank-PFD}\xspace}
\newcommand{\RRSS}{\textsc{ReduceRank-SS}\xspace}
\newcommand{\R}{\mathbb{R}}

\title{\sffamily Improved Practical Matrix Sketching with Guarantees}
\author{
Amey Desai
\and
Mina Ghashami
\and
Jeff M. Phillips
\thanks{Thanks to support by NSF CCF-1115677, CCF-1350888, IIS-1251019, and ACI-1443046.}
}

\begin{document}

\maketitle

\begin{abstract}
Matrices have become essential data representations for many large-scale problems in data analytics, and hence matrix sketching is a critical task.  Although much research has focused on improving the error/size tradeoff under various sketching paradigms, the many forms of error bounds make these approaches hard to compare in theory and in practice.  This paper attempts to categorize and compare most known methods under row-wise streaming updates with provable guarantees, and then to tweak some of these methods to gain practical improvements while retaining guarantees.  

For instance, we observe that a simple heuristic \iSVD, with no guarantees, tends to outperform all known approaches in terms of size/error trade-off.  We modify the best performing method with guarantees \fd under the size/error trade-off to match the performance of \iSVD and retain its guarantees.  We also demonstrate some adversarial datasets where \iSVD performs quite poorly.  
In comparing techniques in the time/error trade-off, techniques based on hashing or sampling tend to perform better.  In this setting we modify the most studied sampling regime to retain error guarantee but obtain dramatic improvements in the time/error trade-off.

Finally, we provide easy replication of our studies on APT, a new testbed which makes available not only code and datasets, but also a computing platform with fixed environmental settings. 
\end{abstract}

\section{Introduction}
\label{sec:intro}
Matrix sketching has become a central challenge~\cite{frieze2004fast,papadimitriou1998latent,drineas2003pass,achlioptas2001fast,liberty2013simple} in large-scale data analysis as many large data sets including customer recommendations, image databases, social graphs, document feature vectors can be modeled as a matrix, and sketching is either a necessary first step in data reduction or has direct relationships to core techniques including PCA, LDA, and clustering.  

There are several variants of this problem, but in general the goal is to process an $n \times d$ matrix $A$ to somehow represent a matrix $B$ so $\|A - B\|_F$ or (examining the covariance) $\|A^T A - B^T B\|_2$ is small.  

In both cases, the best rank-$k$ approximation $A_k$ can be computed using the singular value decomposition (\svd); however this takes $O(nd \min(n,d))$ time and $O(nd)$ memory. 
This is prohibitive for modern applications which usually desire a small space streaming approach, or even an approach that works in parallel.  
For instance diverse applications receive data in a potentially unbounded and time-varying stream and want to maintain some sketch $B$.  
Examples of these applications include data feeds from sensor networks ~\cite{bonnet2001}, financial tickers~\cite{chen2000niagaracq , zhu2002statstream}, on-line auctions~\cite{arasu2002abstract}, network traffic~\cite{gilbert2001quicksand , sullivan1998system}, and telecom call records~\cite{cortes2000hancock}.

In recent years, extensive work has taken place to improve theoretical bounds in the size of $B$.  
Random projections~\cite{sarlos2006improved,achlioptas2001fast} and hashing~\cite{clarkson2009numerical,WDALS09} approximate $A$ in $B$ as a random linear combination of rows and/or columns of $A$.  
Column sampling methods~\cite{drineas2003pass, drineas2008relative, mahoney2009cur,boutsidis2011near,deshpande2006adaptive, drineas2006fast2, rudelson2007sampling} choose a set of columns (and/or rows) from $A$ to represent $B$; the best bounds require multiple passes over the data.  
We refer readers to recent work~\cite{ghashami2014relative,clarkson2009numerical,Woo14} for extensive discussion of various models and error bounds.  

Recently Liberty~\cite{liberty2013simple} introduced a new technique \fd (abbreviated \FD) which is deterministic, achieves the best bounds on the covariance $\|A^TA - B^T B\|_2$ error, the direct error $\|A - B\|_F^2$~\cite{ghashami2014relative} (using $B$ as a projection), and moreover, it seems to greatly outperform the projection, hashing, and column sampling techniques in practice.  
Hence, this problem
has seen immense progress in the last decade with a wide variety of algorithmic improvements in a variety
of models\cite{sarlos2006improved,achlioptas2001fast,clarkson2009numerical,dasgupta2010sparse,drineas2006fast2,drineas2008relative,papapailiopoulos2014provable,deshpande2006adaptive, drineas2006fast2,ghashami2014relative,clarkson2013low,liberty2013simple}; which we review thoroughly in Section \ref{sec:mat-algo} and assess comprehensively empirically in Section \ref{sec:eval}.

In addition, there is a family of heuristic techniques~\cite{golub2012matrix, hall1998incremental, levey2000sequential, brand2002incremental,ross2008incremental} (which we refer to as \iSVD, described relative to \FD in Section \ref{sec:algos}), which are used in many practical settings, but are not known to have any error guarantees.  
In fact, we observe (see Section \ref{sec:exp}) on many real and synthetic data sets that \iSVD noticeably outperforms \FD, yet there are adversarial examples where it fails dramatically.  
Thus we ask (and answer in the affirmative): can one achieve a matrix sketching algorithm that matches the usual-case performance of \iSVD, and the adversarial-case performance of \FD, and error guarantees of \FD?

\subsection{Notation and Problem Formalization}
We denote an $n \times d$ matrix $A$ as a set of $n$ rows as $[a_1; a_2; \ldots, a_n]$ where each $a_i$ is a row of length $d$.  Alternatively a matrix $V$ can be written as a set of columns $[v_1, v_2, \ldots, v_d]$.  We assume $d \ll n$.  
We will consider streaming algorithms where each element of the stream is a row $a_i$ of $A$.  Some of the algorithms also work in a more general setting (e.g., allowing deletions or distributed streams).  

The squared Frobenius norm of a matrix $A$ is defined $\|A\|_F^2 = \sum_{i=1} \|a_i\|^2$ where $\|a_i\|$ is Euclidean norm of row $a_i$, and it intuitively represents the total size of $A$.  
The spectral norm $\|A\|_2 = \max_{x : \|x\|=1} \|A x\|$, and represents the maximum influence along any unit direction $x$.  It follows that $\|A^T A - B^T B\|_2 = \max_{x : \|x\|=1} | \|A x\|^2 - \|B x\|^2 |$.  

Given a matrix $A$ and a low-rank matrix $X$ let $\pi_X(A) = A X^\dagger X$ be a \emph{projection} operation of $A$ onto the rowspace spanned by $X$; that is if $X$ is rank $r$, then it projects to the $r$-dimensional subspace of points (e.g. rows) in $X$.  
Here $X^\dagger$ indicates taking the Moore-Penrose pseudoinverse of $X$.
The singular value decomposition of $A$, written $\svd(A)$, produces three matrices $[U,S,V]$ so that $A = U S V^T$.  Matrix $U$ is $n \times n$ and orthogonal.  Matrix $V$ is $d \times d$ and orthogonal; its columns $[v_1, v_2, \ldots, v_d]$ are the right singular vectors, describing directions of most covariance in $A^T A$.  $S$ is $n \times d$ and is all $0$s except for the diagonal entries $\{\sigma_1, \sigma_2, \ldots, \sigma_r\}$, the \emph{singular values}, where $r \leq d$ is the rank.  Note that $\sigma_j \geq \sigma_{j+1}$, $\|A\|_2 = \sigma_1$, and $\sigma_j = \|A v_j\|$ describes the norm along direction $v_j$.
\paragraph{Error measures.}
We consider two classes of error measures between input matrix $A$ and its $\ell \times d$ sketch $B$.  
The \emph{covariance error} ensures for any unit vector $x \in \R^d$ such that $\|x\|=1$, that $\|A x\|^2 - \|B x\|^2$ is small.  This can be mapped to the covariance of $A^T A$ since 
$\max_{\|x\|=1} \|A x\|^2 - \|B x\|^2 = \|A^TA - B^T B\|_2$.  To normalize \emph{covariance error} we set
\[
\s{cov-err}(A,B) = \|A^T A - B^T B\|_2 / \|A\|_F^2
\]
which is always greater than $0$ and will typically be less than $0.15$.%
\footnote{The normalization term $\|A\|_F^2$ is invariant to the desired rank $k$ and the unit vector $x$.  Some methods have bounds on $\|Ax\|^2 - \|Bx\|^2$ that are relative to $\|A - A_k\|^2$ or $\|Ax\|^2$; but as these introduce an extra parameter they are harder to measure empirically.}  

The \emph{projection error} describes the error in the subspace captured by $B$ without focusing on its scale.  Here we compare the best rank $k$ subspace described by $B$ (as $B_k$) compared against the same for $A$ (as $A_k$).  We measure the error by comparing the tail (what remains after projecting $A$ onto this subspace) as $\|A - \pi_{B_k}(A)\|_F^2$.  Specifically we normalize \emph{projection error} so that it is comparable across datasets as 
\[
\s{proj-err}(A,B) = \|A - \pi_{B_k}(A)\|_F^2 / \|A - A_k\|_F^2.
\]  
Note the denominator is equivalent to $\|A - \pi_{A_k}(A)\|^2$ so this ensures the projection error is at least $1$; typically it will be less than $1.5$.  We set $k=10$ as a default in all experiments.%
\footnote{There are variations in bounds on this sort of error.  Some measure spectral (e.g. $\|\cdot\|_2$) norm.  Others provide a weaker error bound on $\|A - [\pi_{B}(A)]_k\|_F^2$, where the ``best rank $k$ approximation,'' denoted by $[\cdot]_k$, is taken after the projection.  This is less useful since then for a very large rank $B$ (might be rank $500$) it is not clear which subspace best approximates $A$ until this projection is performed.  
Additionally, some approaches create a set of (usually three) matrices (e.g. $CUR$) with product equal to $\pi_{B_k}(A)$, instead of just $B$.  This is a stronger result, but it does not hold for many approaches, so we omit consideration.}

\subsection{Frequency Approximation, Intuition, and Results}  
There is a strong link between matrix sketching and the frequent items (similar to heavy-hitters) problem commonly studied in the streaming literature.  That is, given a stream $S = \langle s_1, s_2, \ldots, s_n \rangle$ of items $s_i \in [u] = \{1,2,\ldots,u\}$, represent $f_j = |\{ s_i \in S \mid s_i = j\}|$; the frequency of each item $j \in [u]$.  The goal is to construct an estimate $\hat f_j$ (for \emph{all} $j \in [u]$) so that $| f_j - \hat f_j| \leq \eps n$.
There are many distinct solutions to this problem, and many map to families of matrix sketching paradigms.  

The Misra-Gries~\cite{mg-fre-82} (MG) frequent items sketch inspired the \FD matrix sketching approach, and the extensions discussed herein.  The MG sketch uses $O(1/\eps)$ space to keep $\ell -1 = 1/\eps$ counters, each labeled by some $j \in [u]$: it increments a counter if the new item matches the associated label or for an empty counter, and it decrements all counters if there is no empty counter and none match the stream element.  $\hat f_j$ is the associated counter value for $j$, or $0$ if there is no associated counter.  
There are other variants of this algorithm~\cite{karp2003simple,demaine2002frequency,metwally2006integrated} that have slightly different properties~\cite{cormode2008finding,agarwal2012mergeable} that we will describe in Section \ref{sec:new-iter} to inspire variants of \FD.

A folklore approach towards the frequent items problem is to sample $\ell = O((1/\eps^2) \log 1/\delta)$ items from the stream.  Then $\hat f_j$ is the count of the sample scaled by $n/\ell$. This achieves our guarantees with probability at least $1-\delta$~\cite{VC71,LLS01}, and will correspond with sampling algorithms for sketching which sample rows of the matrix.  

Another frequent items sketch is called the \emph{count-sketch}~\cite{CCF02}.  It maintains and averages the results of $O(\log 1/\delta)$ of the following sketch.  Each $j \in [u]$ is randomly hashed $h : [u] \to [\ell]$ to a cell of a table of size $\ell = O(1/\eps^2)$; it is added or subtracted from the row based on another random hash $s : [u] \to \{-1,+1\}$.  Estimating $\hat f_j$ is the count at cell $h(j)$ times $s(j)$.  This again achieves the desired bounds with probability at least $1-\delta$.  
This work inspires the \emph{hashing} approaches to matrix sketching.  
If each element $j$ had its effect spread out over all $O(\log 1/\delta)$ instances of the base sketch, it would then map to the \emph{random projection} approaches to sketching, although this comparison is a bit more tenuous.  

Finally, there is another popular frequent items sketch called the count-min sketch~\cite{CM05}.  It is similar to the count-sketch but without the sign hash, and as far as we know has not been directly adapted to a matrix sketching algorithm.

\subsection{Contributions}
We survey and categorize the main approaches to matrix sketching in a row-wise update stream.  We consider three categories:  sampling, projection/hashing, and iterative and show how all approaches fit simply into one of these three categories.  We also provide an extensive set of experiments to compare these algorithms along size, error (projection and covariance), and runtime on real and synthetic data.  

To make this study easily and readily \emph{reproducible}, we implement all experiments on a new extension of Emulab~\cite{emulab} called Adaptable Profile-Driven Testbed or APT~\cite{apt}.  It allows one to check out a virtual machine with the same specs as we run our experiments, load our precise environments and code and data sets, and directly reproduce all experiments. 

We also consider new variants of these approaches which maintain error guarantees but significantly improving performance.  We introduce several new variants of \FD, one of which $\alpha$-\FD matches or exceeds the performance of a popular heuristic \iSVD.  Before this new variant, \iSVD is a top performer in space/error trade-off, but has no guarantees, and as we demonstrate on some adversarial data sets, can fail spectacularly.   
We also show how to efficiently implement and analyze \emph{without-replacement} row sampling for matrix sketching, and how this can empirically improve upon more traditional (and easier to analyze) with-replacement row sampling.

\section{Matrix Sketching Algorithms}
\label{sec:mat-algo}\label{sec:algos}
In this section we review the main algorithms for sketching matrices.  We divide them into 3 main categories:
(1) sampling algorithms, these select a subset of rows from $A$ to use as the sketch $B$;
(2) projection algorithms, these project the $n$ rows of $A$ onto $\ell$ rows of $B$, sometimes using hashing;
(3) incremental algorithms, these maintain $B$ as a low-rank version of $A$ updated as more rows are added.  

There exist other forms of matrix approximation, for instance, using sparsification techniques~\cite{arora2006fast,achlioptas2001fast,drineas2011note}, or allowing more general element-wise updates at the expense of larger sketch sizes~\cite{clarkson2009numerical,sarlos2006improved,clarkson2013low}.  We are interested in preserving the right singular vectors and other statistical properties on the rows of $A$.  

\subsection{Sampling Matrix Sketching}
\begin{table}[h!!!]
\vspace{-5mm}
\begin{center}
\resizebox{\textwidth}{!}{
\begin{tabular}{|r||c|c||c|c||c|}
\hline
\textbf{} & $\ell$ & \s{cov-err} & $\ell$ & \s{proj-err} & \s{runtime}
\\ \hline 
\s{Norm Sampling} & $d/\eps^2$ & $\eps$  ($\dagger$) \cite{drineas2006fast2}  & $k/\eps^2$ & $1+\eps \frac{\|A\|_F^2}{\|A - A_k\|_F^2}$~\cite{drineas2006fast2} & $\s{nnz}(A) \cdot \ell$
\\
\s{Leverage Sampling} & $d/\eps^2$& $\eps$ ($\dagger$) & $(k \log k)/\eps^2$ & $1+\eps$\cite{mahoney2009cur} & $\svd(A) + \s{nnz}(A) \cdot \ell$
\\ 
\s{Deterministic Leverage} & $\ell$  & - & $(k/\eta \eps)^{1/\eta}$($\star$)  & $1+\eps$ \cite{papapailiopoulos2014provable} & $\svd(A) + \s{nnz}(A) \cdot \ell \log \ell$
\\ \hline 
\end{tabular} 
}
\end{center}
\vspace{-2mm}
\caption{\label{tbl:compare-samp} Theoretical Bounds for Sampling Algorithms.  
The \s{proj-err} bounds are based on a slightly weaker $\|A - \pi_B(A)\|_F^2$ numerator instead of $\|A - \pi_{B_k}(A)\|_F^2$ one where we first enforce $B_k$ is rank $k$.  
($\dagger$) See Appendix \ref{app:cov-err}; the \s{Leverage Sampling} bound assumes a constant lower bound on leverage scores.  
($\star$) Maximum of this and $\{k, (k/\eps)^{1/(1+\eta)}\}$ where leverage scores follow power-law with decay exponent $1+\eta$. }
\end{table}

Sampling algorithms assign a probability $p_i$ for each row $a_i$ and then selecting $\ell$ rows from $A$ into $B$ using this probability.  In $B$, each row has its squared norm rescaled to $w_i$ as a function of $p_i$ and $\|a_i\|$.  
One can achieve additive error bound using importance sampling with $p_i = \|a_i\|^2/\|A\|_F^2$ and $w_i = \|a_i\|^2/(\ell p_i) = \|A\|_F^2/\ell$, as analyzed by Drineas \etal~\cite{drineas2006fast1} and\cite{frieze2004fast}.
These algorithms typically advocate sampling $\ell$ items independently (with replacement) using $\ell$ distinct reservoir samplers, taking $O(\ell)$ time per element.  Another version~\cite{drineas2008relative} samples each row independently, and only retains $\ell$ rows in expectation.  
We discuss two improvements to this process in Section \ref{sec:new-SwoR}.  

Much of the related literature describes selecting columns instead of rows (called the \emph{column subset selection problem})~\cite{BMD09}.  This is just a transpose of the data and has no real difference from what is described here.  There are also techniques~\cite{drineas2008relative} that select both columns and rows, but are orthogonal to our goals.  

This family of techniques has the advantage that the resulting sketch is \emph{interpretable} in that each row of $B$ corresponds to data point in $A$, not just a linear combination of them.  

\paragraph{Leverage Sampling.}
An insightful adaptation changes the probability $p_i$ using \emph{leverage scores}~\cite{DM10} or \emph{simplex volume}~\cite{deshpande2006matrix,DR10}.  These techniques take into account more of the structure of the problem than simply the rows norm, and can achieve stronger relative error bounds.  But they also require an extra parameter $k$ as part of the algorithm, and for the most part require much more work to generate these modified $p_i$ scores.  
We use \s{Leverage Sampling}~\cite{drineas2008relative} as a representative; it samples rows according to leverage scores (described below).  Simplex volume calculations~\cite{deshpande2006matrix,DR10} were too involved to be practical.  There are also recent techniques to improve on the theoretical runtime for leverage sampling~\cite{DMMW12} by approximating the desired values $p_i$, but as the exact approaches do not demonstrate consistent tangible error improvements, we do not pursue this complicated theoretical runtime improvement.  

To calculate leverage scores, we first calculate the \svd\ of $A$ (the task we hoped to avoid).  Let $U_k$ be the matrix of the top $k$ left singular vectors, and let $U_k(i)$ represent the $i$th row of that matrix.  Then the \emph{leverage score} for row $i$ is $s_i = \| U_k(i)\|^2$, the fraction of squared norm of $a_i$ along subspace $U_k$.  Then set $p_i$ proportional to $s_i$ (e.g. $p_i = s_i / k$. Note that $\sum_i s_i = k$).

\paragraph{Deterministic Leverage Scores.}Another option is to deterministically select rows with the highest $s_i$ values instead of at random.  It can be implemented with a simple priority queue of size $\ell$.  
This has been applied to using the leverage scores by Papailiopoulos \etal \cite{papapailiopoulos2014provable}, which again first requires calculating the \svd\ of $A$.  We refer to this algorithm as \s{Deterministic Leverage} Sampling.

\subsection{Projection Matrix Sketching}
\begin{table}[h!!!]
\vspace{-.1in}
\begin{center}
\resizebox{\textwidth}{!}{
\begin{tabular}{|r||c|c||c|c||c|}
\hline
\textbf{} & $\ell$ & \s{cov-err} & $\ell$ & \s{proj-err} & \s{runtime}
\\ \hline 
\s{Random Projection} & $d/\eps^2$\cite{sarlos2006improved}  & $\eps / \rho(A)$ & $d/\eps^2$ \cite{sarlos2006improved} & $1+\eps$ & $\s{nnz}(A) \cdot \ell$ 
\\
\s{Fast JLT} & $d/\eps^2$ \cite{sarlos2006improved} &  $\eps / \rho(A)$ & $d/\eps^2$ \cite{sarlos2006improved} & $1+\eps$& $nd\log d + (d/\eps^2)\log n$~\cite{ailon2006approximate} 
\\ 
\s{Hashing} & $d^2/\eps^2$ \cite{clarkson2013low,NN13}  &  $\eps / \rho(A)$ & $d^2/\eps^2$ \cite{clarkson2013low,NN13}  & $1+\eps$  & $\s{nnz}(A) +n\poly{(d/\eps)}$
\\ 
\s{OSNAP} &  $d^{1+o(s/\eps)}/\eps^2$ \cite{NN13} &  $\eps / \rho(A)$ & $d^{1+o(s/\eps)}/\eps^2$   \cite{NN13} & $1+\eps$  & $\s{nnz}(A) \cdot s + n\poly{(d/\eps)}$
\\ \hline 
\end{tabular} 
}
\end{center}
\vspace{-2mm}
\caption{\label{tbl:compare-proj} Theoretical Bounds for Projection Algorithms (via an $\ell_2$ subspace embedding; see Appendix \ref{app:rp-convert}).  
Where $\ell$ is the number of rows maintained, and $\rho(A) = \frac{\|A\|_F^2}{\|A\|_2^2} \geq 1$ is the \emph{numeric rank} of $A$.  
}
\end{table}

These methods linearly project the $n$ rows of $A$ to $\ell$ rows of $B$.  A survey by Woodruff~\cite{Woo14} (especially Section 2.1) gives an excellent account of this area.  
In the simplest version, each row $a_i \in A$ would map to a row $b_j \in B$ with element $s_{j,i}$ ($j$th row and $i$th column) of a projection matrix $S$, and each $s_{j,i}$ is a Gaussian random variable with $0$ mean and $\sqrt{n/\ell}$ standard deviation.  That is, $B = S A$, where $S$ is $\ell \times n$.  
This follows from the celebrated Johnson-Lindenstrauss lemma~\cite{johnson1984extensions} as first shown by Sarlos~\cite{sarlos2006improved} and strengthened by Clarkson and Woodruff~\cite{clarkson2009numerical}.  
Gaussian random variables $s_{j,i}$ can be replaced with (appropriately scaled) $\{-1, 0, +1\}$ or $\{-1,+1\}$ random variables~\cite{achlioptas2003database}.  
We call the version with scaled $\{-1,+1\}$ random variables as \s{Random Projection}.  

\paragraph{Fast JLT.}
Using a sparse projection matrix $X$ would improve the runtime, but these lose guarantees if the input is also sparse (if the non-zero elements do not align).  This is circumvented by rotating the space with a Hadamard matrix~\cite{ailon2006approximate}, which can be applied more efficiently using FFT tricks, despite being dense.  More precisely, we use three matrices: 
$P$ is $\ell \times n$ and has entries with iid $0$ with probability $1-q$ and a Gaussian random variable with variance $\ell/q$ with probability $q = \min\{1, \Theta((\log^2 n)/d)\}$.  
$H$ is $n \times n$ and a random Hadamard (this requires $n$ to be padded to a power of $2$).  
$D$ is diagonal with random $\{-1,+1\}$ in each diagonal element.  
And then the projection matrix is $S = PHD$, although algorithmically the matrices are applied implicitly.  We refer to this algorithm as \s{Fast JLT}.  
Ultimately, the runtime is brought from $O(n d \ell)$ to $O(nd \log d + (d/\eps^2) \log n)$.  
The second term in the runtime can be improved with more complicated constructions~\cite{dasgupta2010sparse,AL11} which we do not pursue here; we point the reader here~\cite{VW11} for a discussion of some of these extensions.

\paragraph{Sparse Random Projections.}
Clarkson and Woodruff~\cite{clarkson2013low} analyzed a very sparse projection matrix $S$, conceived of earlier~\cite{dasgupta2010sparse,WDALS09}; it has exactly $1$ non-zero element per column.  To generate $S$, for each column choose a random value between $1$ and $\ell$ to be the non-zero, and then choose a $-1$ or $+1$ for that location.  Thus each rows can be processed in time proportional to its number of non-zeros; it is randomly added or subtracted from $1$ row of $B$, as a count sketch~\cite{CCF02} on rows instead of counts.  We refer this as \s{Hashing}.  

A slight modification by Nelson and Nguyen~\cite{NN13}, called \s{OSNAP}, stacks $s$ instances of the projection matrix $S$ on top of each other.  If \textsc{Hashing} used $\ell'$ rows, then \s{OSNAP} uses $\ell = s \cdot \ell'$ rows (we use $s=4$).

\subsection{Iterative Matrix Sketching}
\label{sec:iterative}
\begin{table}[h!!!]
\begin{center}
\begin{tabular}{|r||c|c||c|c|c|}
\hline
\textbf{} & $\ell$ & \s{cov-err} & $\ell$ & \s{proj-err} & \s{runtime}
\\ \hline 
\s{Frequent Directions} & $k + 1/\eps$ & $\eps \frac{\|A - A_k\|_F^2}{\|A\|_F^2}$ \cite{FD-journal} & $k/\eps$ & $1+\eps$ \cite{ghashami2014relative} & $n d \ell$ 
\\
\s{Iterative SVD} & $\ell$ & - & $\ell$ & - & $nd\ell^2$
\\ \hline 
\end{tabular} 
\end{center}
\vspace{-2mm}
\caption{\label{tbl:compare-iter} Theoretical Bounds for Iterative Algorithms.}
\end{table}

The main structure of these algorithms is presented in Algorithm \ref{alg:generic}, they maintain a sketch $B_{[i]}$ of $A_{[i]}$, the first $i$ rows of $A$.  The sketch of $B_{[i-1]}$ always uses at most $\ell-1$ rows.  On seeing the $i$th row of $A$, it is appended to $[B_{[i-1]}; a_{i}] \rightarrow B_{[i]}$, and if needed the sketch is reduced to use at most $\ell-1$ rows again using some $\RR$ procedure.
Notationally we use $\sigma_j$ as the $j$th singular value in $S$, and $\sigma_j'$ as the $j$th singular value in $S'$.  

\begin{algorithm}
\caption{\label{alg:generic} (Generic) \FD Algorithm}
\begin{algorithmic}
\STATE \textbf{Input:} $\ell, \alpha \in (0,1], A \in \bb{R}^{n \times d}$
\STATE $B_{[0]} \leftarrow$ all zeros matrix $\in \bb{R}^{\ell \times d}$ 
\FOR {$i \in [n]$}
  \STATE  Insert $a_i$ into a zero valued rows of $B_{[i-1]}$; result is $B_{[i]}$
  \IF {($B_{[i]}$ has no zero valued rows)}
  	\STATE  $[U, S, V] \leftarrow \svd(B_{[i]})$
  	\STATE $C_{[i]} = S V^T$ \hfill \# Only needed for proof notation
  	\STATE $S' \leftarrow \text{\RR}(S)$
  	\STATE $B_{[i]} \leftarrow S' V^T$
  \ENDIF 
\ENDFOR
\STATE \textbf{return} $B = B_{[n]}$ 
\end{algorithmic}
\end{algorithm}

\paragraph{Iterative SVD.}
The simplest variant of this procedure is a heuristic rediscovered several times~\cite{golub2012matrix, hall1998incremental, levey2000sequential, brand2002incremental,ross2008incremental}, with a few minor modifications, and we  refer to as \emph{iterative SVD} or \s{iSVD}.  
Here \s{ReduceRank$(S,V)$} simply keeps ${\sigma'}_j = \sigma_j$ for $j < \ell$ and sets $\sigma'_{\ell} = 0$.
This has no worst case guarantees (despite several claims).

\paragraph{Frequent Directions.}
Recently Liberty~\cite{liberty2013simple} proposed an algorithm called \s{Frequent Directions} (or \s{FD}), further analyzed by Ghashami and Phillips~\cite{ghashami2014relative}, and then together jointly with Woodruff~\cite{FD-journal}. 
The \RR step sets each ${\sigma'}_j = \sqrt{\sigma_j^2 - \delta_i}$ where $\delta_i = \sigma_{\ell}^2$.  

Liberty also presented a faster variant \s{FastFD}, that instead sets $\delta_i = \sigma_{\ell/2}^2$ (the $(\ell/2)$th squared singular value of $B_{[i]}$) and updates new singular values to ${\sigma'}_j = \max \{0, \sqrt{\sigma_j^2 - \delta_i}\}$, hence ensuring at most half of the rows are all zeros after each such step. This reduces the runtime from $O(nd\ell^2)$ to $O(nd\ell)$ at expense of a sketch sometimes only using half of its rows.

\section{New Matrix Sketching Algorithms}
Here we describe our new variants on \FD that perform better in practice and are backed with error guarantees. In addition, we explain a couple of new matrix sketching techniques that makes subtle but tangible improvements to the other state-of-the-art algorithms mentioned above.

\subsection{New Variants On \fd}
\label{sec:new-iter}

\begin{table}[h!!!]
\begin{center}
\begin{tabular}{|r||c|c||c|c|c|}
\hline
\textbf{} & $\ell$ & \s{cov-err} & $\ell$ & \s{proj-err} & \s{runtime}
\\ \hline 
\s{Fast $\alpha$-FD} & $(k + 1/\eps)/\alpha$ & $\eps \frac{\|A - A_k\|_F^2}{\|A\|_F^2}$ & $k/(\eps\alpha)$ & $1+\eps$ & $n d \ell/\alpha$ 
\\ 
\s{SpaceSaving Directions} & $k + 1/\eps$ & $\eps \frac{\|A - A_k\|_F^2}{\|A\|_F^2}$ & $k/\eps$ & $1+\eps$ & $n d \ell$ 
\\ 
\s{Compensative FD} & $k + 1/\eps$ & $\eps \frac{\|A - A_k\|_F^2}{\|A\|_F^2}$ & $k/\eps$ & $1+\eps$ & $n d \ell$ 
\\ \hline 
\end{tabular} 
\end{center}
\vspace{-2mm}
\caption{\label{tbl:compare-new-iter} Theoretical Bounds for New Iterative Algorithms.}
\end{table}

Since all our proposed algorithms on \fd share the same structure, to avoid repeating the proof steps, we abstract out three properties that these algorithms follow and prove that \emph{any} algorithm with these properties satisfy the desired error bounds.  This slightly generalizes (allowing for $\alpha \neq 1$) a  recent framework~\cite{FD-journal}.  We prove these generalizations in Appendix \ref{app:aFD}.  

Consider any algorithm that takes an input matrix $A \in \mathbb{R}^{n \times d}$ and outputs a matrix $B \in \mathbb{R}^{\ell \times d}$ which follows three proprties below (for some parameter $\alpha \in (0,1]$ and some value $\Delta > 0$):
\begin{itemize} \denselist
\item Property 1: For any unit vector $x$ we have $\|Ax\|^2 - \|Bx\|^2 \geq 0$.
\item Property 2: For any unit vector $x$ we have $\|Ax\|^2 - \|Bx\|^2 \leq \Delta$.
\item Property 3: $\|A\|_F^2 - \|B\|_F^2 \geq \alpha \Delta \ell$.  
\end{itemize}

\begin{lemma}
Any $B$ satisfying the above three properties satisfies 
\begin{align*}
0 \leq \|A^T A - B^T B\|_2 &\leq \frac{1}{\alpha \ell-k} \|A -A_k\|_F^2, 
\\ \text{and }\;\;\;\;\;\;\; 
\|A - \pi_{B_k}(A)\|_F^2 &\leq \frac{\alpha \ell}{\alpha \ell -k} \|A - A_k\|_F^2, 
\end{align*}
where $\pi_{B_k}(\cdot)$ represents the projection operator onto $B_k$, the top $k$ singular vectors of $B$.  
\end{lemma}

Thus setting $\ell = k+1/\eps$ achieves $\|A^T A - B^T B\|_2 \leq \eps \|A- A_k\|_F^2$, and setting $\ell = k + k/\eps$ achieves $\|A - \pi_{B_k}(A)\|_F^2 \leq (1+\eps) \|A - A_k\|_F^2$.  
\FD maintains an $\ell \times d$ matrix $B$ (i.e. using $O(\ell d)$ space), and it is shown~\cite{ghashami2014relative} that there exists a value $\Delta$ that \FD satisfies three above-mentioned properties with $\alpha=1$.

\paragraph{Parameterized \FD.}

Parameterized \FD uses the following subroutine (Algorithm \ref{alg:RRPFD}) to reduce the rank of the sketch; it zeros out row $\ell$. This method has an extra parameter $\alpha \in [0,1]$ that  describes the fraction of singular values which will get affected in the \RR subroutine.  Note \iSVD has $\alpha =0$ and \FD has $\alpha=1$.  The intuition is that the smaller singular values are more likely associated with noise terms and the larger ones with signals, so we should avoid altering the signal terms in the \RR step.  

  \begin{algorithm}
    \caption{\label{alg:RRPFD} $\RRPFD(S,\alpha)$}
    \begin{algorithmic}
   \STATE $\delta_i \leftarrow \sigma_\ell^2$
	\STATE \textbf{return}  $\diag(\sigma_1, \ldots, \sigma_{\ell(1-\alpha)}, \sqrt{\sigma_{\ell(1-\alpha)+1}^2 - \delta_i}, \ldots, \sqrt{\sigma_{\ell}^2 - \delta_i})$  
   \end{algorithmic}
  \end{algorithm}

Here we show error bounds asymptotically matching \FD for $\alpha$-\FD (for constant $\alpha >0$), by showing the three Properties hold.  
We use $\Delta = \sum_{i=1}^n \delta_i$.

\begin{lemma}\label{Lemma1:alpha_fd}
For any unit vector $x$ and any $\alpha \geq 0$: $0 \leq \|C_{[i]} x\|^2 - \|B_{[i]} x\|^2 \leq \delta_i$.
\end{lemma}
\begin{proof} The right hand side is shown by just expanding $\|C_{[i]} x\|^2 - \|B_{[i]} x\|^2$.  
\begin{align*}
\|C_{[i]} x\|^2 - \|B_{[i]} x\|^2 &= \sum_{j=1}^\ell \sigma_j^2 \langle v_j, x\rangle^2 - \sum_{j=1}^\ell {\sigma'}_j^2 \langle v_j, x\rangle^2 
 =\sum_{j=1}^\ell (\sigma_j^2 - {\sigma'}_{j}^2) \langle v_j, x\rangle^2 \\
& = \delta_i \sum_{j=(1-\alpha)\ell + 1}^\ell \langle v_j, x\rangle^2 \leq \delta_i \|x\|^2 = \delta_i
\end{align*}
To see the left side of the inequality $\delta_i \sum_{j=(1-\alpha)\ell + 1}^\ell \langle v_j, x\rangle^2 \geq 0$.
\end{proof}

Then summing over all steps of the algorithm (using $\|a_i x\|^2 = \|C_{[i]} x\|^2 - \|B_{[i-1]} x\|^2$) it follows (see Lemma 2.3 in \cite{ghashami2014relative}) that
\[
0 \leq \|A x\|^2 - \|B x\|^2 \leq \sum_{i=1}^n \delta_i = \Delta,
\]
proving Property 1 and Property 2 about $\alpha$-\FD for any $\alpha \in [0,1]$.

\begin{lemma}\label{Lemma3:alpha_fd}
For any $\alpha \in (0,1]$, $\|A\|_F^2 - \|B\|_F^2 = \alpha \Delta \ell$, proving Property 3.
\end{lemma}
\begin{proof}
We expand that $\|C_{[i]}\|_F^2 = \sum_{j=1}^\ell \sigma_j^2$ to get
\begin{align*}
\|C_{[i]}\|^2_F & = \sum_{j=1}^{(1-\alpha)\ell} \sigma_j^2 + \sum_{j=(1-\alpha)\ell + 1}^{\ell} \sigma_j^2
\\& 
= \sum_{j=1}^{(1-\alpha)\ell} {\sigma'}_j^2 + \sum_{j=(1-\alpha)\ell + 1}^{\ell} ({\sigma'}_j^2 + \delta_i) = \|B_{[i]}\|^2_F + \alpha \ell \delta_i.  
\end{align*}
By using $\|a_i\|^2 = \|C_{[i]}\|_F^2 - \|B_{[i-1]}\|_F^2 = (\|B_{[i]}\|_F^2 + \alpha \ell \delta_i) - \|B_{[i-1]}\|_F^2$, and summing over $i$ we get
\[
\|A\|^2_F = \sum_{i=1}^n \|a_i\|^2 = \sum_{i=1}^n \|B_{[i]}\|^2_F - \|B_{[i-1]}\|_F^2 + \alpha \ell \delta_i = \|B\|^2_F + \alpha \ell \Delta. 
\]
Subtracting $\|B\|_F^2$ from both sides, completes the proof.
\end{proof}

The combination of the three Properties, provides the following results.

\begin{theorem}
Given an input matrix $A \in \mathbb{R}^{n \times d}$, $\alpha$-\FD with parameter $\ell$  returns a sketch $B \in \mathbb{R}^{\ell \times d}$ that satisfies for all $k > \alpha \ell$
\[
0 \leq \|Ax\|^2 - \|Bx\|^2 \leq \|A- A_k\|_F^2/(\alpha \ell -k)
\]
and projection of $A$ onto $B_k$, the top $k$ rows of $B$ satisfies
\[
\|A - \pi_{B_k}(A)\|_F^2 \leq \frac{\alpha \ell}{\alpha \ell - k} \|A-A_k\|_F^2. 
\]
\end{theorem}

Setting $\ell = (k + 1/\eps)/\alpha$ yields 
$0 \leq \|Ax\|^2 - \|Bx\|^2 \leq \eps \|A- A_k\|_F^2$ 
and 
setting $\ell  = (k + k/\eps)/\alpha$ yields $\|A - \pi_{B_k}(A)\|_F^2 \leq (1+\eps)\|A-A_k\|_F^2$.

\paragraph{Fast Parameterized \FD.}

\label{sec:new-Fast-aFD}

\s{Fast Parameterized \FD}(or \s{Fast $\alpha$-\FD}) improves the runtime performance of parameterized \FD in the same way \s{Fast FD} improves the performance of \FD.  More specifically, in \RR we set $\delta_i$ as the $(\ell - \ell \alpha/2)$th squared singular value, i.e. $\delta_i = \sigma_t^2$ for $t = \ell - \ell\alpha/2$.  Then we update the sketch by only changing the last $\alpha \ell$ singular values: we set ${\sigma'}_j^2 = \max(\sigma_j^2 - \delta_i, 0)$.  This sets at least $\alpha \ell /2$ singular values to $0$ once every $\alpha \ell/2$ steps. Thus the algorithm takes total time $O(nd + n/(\alpha \ell/2) \cdot d \ell^2) = O(nd \ell / \alpha)$.  

It is easy to see that \s{Fast $\alpha$-\FD} inherits the same worst case bounds as \s{$\alpha$-FD} on \s{cov-err} and \s{proj-err}, if we use twice as many rows.  
That is, setting $\ell = 2(k+1/\eps)/\alpha$ yields $\|A^TA - B^TB\|_2 \leq \eps \|A - A_k\|_F^2$ and 
setting $\ell = 2(k+k/\eps)/\alpha$ yields $\|A - \pi_{B_k}(A)\|_F^2 \leq (1+\eps)\|A-A_k\|_F^2$.  
In experiments we consider \s{Fast $0.2$-FD}.

\paragraph{SpaceSaving Directions.}

Motivated by an empirical study~\cite{cormode2008finding} showing that the SpaceSaving algorithm~\cite{metwally2006integrated} tends to outperform its analog Misra-Gries~\cite{mg-fre-82} in practice, we design an algorithm called \ssd (abbreviated \SSD) to try to extend these ideas to matrix sketching.  
It uses Algorithm \ref{alg:SS} for \RR.  
Like the SS algorithm for frequent items, it assigns the counts for the second smallest counter (in this case squared singular value $\sigma_{\ell-1}^2$) to the direction of the smallest.  
Unlike the SS algorithm, we do not use $\sigma_{\ell-1}^2$ as the squared norm along each direction orthogonal to $B$, as that gives a consistent over-estimate.

\begin{algorithm}[h!!!]
  \caption{\label{alg:SS} $\RRSS(S)$}
  \begin{algorithmic}
    \STATE $\delta_i \leftarrow \sigma_{\ell-1}^2$
    \STATE \textbf{return} $\diag(\sigma_{1}, \ldots, \sigma_{\ell-2}, 0, \sqrt{\sigma_{\ell}^2 + \delta_i})$.  
 \end{algorithmic}
\end{algorithm}
  
We can also show similar error bounds for \SSD.  It shows that a simple transformation of the output sketch $B \leftarrow \SSD(A)$ satisfies the three Properties, although $B$ itself does not.  We defer these proofs to Appendix \ref{app:SSD}, just stating the main bounds here.

\begin{theorem}\label{thm:SS}
After obtaining a matrix $B$ from \SSD on a matrix $A$ with parameter $\ell$, the following properties hold:
\begin{itemize} \vspace{-.1in}
\item[$\bullet$] $\|A\|_F^2 = \|B\|_F^2$.
\item[$\bullet$] for any unit vector $x$ and for $k < \ell/2-1/2$, we have $| \|A x\|^2 - \|B x \|^2 | \leq \|A - A_k\|_F^2 / (\ell/2 - 1/2 - k)$.
\item[$\bullet$] for $k < \ell/2-1$ we have $\|A - \pi_B^k(A)\|_F^2 \leq \|A - A_k\|_F^2 (\ell-1) / (\ell - 1 - 2k)$.  
\end{itemize}
\end{theorem}

Setting $\ell = 2k + 2/\eps + 1$ yields 
$0 \leq \|Ax\|^2 - \|Bx\|^2 \leq \eps \|A- A_k\|_F^2$ 
and 
setting $\ell  = 2k + 1 + 2k/\eps$ yields $\|A - \pi_{B_k}(A)\|_F^2 \leq (1+\eps)\|A-A_k\|_F^2$.  

\paragraph{Compensative Frequent Directions.}
Inspired by the isomorphic transformation~\cite{agarwal2012mergeable} between the Misra-Gries~\cite{mg-fre-82} and the SpaceSaving sketch~\cite{metwally2006integrated},  which performed better in practice on the frequent items problem~\cite{cormode2008finding}, we consider another variant of \FD for matrix sketching.  In the frequent items problem, this would return an identical result to \SSD, but in the matrix setting it does not.  

We call this approach \cfd (abbreviated \CFD).  
In original \FD, the computed sketch $B$ underestimates the Frobenius norm of stream\cite{ghashami2014relative}, in \CFD we try to compensate for this.  
Specifically, we keep track of the total mass $\Delta = \sum_{i=1}^n \delta_i$ subtracted from squared singular values (this requires only an extra counter).  Then we slightly modify the \FD algorithm.  In the final step where 
$B = S' V^T$, we modify $S'$ to $\hat S$ by setting each singular value $\hat \sigma_j = \sqrt{{\sigma'}_j^2 + \Delta}$, then we instead return $B = \hat S V^T$.  

It now follows that for any $k \leq \ell$, including $k=0$, that $\|A\|_F^2 = \|B\|_F^2$, 
that for any unit vector $x$ we have $| \|A x\|_F^2 - \|B x\|_F^2 | \leq \Delta \leq \|A - A_k\|_F^2 / (\ell -k)$ for any $k < \ell$, and since $V$ is unchanged that
$\|A - \pi_B^k(A) \|_F^2 \leq \|A - A_k\|_F^2 \cdot \ell / (\ell-k)$.  
Also as in \FD, setting $\ell = k + 1/\eps$ yields 
$0 \leq \|Ax\|^2 - \|Bx\|^2 \leq \eps \|A- A_k\|_F^2$ 
and 
setting $\ell  = k/\eps$ yields $\|A - \pi_{B_k}(A)\|_F^2 \leq (1+\eps)\|A-A_k\|_F^2$.

\subsection{New Without Replacement Sampling Algorithms}
\label{sec:new-SwoR}
\begin{table}[h!!!]
\begin{center}
\begin{tabular}{|r||c|c||c|c|c|}
\hline
\textbf{} & $\ell$ & \s{cov-err} & $\ell$ & \s{proj-err} & \s{runtime}
\\ \hline 
\s{Priority} & $d/\eps^2$ & $\eps$ & $\ell$ & - & $\s{nnz}(A) \log \ell$ 
\\  
\s{VarOpt} & $d/\eps^2$ & $\eps$ & $\ell$ & - & $\s{nnz}(A) \log \ell$ 
\\  \hline
\end{tabular} 
\end{center}
\vspace{-2mm}
\caption{\label{tbl:compare-new-samp} Theoretical Bounds for New Sampling Algorithms.}
\end{table}

As mentioned above, most sampling algorithms use \emph{sampling with replacement} (SwR) of rows.  This is likely because, in contrast to \emph{sampling without replacement} (SwoR), it is easy to analyze and for weighted samples conceptually easy to compute.  SwoR for unweighted data can easily be done with variants of reservoir sampling~\cite{vitter1985random}; however, variants for weighted data have been much less resolved until recently~\cite{duffield2007priority,cohen2009stream}.  
\paragraph{Priority Sampling.}
A simple technique~\cite{duffield2007priority} for SwoR on weighted elements first assigns each element $i$ a random number $u_i \in \s{Unif}(0,1)$.  This implies a priority $\rho_i = w_i / u_i$, based on its weight $w_i$ (which for matrix rows $w_i = \|a\|_i^2$).  We then simply retain the $\ell$ rows with largest priorities, using a priority queue of size $\ell$.  Thus each step takes $O(\log \ell)$ time, but on randomly ordered data would take only $O(1)$ time in expectation since elements with $\rho_i \leq \tau$, where $\tau$ is the $\ell$th largest priority seen so far, are discarded.  

Retained rows are given a squared norm $\hat w_i = \max(w_i, \tau)$.  Rows with $w_i \geq \tau$ are always retained with original norm.  Small weighted rows are kept proportional to their squared norms.   
The technique, \s{Priority Sampling}, is simple to implement, but requires a second pass on retained rows to assign final weights.  

\paragraph{VarOpt Sampling.}

VarOpt (or Variance Optimal) sampling~\cite{cohen2009stream} is a modification of priority sampling that takes more care in selecting the threshold $\tau$.  In priority sampling, $\tau$ is generated so $\E[\sum_{a_i \in B} \hat w_i] = \|A\|_F^2$, but if $\tau$ is set more carefully, then we can achieve $\sum_{a_i \in B} \hat w_i = \|A\|_F^2$ deterministically.  
VarOpt selects each row with some probability $p_i = \min(1, w_i / \tau )$, with $\hat w_i = \max(w_i, \tau)$, and so exactly $\ell$ rows are selected.  

The above implies that for a set $L$ of $\ell$ rows maintained, there is a fixed threshold $\tau$ that creates the equality.  We maintain this value $\tau$ as well as the $t$ weights smaller than $\tau$ inductively in $L$.  If we have seen at least $\ell +1$ items in the stream, there must be at least one weight less than $\tau$.  On seeing a new item, we use the stored priorities $\rho_i = w_i/u_i$ for each item in $L$ to either (a) discard the new item, or (b) keep it and drop another item from the reservoir.  As the priorities increase, the threshold $\tau$ must always increase.   It takes amortized constant time to discarding a new item or $O(\log \ell)$ time to keep the new item, and does not require a final pass on $L$.  
We refer to it as \s{VarOpt}.

A similar algorithm using priority sampling was considered in a distributed streaming setting~\cite{GLP14}, which provided a high probability bound on \s{cov-err}.  A constant probability of failure bound for $\ell = O(d/\eps^2)$ and $\s{cov-err} \leq \eps$, follows with minor modification from Section \ref{app:rs-convert}.  
It is an open question to bound the projection error for these algorithms, but we conjecture the bounds will match those of \s{Norm Sampling}.

\section{Experimental Setup}
\label{sec:setup}
We used an OpenSUSE 12.3 machine with 32 cores of Intel(\textsc{r}) Core(\textsc{tm}) i7-4770S CPU(3.10 GHz) and 32GB of RAM. 
Randomized algorithms were run five times; we report the median error value.  

\begin{table}[t!!!!]
\begin{center}
\begin{tabular}{|c||c|c||c|c|c|c|}
\hline
\textbf{DataSet} & \textbf{\# datapoints} & \textbf{\# attributes} & \textbf{rank} & \textbf{numeric rank} & \textbf{nnz}\% & \textbf{excess kurtosis}\\
\hline 
\hline
 \textsf{Birds} & 11789 & 312 & 312 & 12.50 & 100 & 1.72\\ 
\hline 
\textsf{Random Noisy} & 10000 & 500 & 500 & 14.93 & 100 & 0.95\\
\hline
\textsf{CIFAR-10} & 60000 & 3072 & 3072 & 1.19 & 99.75 & 1.34\\ 
\hline 
\textsf{Connectus} & 394792 & 512 & 512 & 4.83 & 0.0055 & 17.60\\ 
\hline 
\textsf{Spam} & 9324 & 499 & 499 & 3.25 & 0.07 & 3.79\\ 
\hline 
\textsf{Adversarial} & 10000 & 500 & 500 & 1.69 & 100 & 5.80\\ 
\hline 
\end{tabular} 
\end{center}
\vspace{-.25in}
\caption{\label{tbl:datasets} Dataset Statistics.}
\end{table}

\begin{figure}[t!]
\begin{centering}
\subfigure[\s{\footnotesize Birds}]{\includegraphics[width=.16\linewidth]{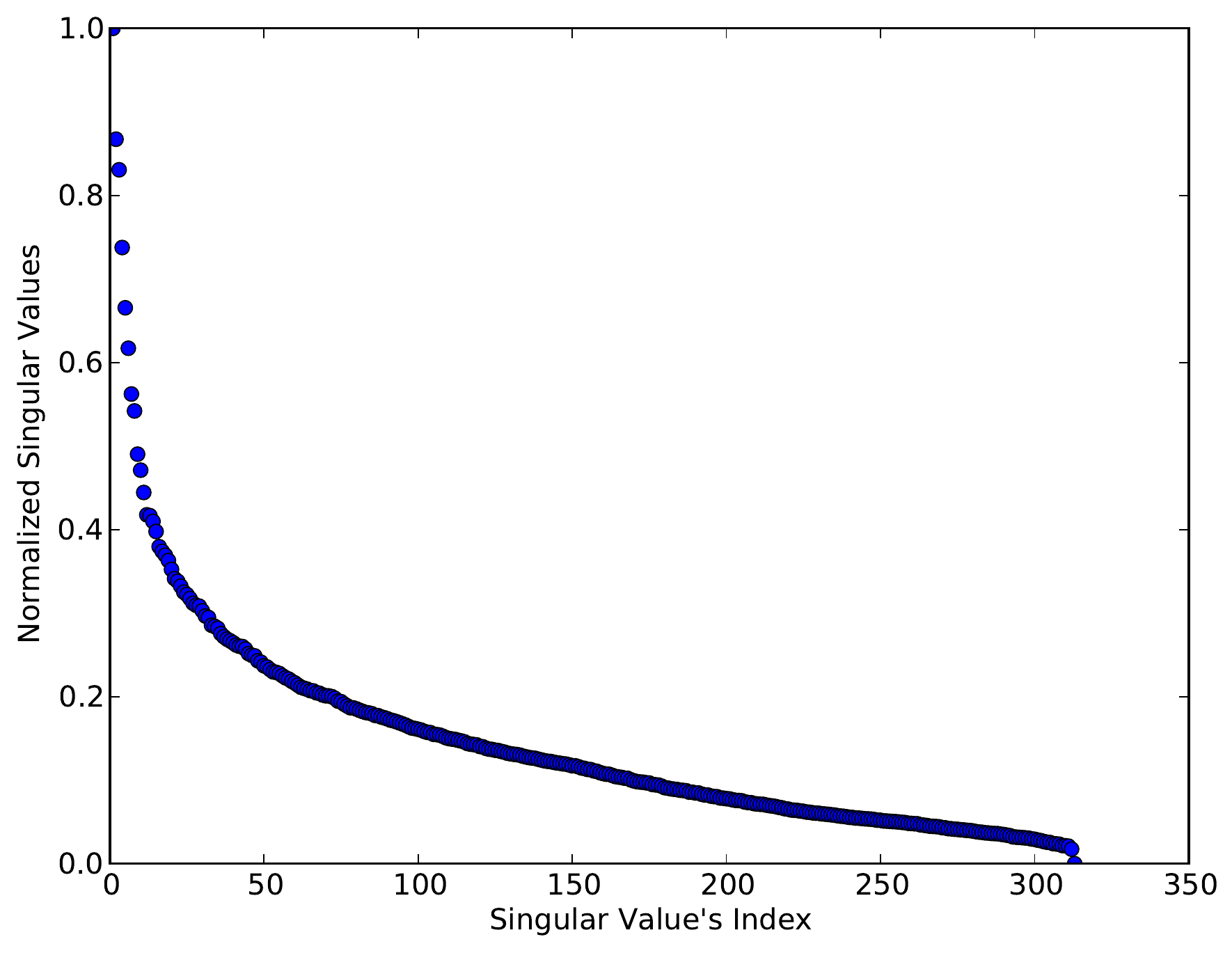}}
\subfigure[\s{\footnotesize Random Noisy}]{\includegraphics[width=.16\linewidth]{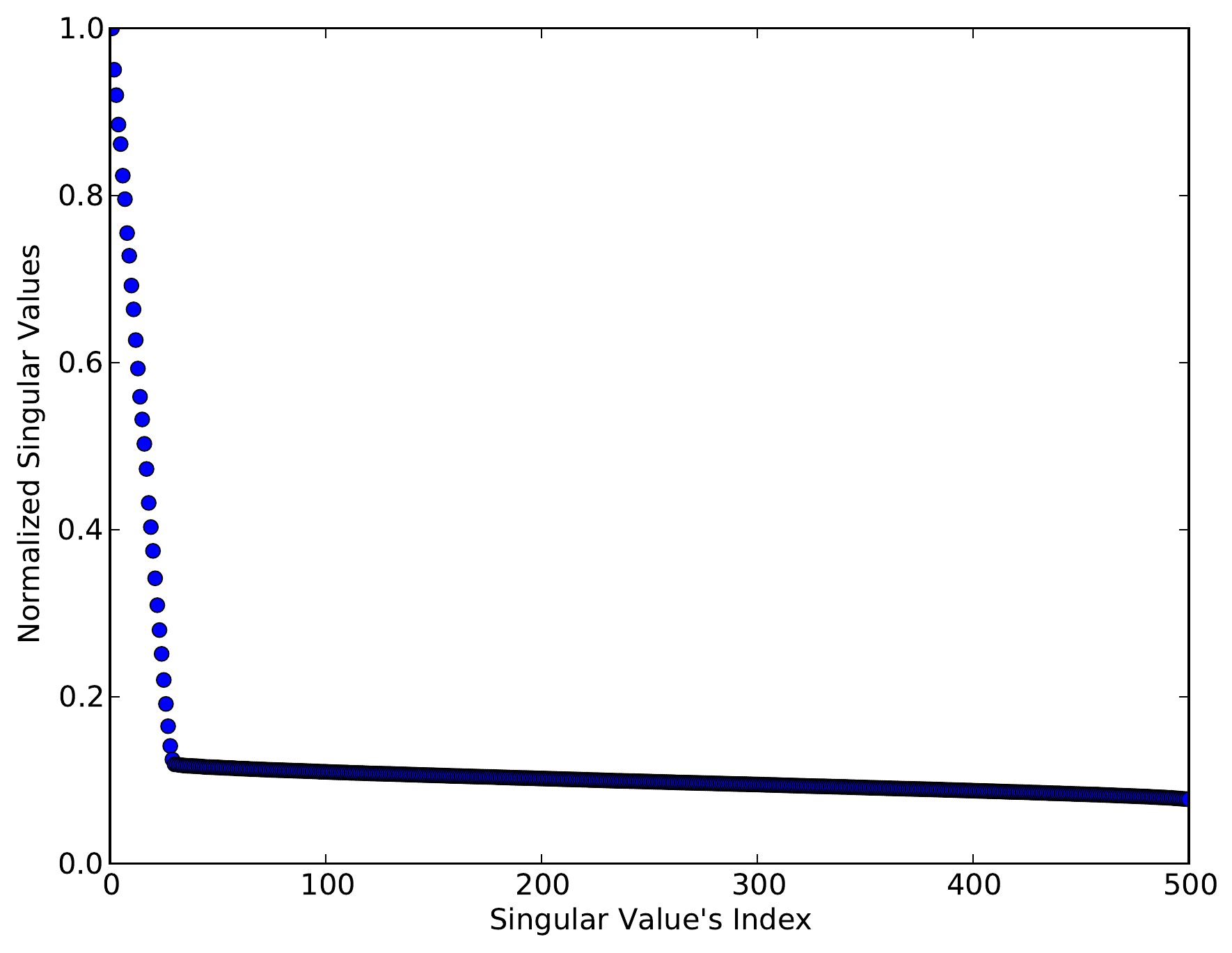}}
\subfigure[\s{\footnotesize CIFAR-10}]{\includegraphics[width=.16\linewidth]{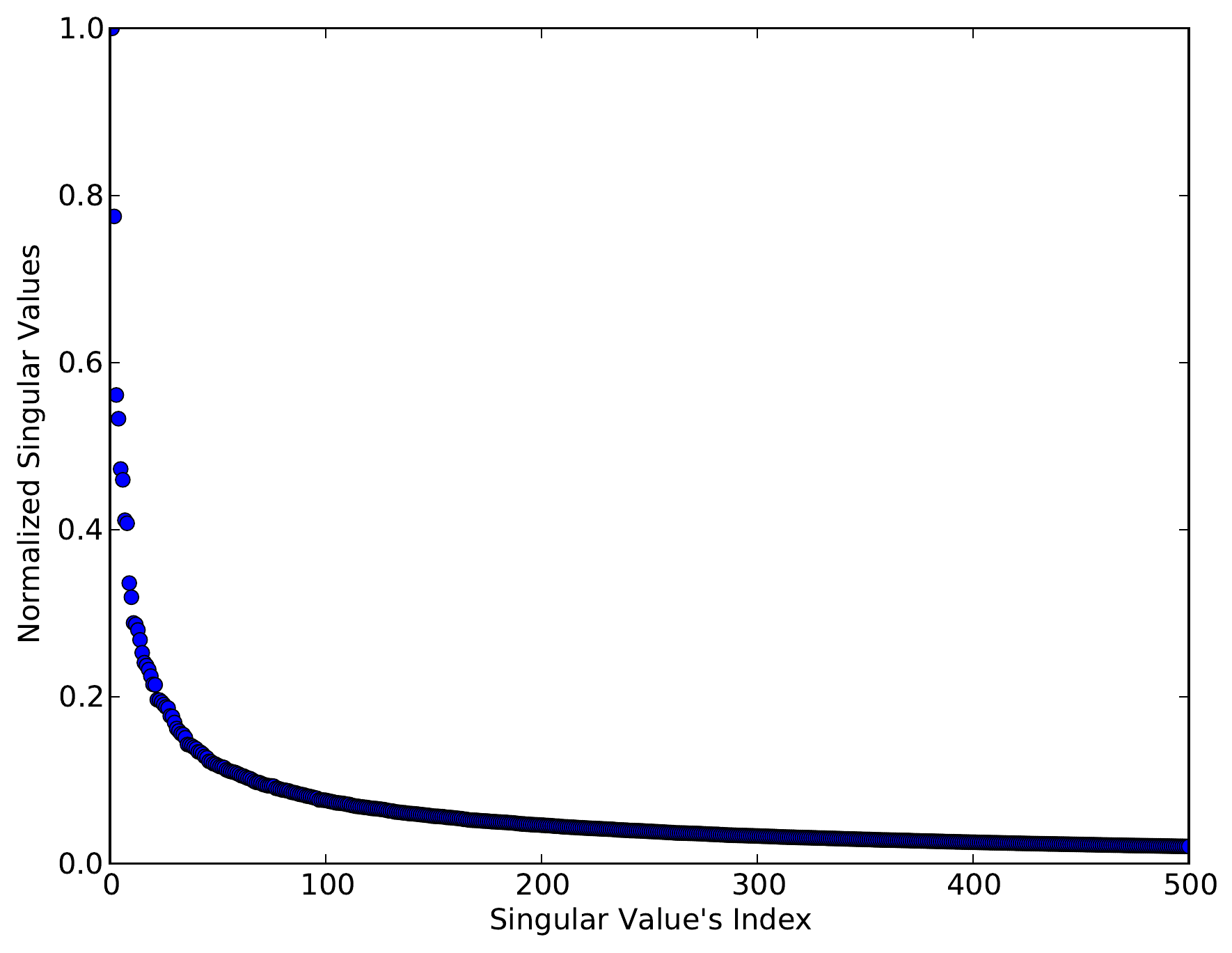}}
\subfigure[\s{\footnotesize ConnectUS}]{\includegraphics[width=.16\linewidth]{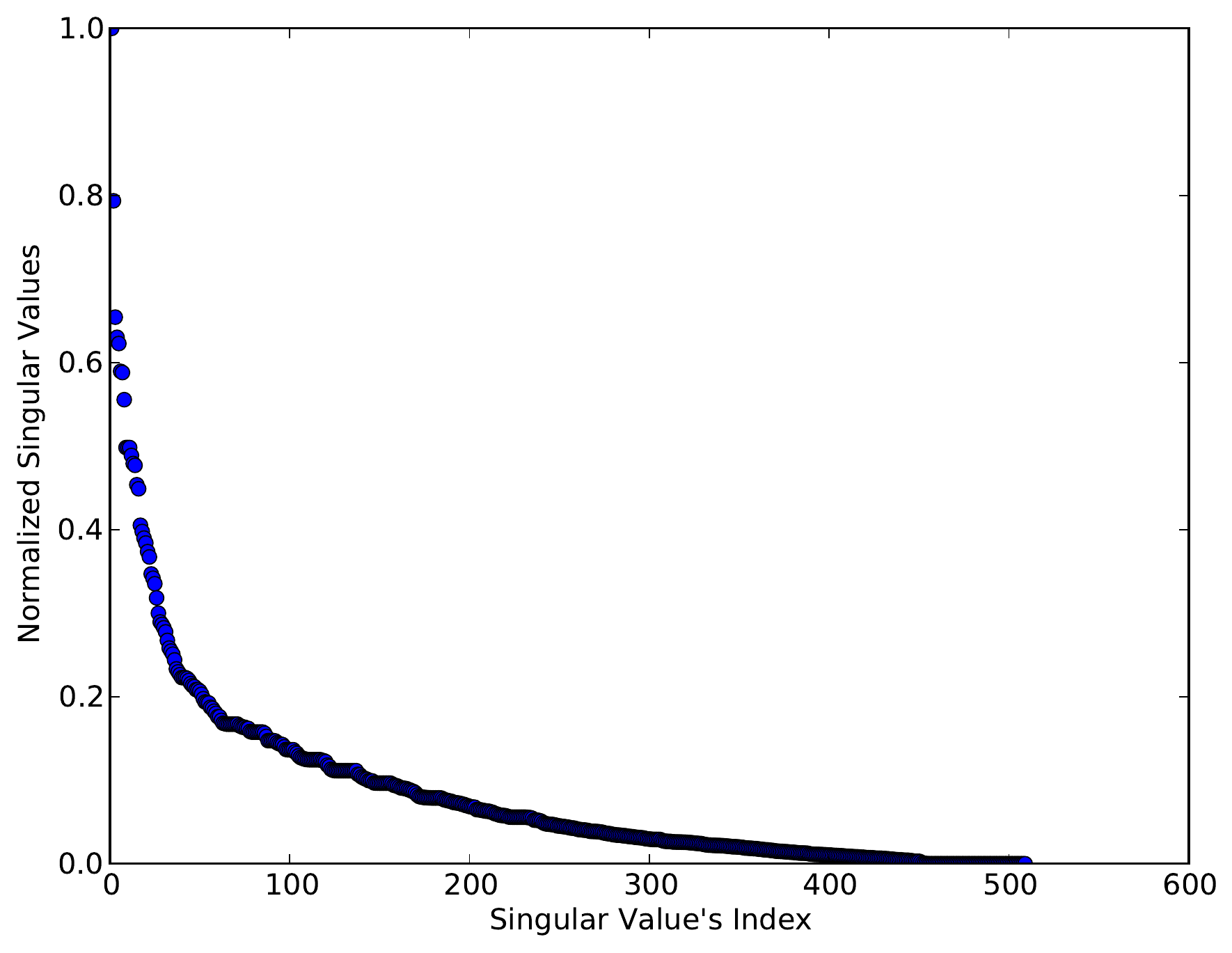}}
\subfigure[\s{\footnotesize Spam}]{\includegraphics[width=.16\linewidth]{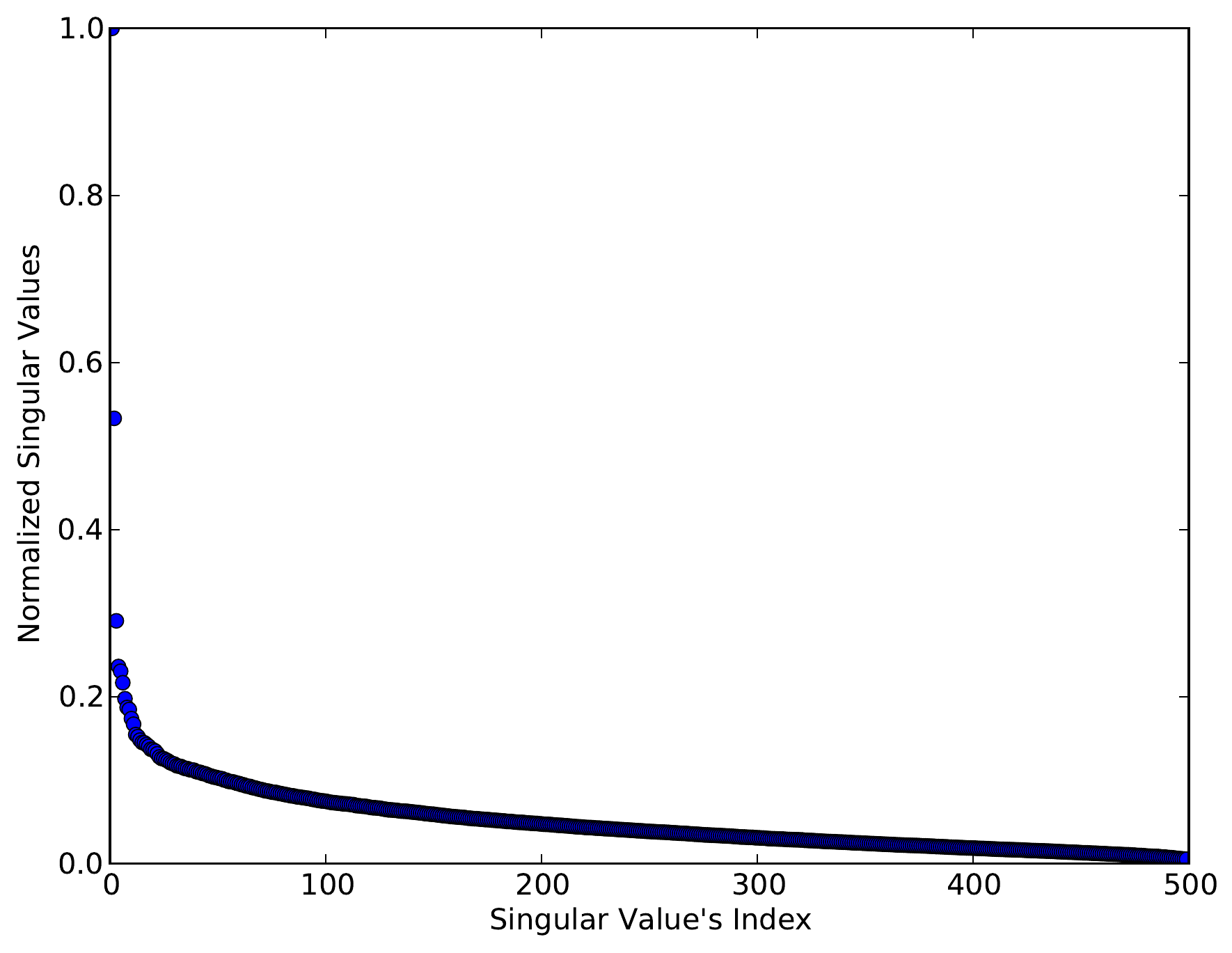}}
\subfigure[\s{\footnotesize Adversarial}]{\includegraphics[width=.16\linewidth]{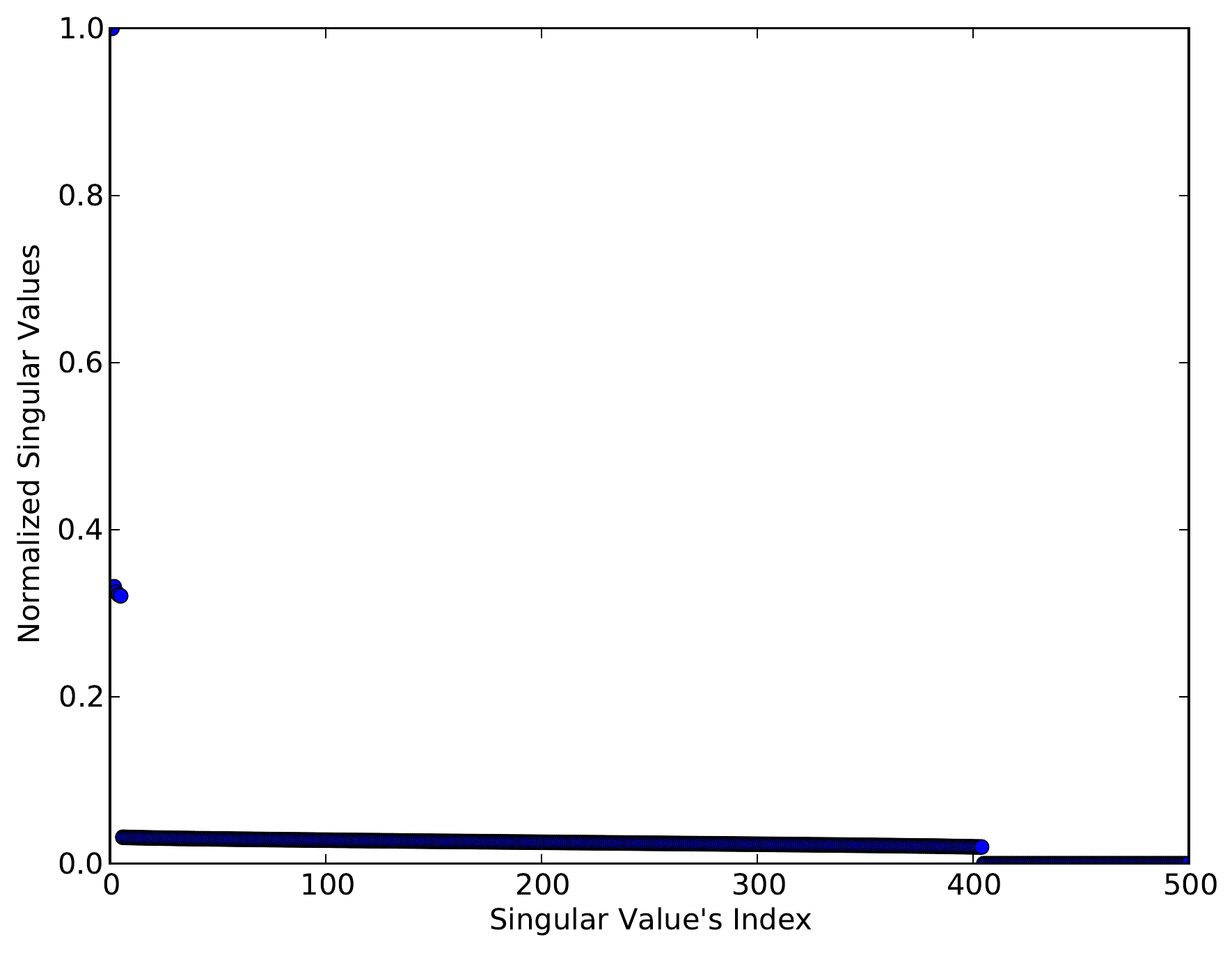}}
\end{centering}
\caption{\label{fig:singular value distribution}
Singular values distribution for datasets in Table \ref{tbl:datasets}.  The $x$-axis is singular value index, and the $y$-axis shows normalized singular values such that the highest singular value is one, i.e., each  value divided by largest singular value of dataset }
\end{figure}

\paragraph{Datasets.}  
We compare performance of the algorithms on both synthetic and real datasets. In addition, we generate adversarial data to show that \iSVD performs poorly under specific circumstances, this explains why there is no theoretical guarantee for them.
Each data set is an $n\times d$ matrix $A$, and the $n$ rows are processed one-by-one in a stream.

Table \ref{tbl:datasets} lists all datasets with information about their $n$, $d$, $\s{rank}(A)$, numeric rank $\|A\|_F^2/\|A\|_2^2$, percentage of non-zeros (as $\s{nnz}\%$, measuring sparsity), and excess kurtosis.  
We follow Fisher's distribution with baseline kurtosis (from normal distribution) is $0$; positive excess kurtosis reflects fatter tails and negative excess kurtosis represents thinner tails.  

For \s{Random Noisy}, we generate the input $n \times d$ matrix $A$ synthetically, mimicking the approach by Liberty~\cite{liberty2013simple}.  We compose $A = SDU + F/\zeta$, where $SDU$ is the $m$-dimensional signal (for $m<d$) and $F/\zeta$ is the (full) $d$-dimensional noise with $\zeta$ controlling the signal to noise ratio.  Each entry $F_{i,j}$ of $F$  is generated i.i.d. from a normal distribution $N(0,1)$, and we set $\zeta=10$.    
For the signal, $S \in \bb{R}^{n \times m}$ again we generate each $S_{i,j} \sim N(0,1)$ i.i.d;  $D$ is diagonal with entries $D_{i,i} = 1-(i-1)/d$ linearly decreasing; and $U \in \bb{R}^{m \times d}$ is just a random rotation.  
We use $n=10000$, $d=500$, and consider $m \in \{10,20,30,50\}$ with $m=30$ as default.  

In order to create \s{Adversarial} data, we constructed two orthogonal subspaces $S_1 = \bb{R}^{m_1}$ and $S_2 = \bb{R}^{m_2}$  ($m_1 = 400$ and  $m_2 = 4$).  Then we picked two separate sets of random vectors $Y$ and $Z$ and projected them on $S_1$ and $S_2$, respectively. Normalizing the projected vectors and concatenating them gives us the input matrix $A$.  All vectors in $\pi_{S_1}(Y)$ appear in the stream before $\pi_{S_2}(Z)$; this represents a very sudden and orthogonal shift.  As the theorems predict, \FD and our proposed algorithms adjust to this change and properly compensate for it.  However, since $m_1 \geq \ell$, then \iSVD cannot adjust and always discards all new rows in $S_2$ since they always represent the smallest singular value of $B_{[i]}$.  

We consider 4 real-world datasets. 
\s{ConnectUS} is taken from the University of Florida Sparse Matrix collection~\cite{uflorida}.  \s{ConnectUS} represents a recommendation system.  Each column is a user, and each row is a webpage, tagged $1$ if favorable, $0$ otherwise. It contains 171 users that share no webpages preferences with any other users.  
\s{Birds}~\cite{birds_link} has each row represent an image of a bird, and each column a feature. PCA is a common first approach in analyzing this data, so we center the matrix.  
\s{Spam}~\cite{spam_link} has each row represent a spam message, and each column some feature; it has dramatic and abrupt feature drift over the stream, but not as much as \s{Adversarial}. 
\s{CIFAR-10} is a standard computer vision benchmark dataset for deep learning~\cite{krizhevsky2009learning}.

The singular values distribution of the datasets is given in Figure \ref{fig:singular value distribution}.  The $x$-axis is the singular value index, and the $y$-axis shows the normalized singular values, i.e. singular values divided by $\sigma_1$, where $\sigma_1$ is the largest singular value of dataset. \s{Birds}, \s{ConnectUS} and \s{Spam} have consistent drop-offs in singular values.  
\s{Random Noisy} has initial sharp and consistent drops in singular values, and then a more gradual decrease.  
The drop-offs in \s{CIFAR-10} and \s{Adversarial} are more dramatic.  

We will focus most of our experiments on three data sets 
\s{Birds} (dense, tall, large numeric rank), 
\s{Spam} (sparse, not tall, negative kurtosis, high numeric rank), and
\s{Random Noisy} (dense, tall, synthetic).  
However, for some distinctions between algorithms require considering much larger datasets; for these we use
\s{CIFAR-10} (dense, not as tall, small numeric rank) and
\s{ConnectUS} (sparse, tall, medium numeric rank).
Finally, \s{Adversarial} and, perhaps surprisingly \s{ConnectUS} are used to show that using \iSVD (which has no guarantees) does not always perform well.

\section{Experimental Evaluation}
\label{sec:eval}\label{sec:exp}
We divide our experimental evaluation into four sections:  The first three sections contain comparisons within algorithms of each group (sampling, projection, and iterative), while the fourth compares accuracy and run time of exemplar algorithm in each group against each other. 

We measure error for all algorithms as we change the parameter $\ell$ (\s{Sketch Size}) determining the number of rows in matrix $B$.  We measure covariance error as 
\s{err} $= \|A^T A - B^T B\|_2 / \|A\|_F^2$ (\s{Covariance Error}); this indicates for instance for \FD, that \s{err} should be at most $1/\ell$, but could be dramatically less if $\|A - A_k\|_F^2$ is much less than $\|A\|_F^2$ for some not so large $k$.  
We consider \s{proj-err} $= \|A - \pi_{B_k}(A)\|_F^2/\|A -A_k\|_F^2$, always using $k=10$ (\s{Projection Error}); for \FD we should have \s{proj-err} $\leq \ell/(\ell-10)$, and $\geq 1$ in general.  
We also measure run-time as sketch size varies.

Within each class, the algorithms are not dramatically different across sketch sizes.  But across classes, they vary in other ways, and so in the global comparison, we will also show plots comparing runtime to \s{cov-err} or \s{proj-err}, which will help demonstrate and compare these trade-offs.


\begin{figure}[t!]
\begin{centering}
{\tiny \textsf{Birds} \hspace{47mm} \textsf{Spam} \hspace{46mm} \textsf{Random Noisy}} 
\\
\includegraphics[width=.305\linewidth]{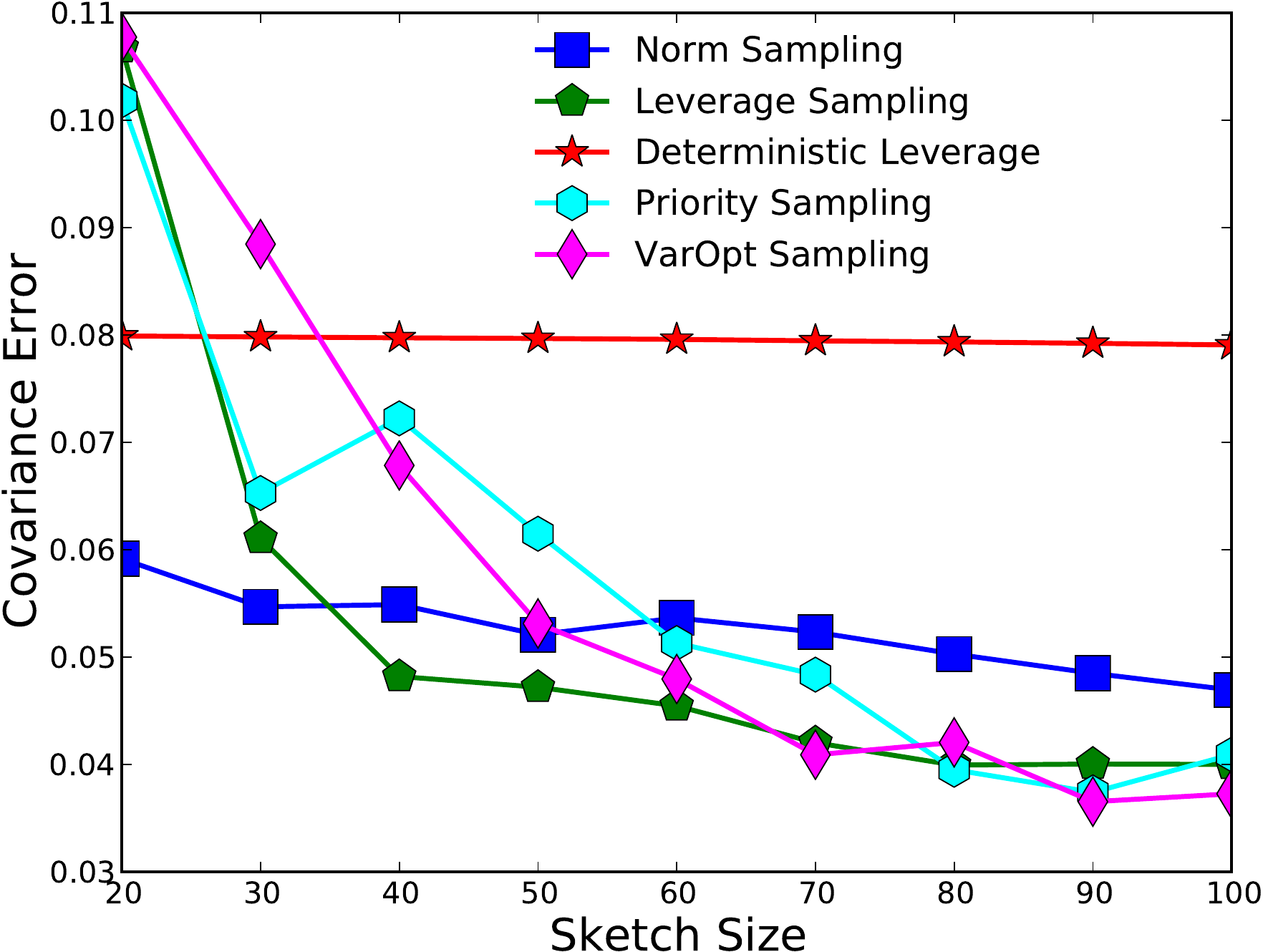}
\includegraphics[width=.345\linewidth]{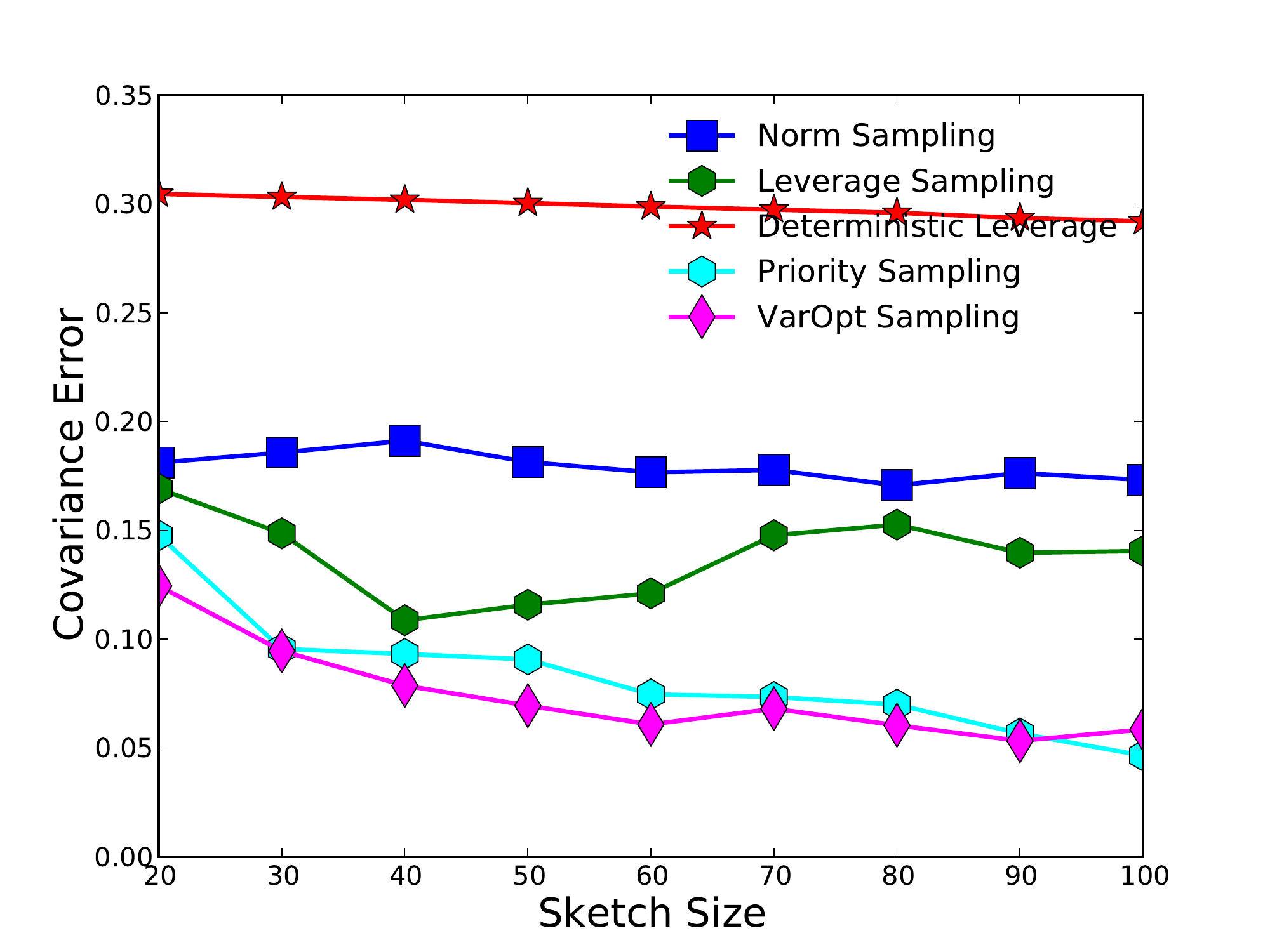}
\includegraphics[width=.305\linewidth]{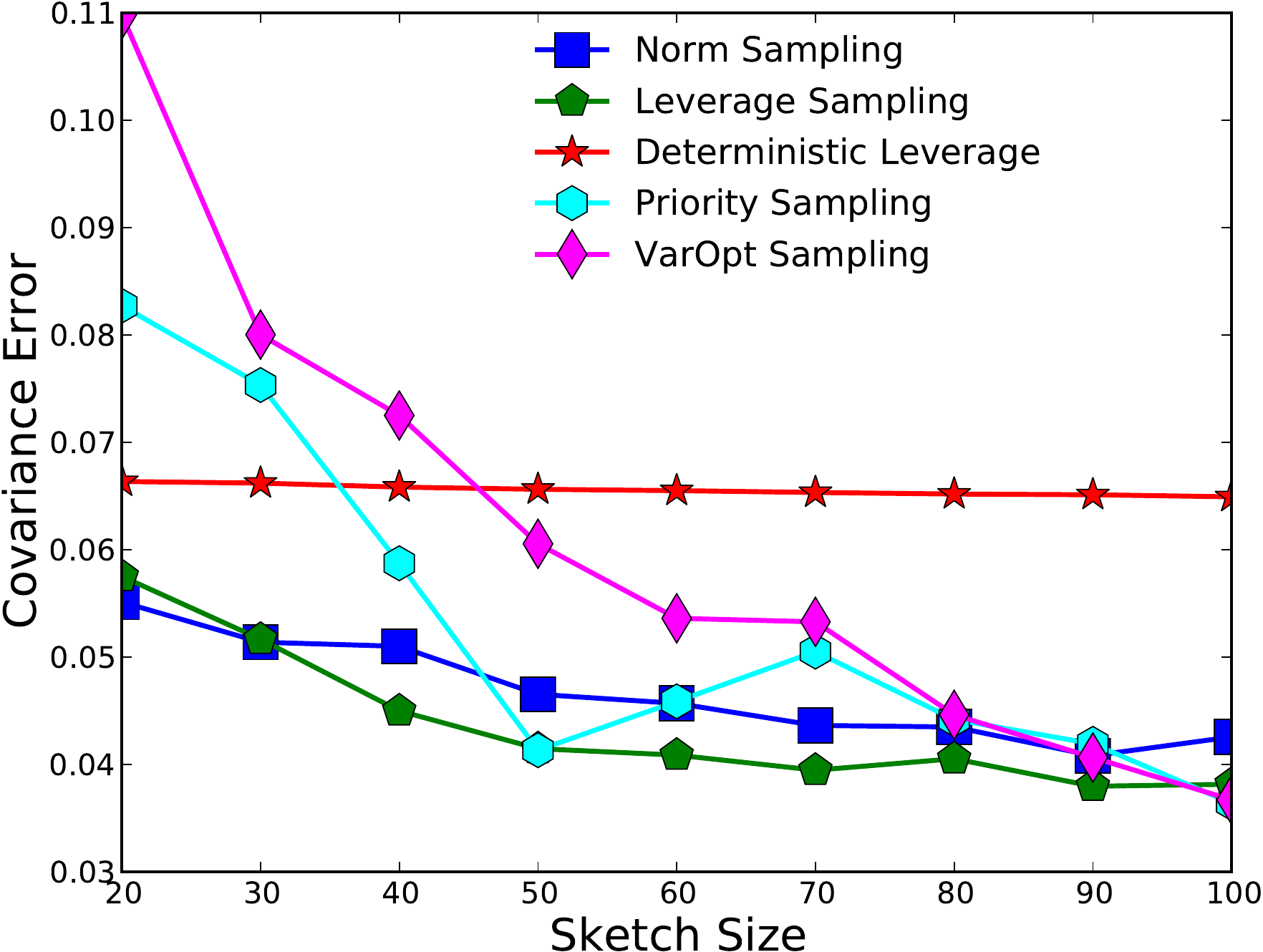}

\includegraphics[width=.305\linewidth]{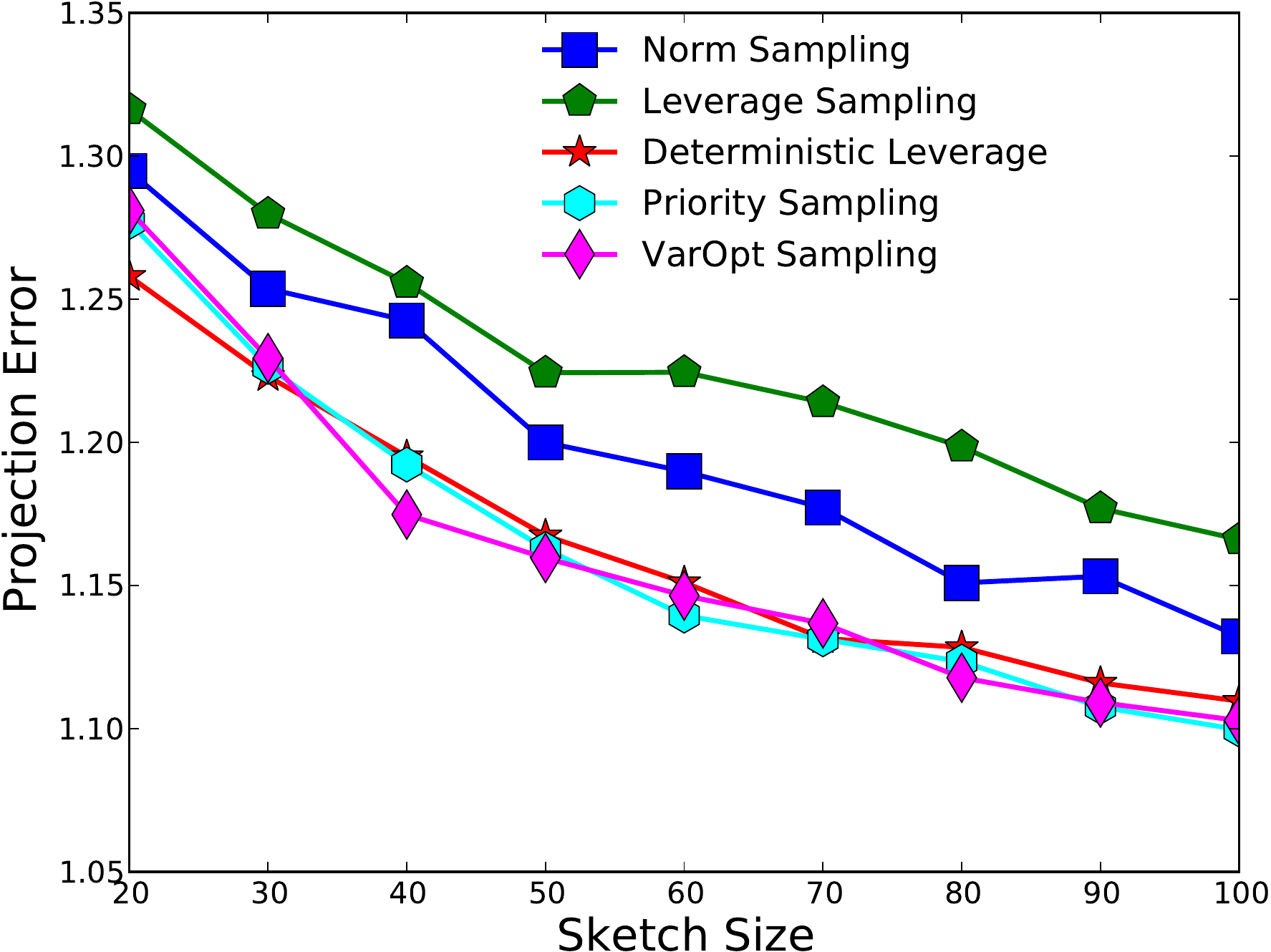}
\includegraphics[width=.345\linewidth]{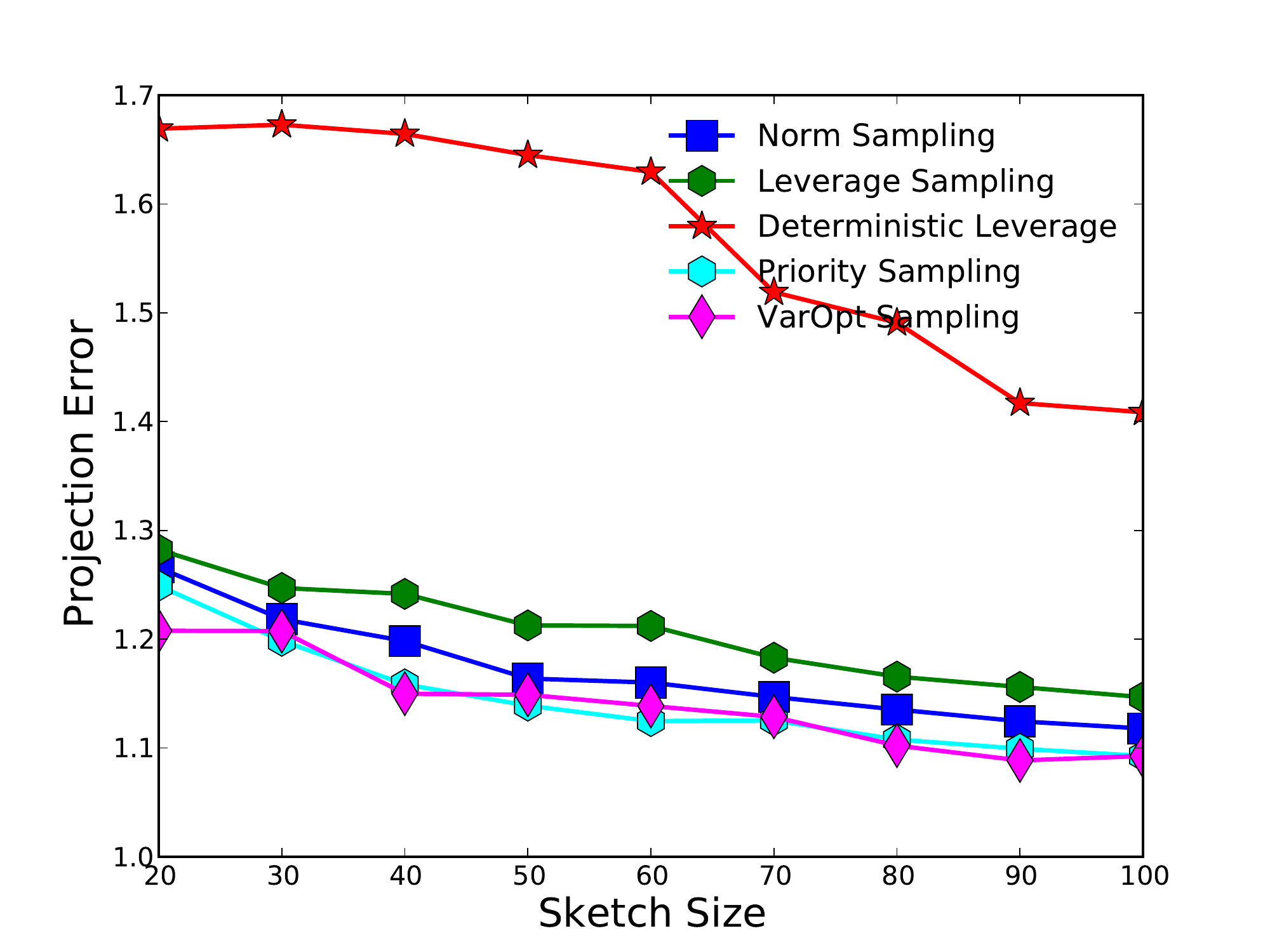}
\includegraphics[width=.305\linewidth]{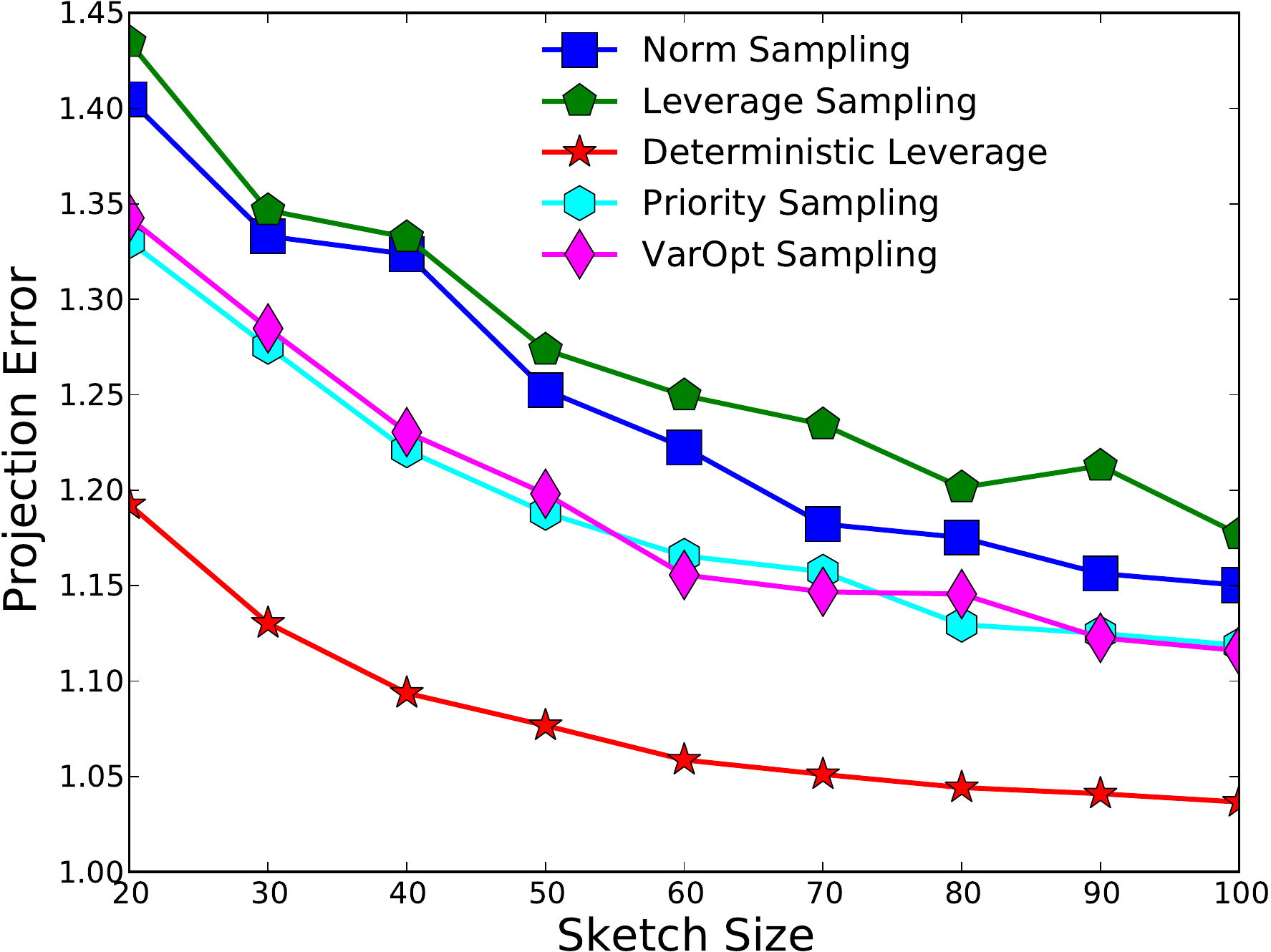}

\includegraphics[width=.305\linewidth]{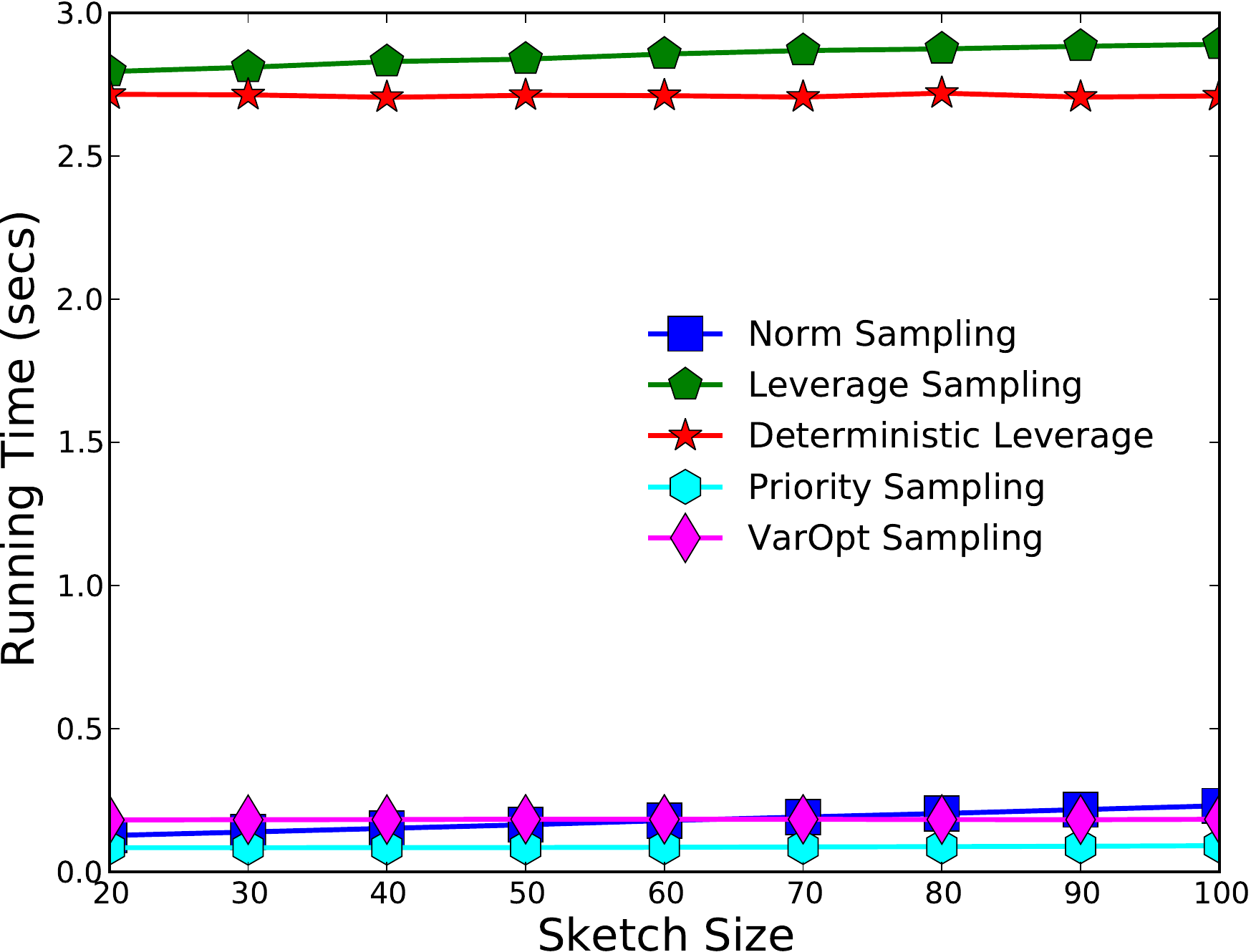}
\includegraphics[width=.345\linewidth]{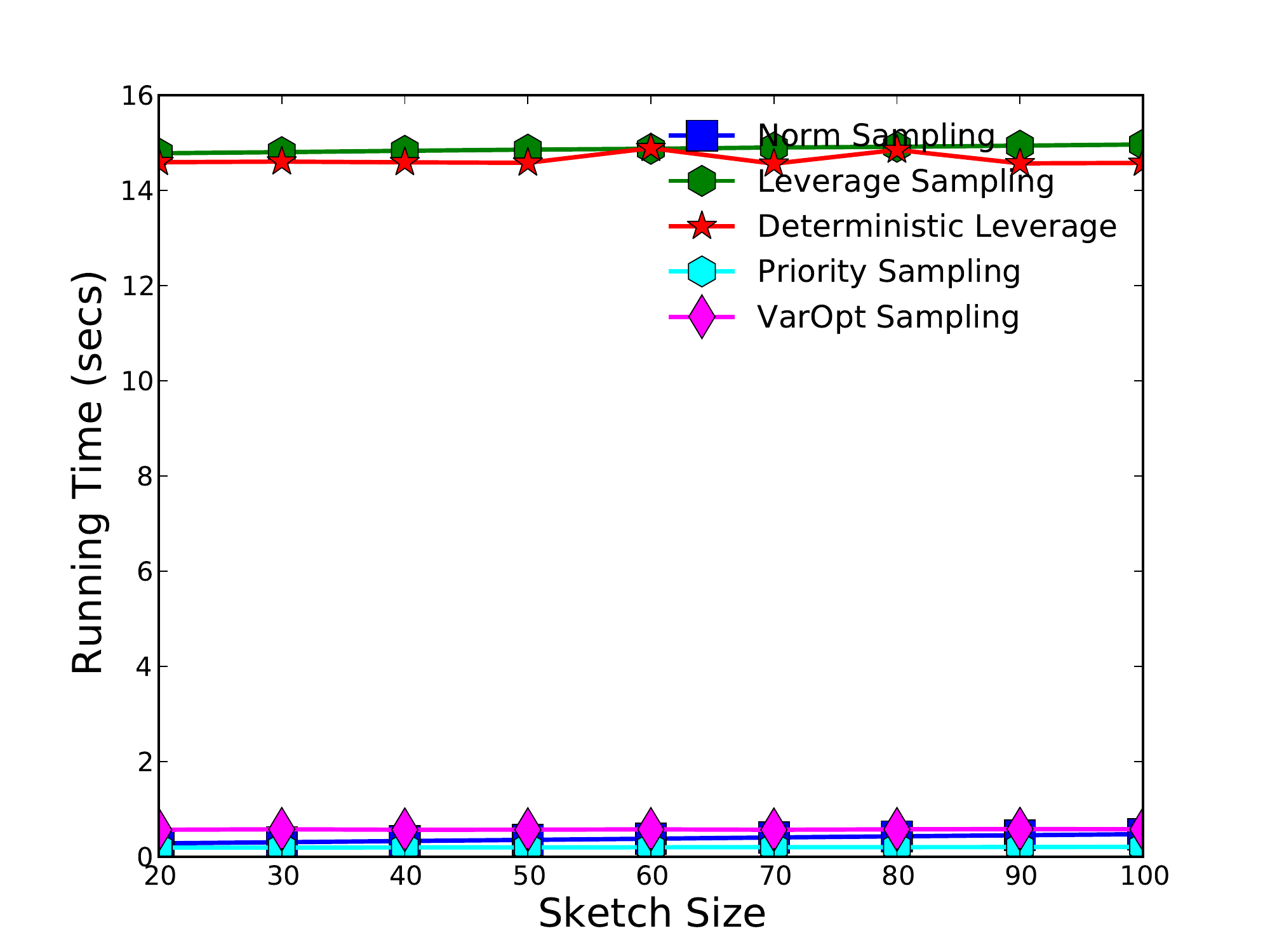}
\includegraphics[width=.305\linewidth]{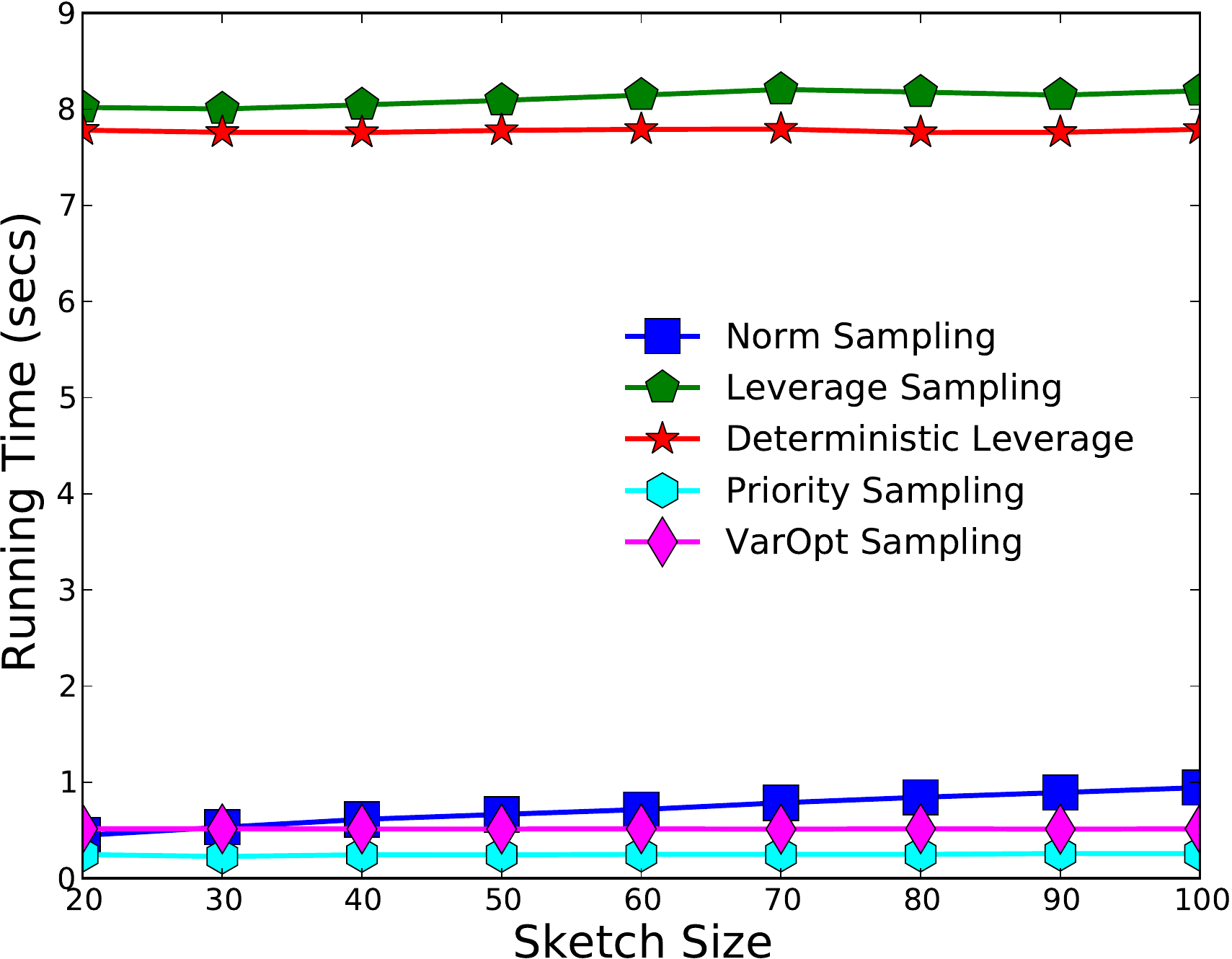}
\caption{\label{fig:samp-alg}Sampling algorithms on \s{Birds}(left), \s{Spam}(middle), and \s{Random Noisy}(30)(right).} 
\end{centering}
\end{figure}

\subsection{Sampling Algorithms}
\label{ssec:samp-eval}

Figure \ref{fig:samp-alg} shows the covariance error, projection error, and runtime for the sampling algorithms as a function of sketch size, run on the \s{Birds}, \s{Spam}, and \s{Random Noisy}(30) datasets with sketch sizes from $\ell=20$ to $100$.  We use parameter $k=10$ for \s{Leverage Sampling}, the same $k$ used to evaluate \s{proj-err}.

First note that \s{Deterministic Leverage} performs quite differently than all other algorithms.  
The error rates can be drastically different: smaller on 
\s{Random Noisy} \s{proj-err} and \s{Birds} \s{proj-err}, while higher on \s{Spam} \s{proj-err} and all \s{cov-err} plots.  The proven guarantees are only for matrices with Zipfian leverage score sequences and \s{proj-err}, and so when this does not hold it can perform worse. But when the conditions are right it outperforms the randomized algorithms since it deterministically chooses the best rows.

Otherwise, there is very small difference between the error performance of all randomized algorithms, within random variation. 
The small difference is perhaps surprising since \s{Leverage Sampling} has a stronger error guarantee, achieving a relative \s{proj-err} bound instead of an additive error of \s{Norm Sampling}, \s{Priority Sampling} and \s{VarOpt Sampling} which only use the row norms.  
Moreover \s{Leverage Sampling} and \s{Deterministic Leverage Sampling} are significantly slower than the other approaches since they require first computing the SVD and leverage scores.  We note that if $\|A - A_k\|_F^2 > c \|A\|_F^2$ for a large enough constant $c$, then for that choice of $k$, the tail is effectively fat, and thus not much is gained by the relative error bounds.  Moreover, \s{Leverage Sampling} bounds are only stronger than \s{Norm Sampling} in a variant of \s{proj-err} where $[\pi_B(A)]_k$ (with best rank $k$ applied \emph{after} projection) instead of $\pi_{B_k}(A)$, and \s{cov-err} bounds are only known (see Appendix \ref{app:rs-convert}) under some restrictions for \s{Leverage Sampling}, while unrestricted for the other randomized sampling algorithms.  

\subsection{Projection Algorithms}
\label{ssec:proj-eval}

Figure \ref{fig:proj-alg} plots the covariance and projection error, as well as the runtime for various sketch sizes of $20$ to $100$ for the projection algorithms.  

Otherwise, there were two clear classes of algorithms.  For the same sketch size, \s{Hashing} and \s{OSNAP} perform a bit worse on projection error (most clearly on \s{Noisy Random}), and roughly the same in covariance error, compared to \s{Random Projections} and \s{Fast JLT}.  Note that \s{Fast JLT} seems consistently better than others in \s{cov-err}, but we have chosen the best $q$ parameter (sampling rate) by trial and error, so this may give an unfair advantage.  Moreover, \s{Hashing} and \s{OSNAP} also have significantly faster runtime, especially as the sketch size grows.  While \s{Random Projections} and \s{Fast JLT} appear to grow in time roughly linearly with sketch size, \s{Hashing} and \s{OSNAP} are basically constant.  
Section \ref{ssec:global-eval} on larger datasets and sketch sizes, shows that if the size of the sketch is not as important as runtime, \s{Hashing} and \s{OSNAP} have the advantage.  

\begin{figure}[t!]
\begin{centering}
{\tiny \textsf{Birds} \hspace{47mm} \textsf{Spam} \hspace{44mm} \textsf{Random Noisy}} 
\\
\includegraphics[width=.305\linewidth]{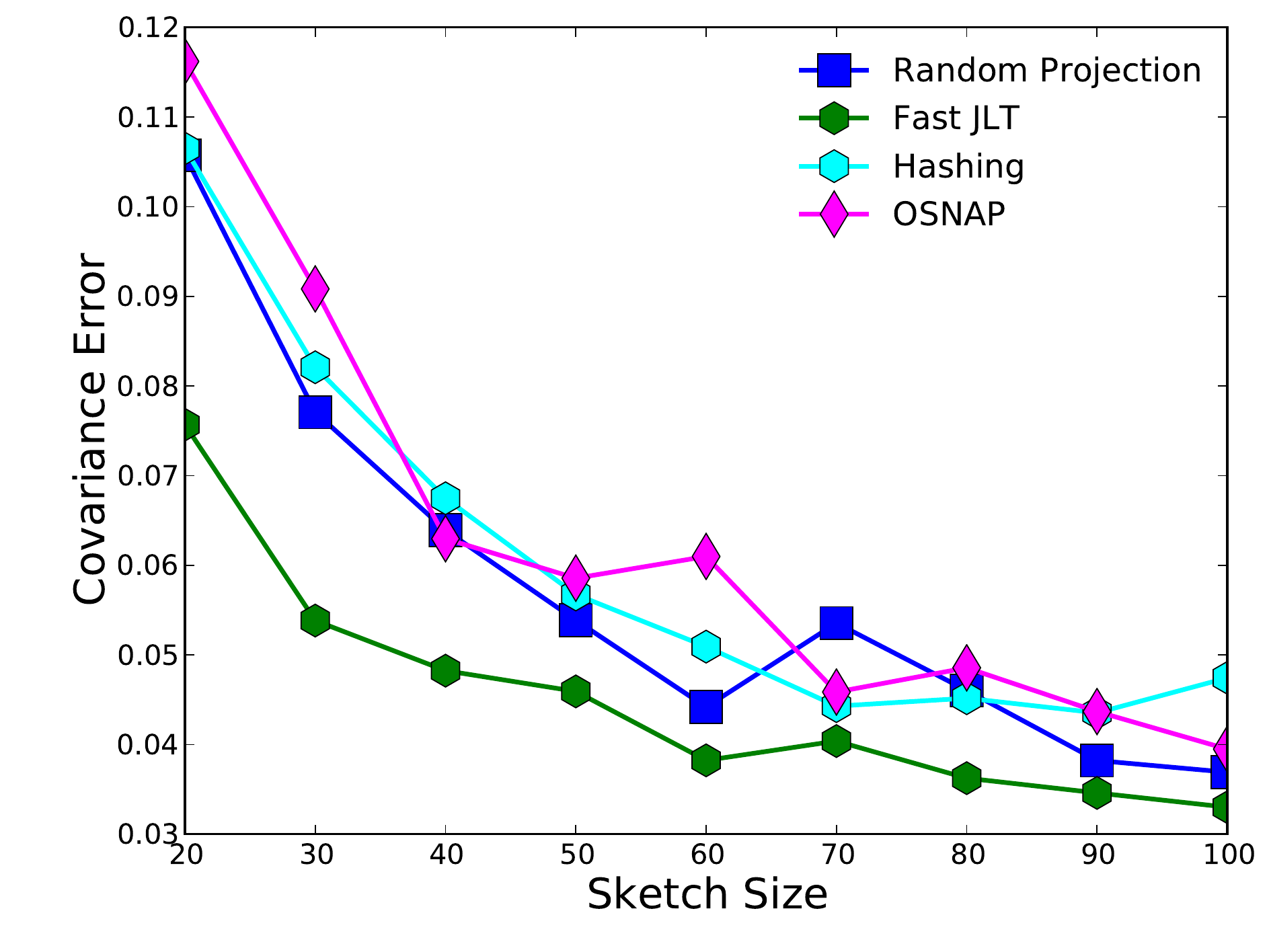}
\includegraphics[width=.305\linewidth]{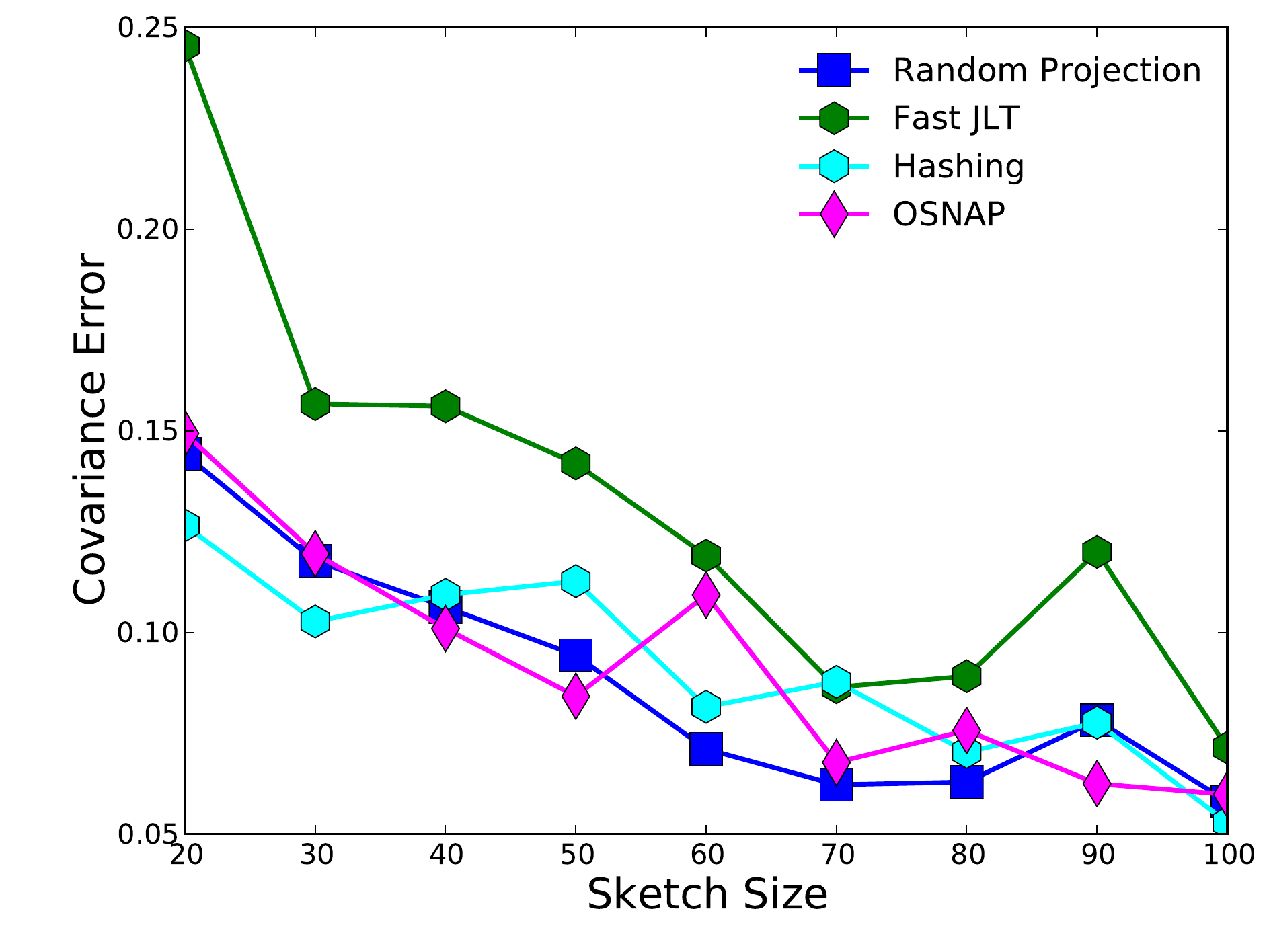}
\includegraphics[width=.305\linewidth]{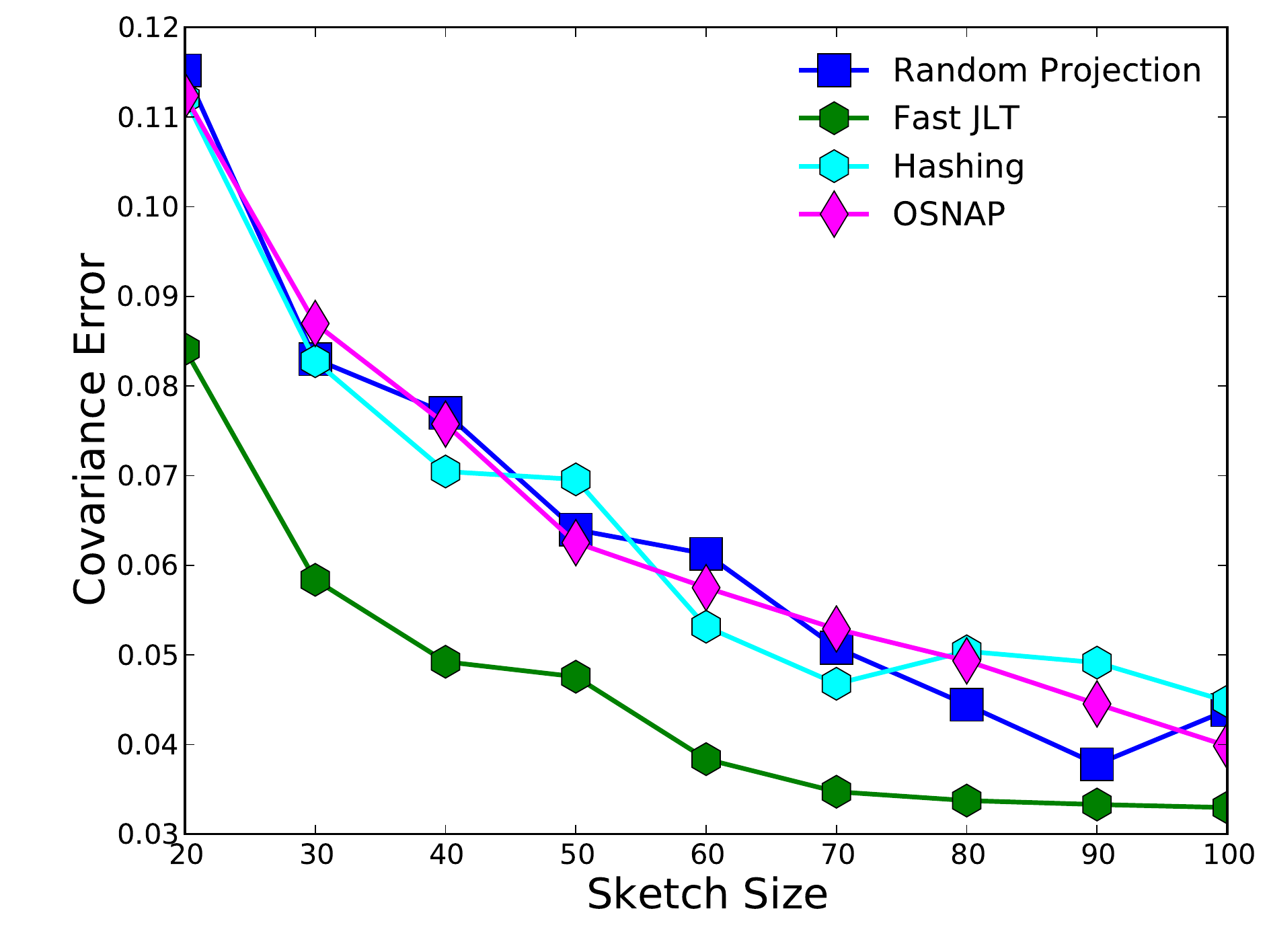}

\includegraphics[width=.305\linewidth]{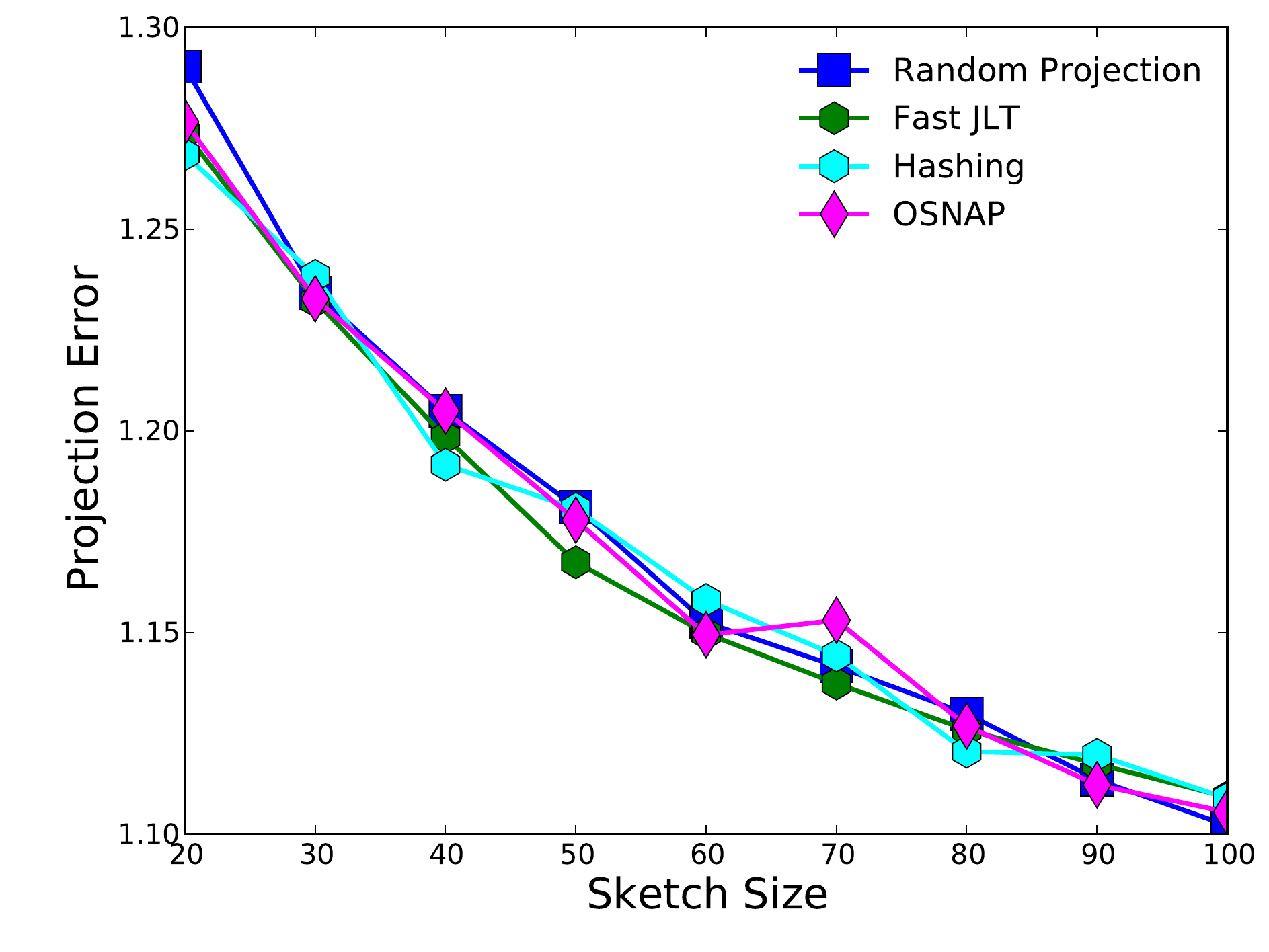}
\includegraphics[width=.305\linewidth]{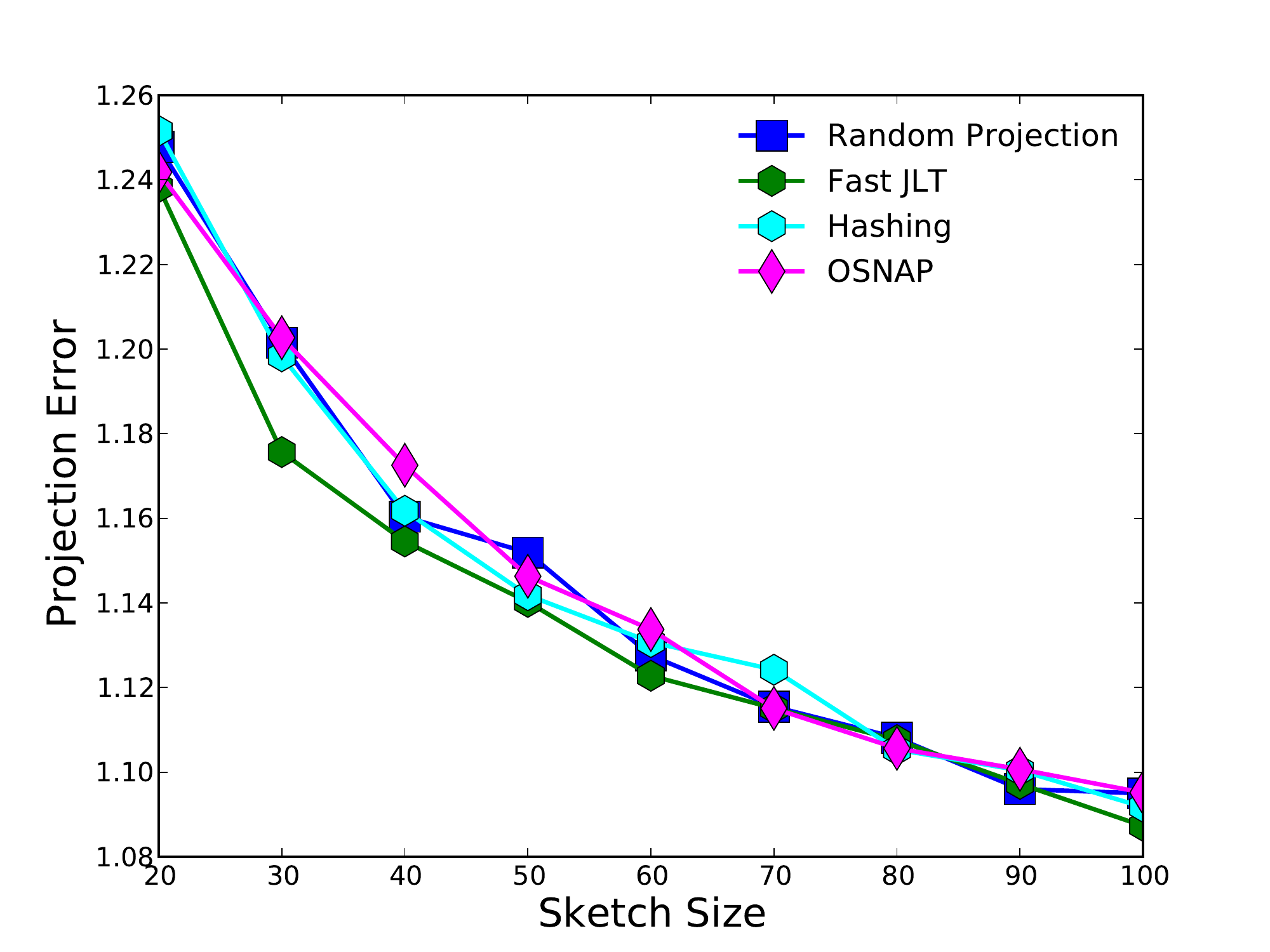}
\includegraphics[width=.305\linewidth]{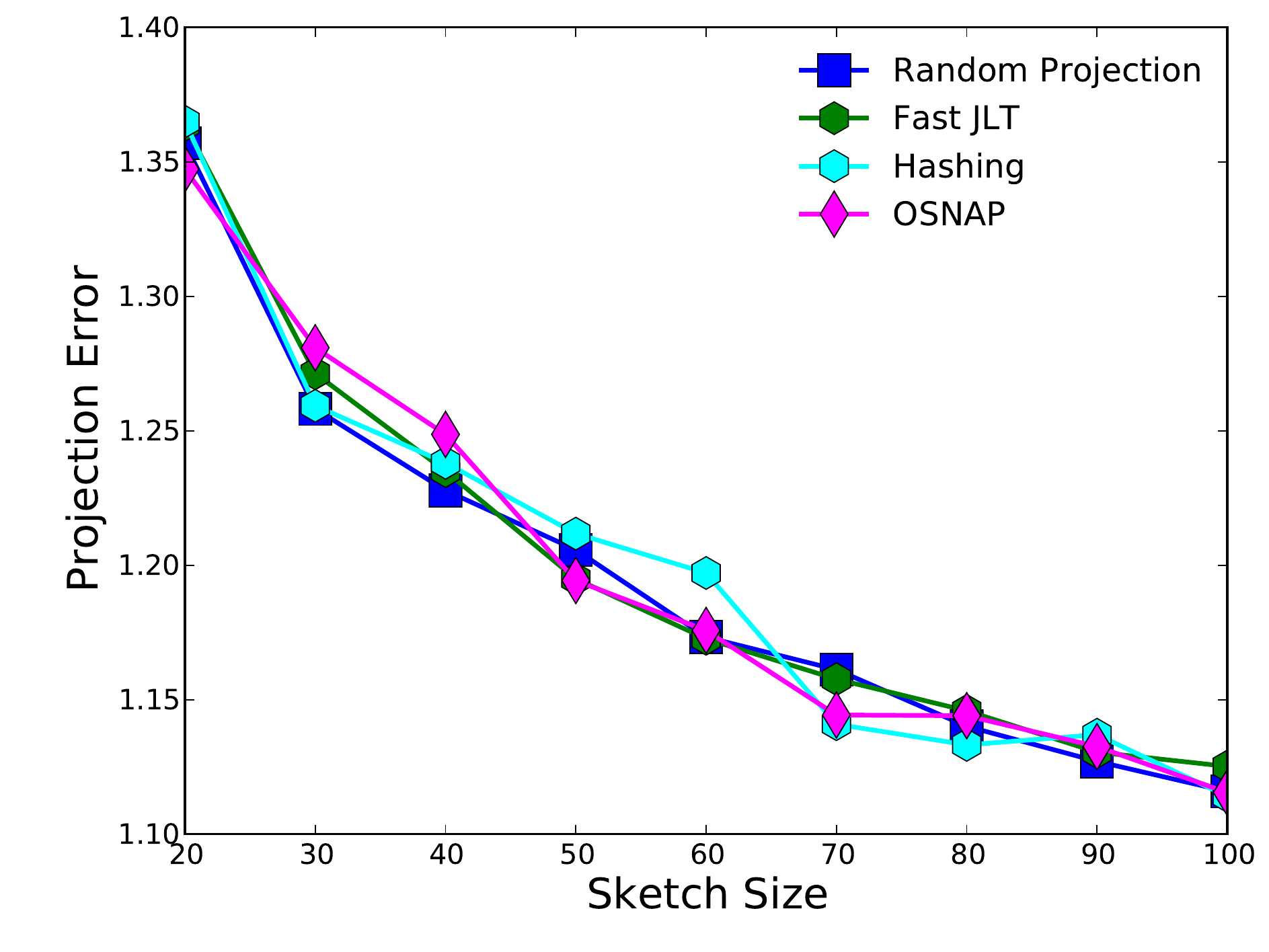}

\includegraphics[width=.305\linewidth]{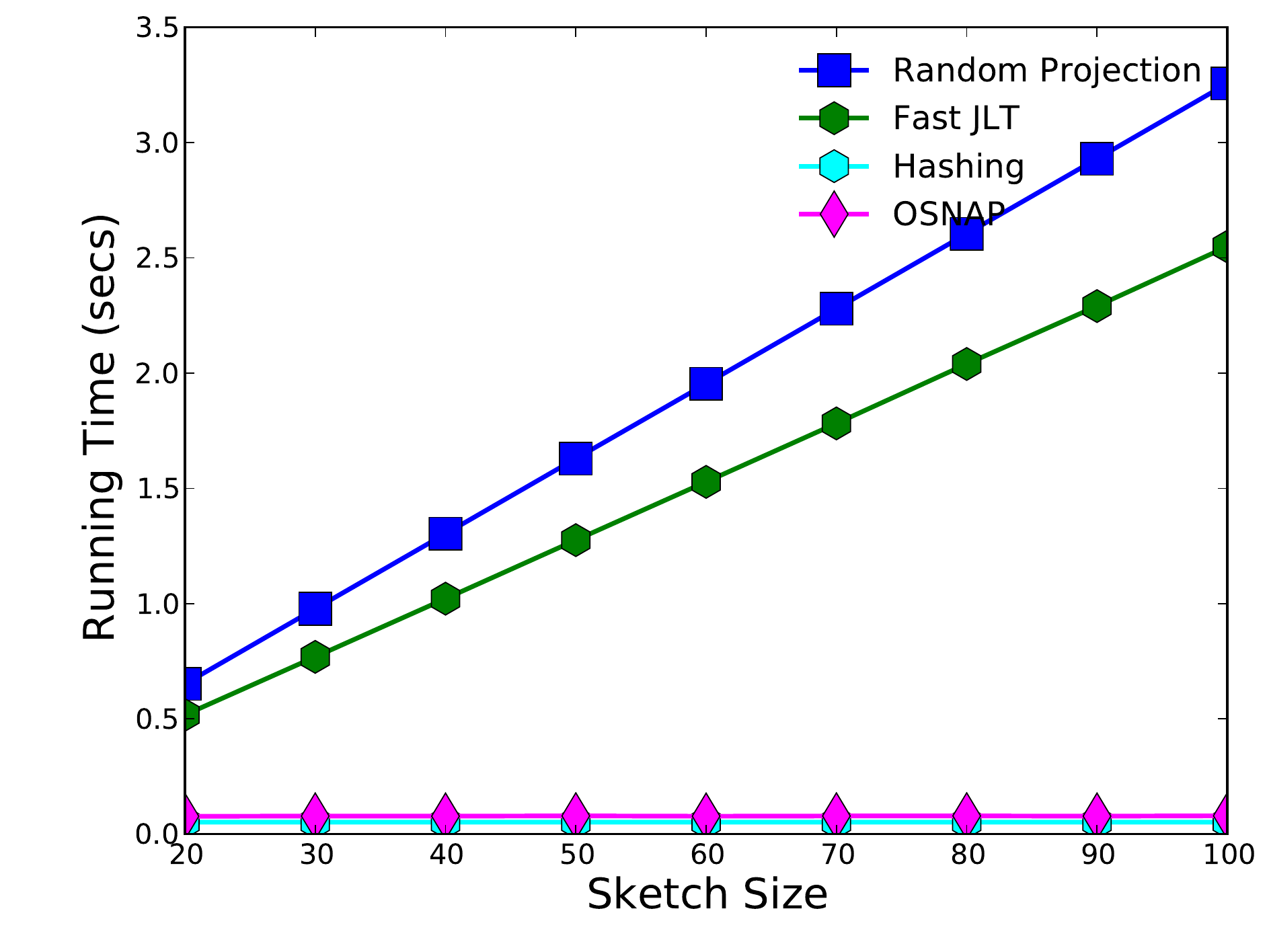}
\includegraphics[width=.305\linewidth]{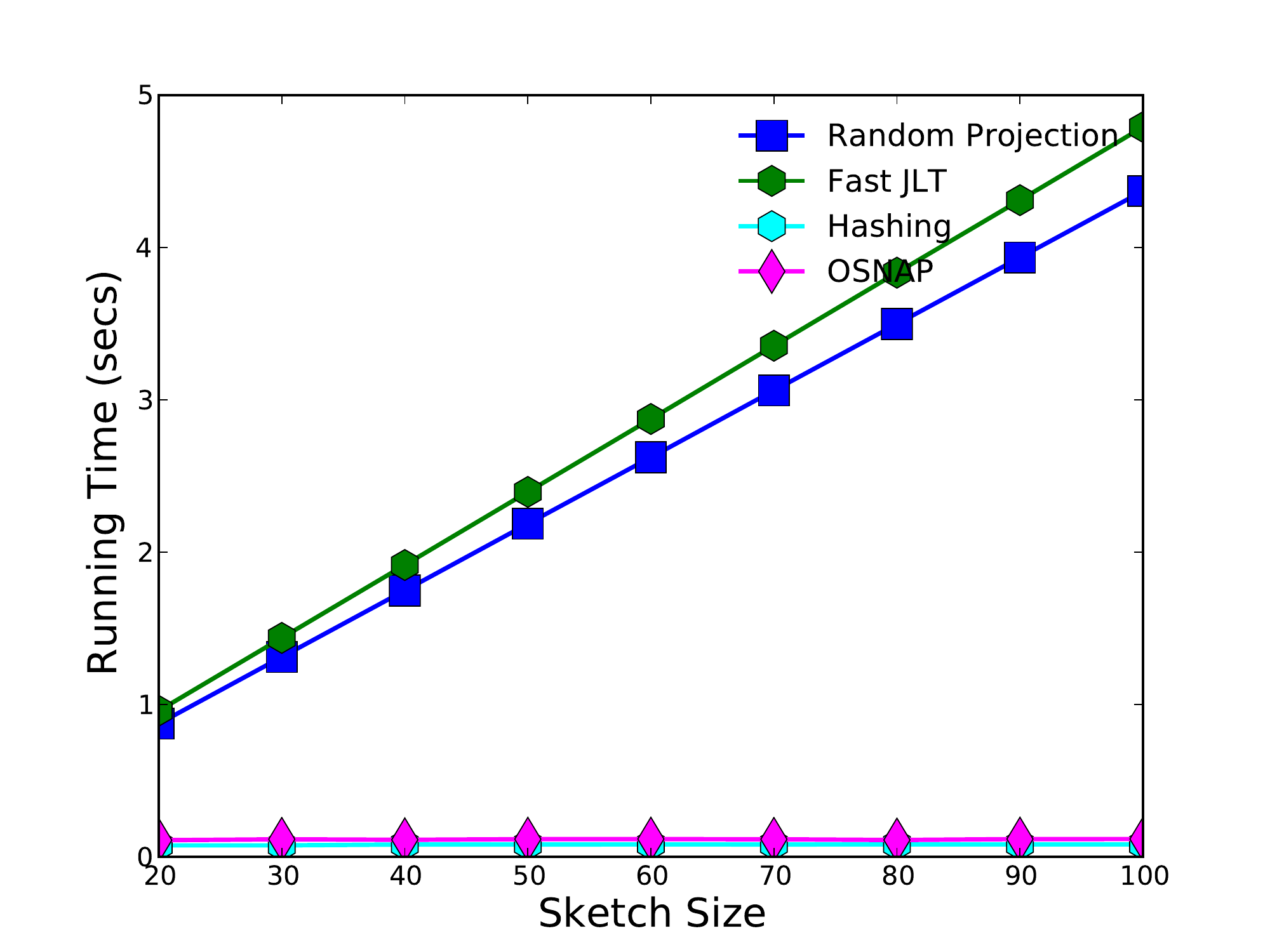}
\includegraphics[width=.305\linewidth]{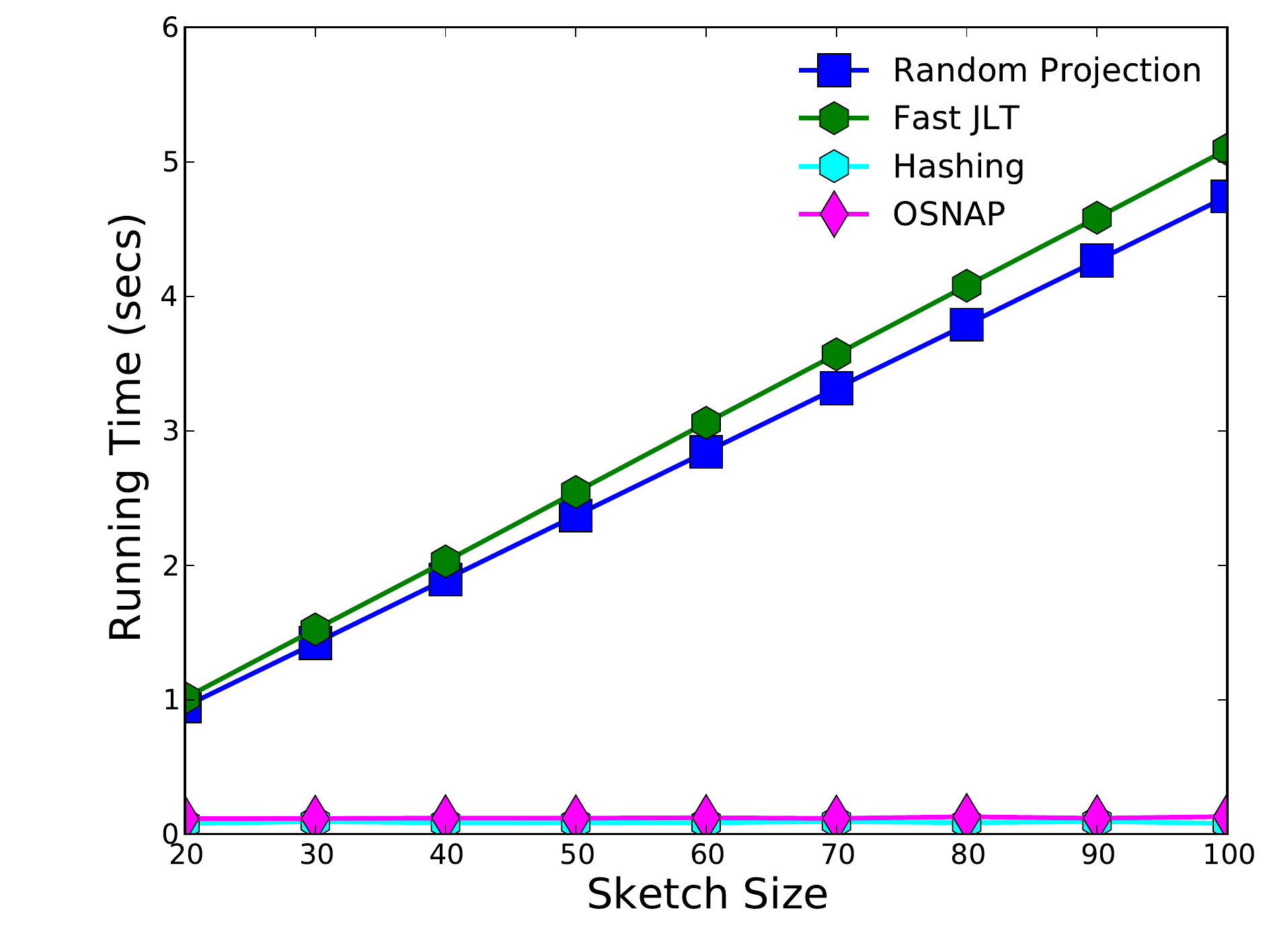}
\caption{\label{fig:proj-alg}Projection algorithms on \s{Birds}(left), \s{Spam}(middle), and \s{Random Noisy}(30)(right).} 
\end{centering}
\end{figure}


\subsection{Iterative Algorithms}

Here we consider variants of \FD.  We first explore the $\alpha$ parameter in Parametrized \FD, writing each version as $\alpha$-\FD. Then we compare against all of the other variants using explores from Parametrized \FD.  

In Figure \ref{fig:PFD} and \ref{fig:PFD-m}, we explore the effect of the parameter $\alpha$, and run variants with $\alpha \in \{0.2, 0.4, 0.6, 0.8\}$, comparing against \FD ($\alpha = 1$) and \iSVD ($\alpha = 0$).   Note that the guaranteed error gets worse for smaller $\alpha$, so performance being equal, it is preferable to have larger $\alpha$.  
Yet, we observe empirically on datasets \s{Birds}, \s{Spam}, and \s{Random Noisy} that \FD is consistently the worst algorithm, and \iSVD is fairly consistently the best, and as $\alpha$ decreases, the observed error improves.  The difference can be quite dramatic; for instance in the \s{Spam} dataset, for $\ell=20$, \FD has \s{err} = 0.032 while \iSVD and $0.2$-\FD have \s{err} = 0.008.  Yet, as $\ell$ approaches $100$, all algorithms seems to be approaching the same small error.   
In Figure \ref{fig:PFD-m}, we explore the effect of $\alpha$-\FD 
on \s{Random Noisy} data by varying $m \in \{10,20,50\}$, and $m=30$ in Figure \ref{fig:PFD}.  
We observe that all algorithms get smaller error for smaller $m$ (there are fewer ``directions'' to approximate), but that each $\alpha$-\FD variant reaches $0.005$ \s{err} before $\ell=100$, sooner for smaller $\alpha$; eventually ``snapping'' to a smaller $0.002$ \s{err} level.

\begin{figure}[t!]
\begin{centering}
\includegraphics[width=0.97\figsize]{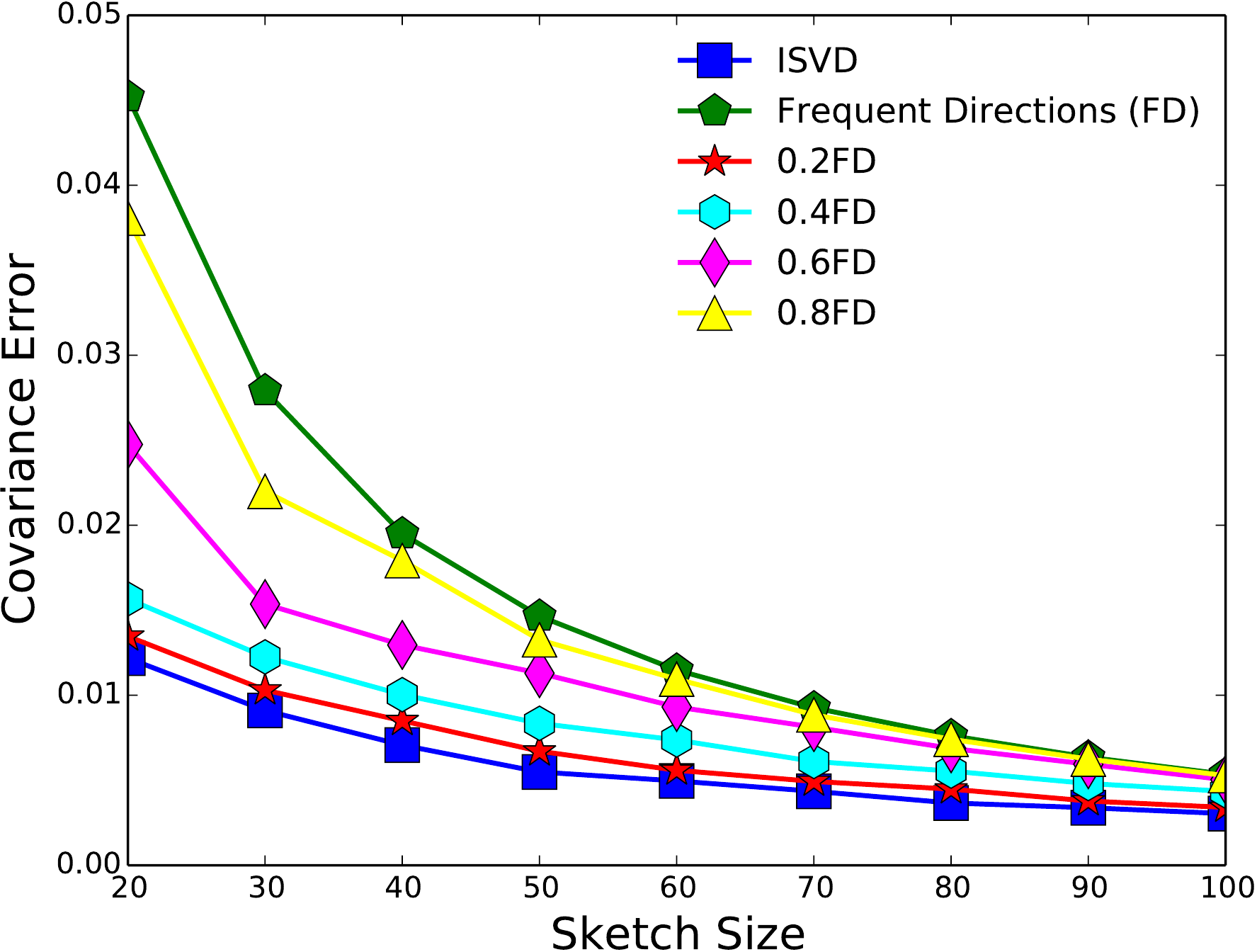}
\includegraphics[width=1.08\figsize]{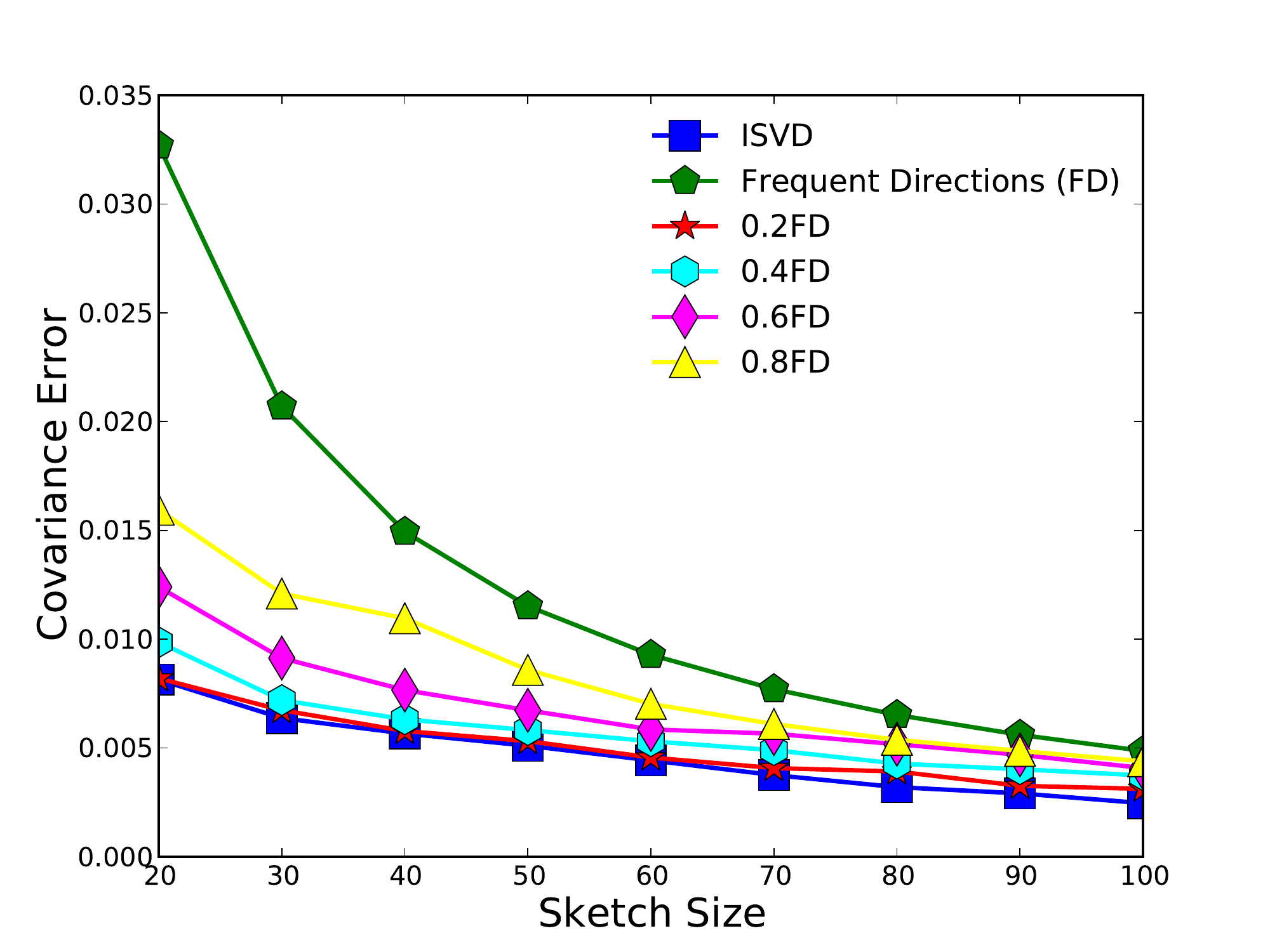}
\includegraphics[width=0.97\figsize]{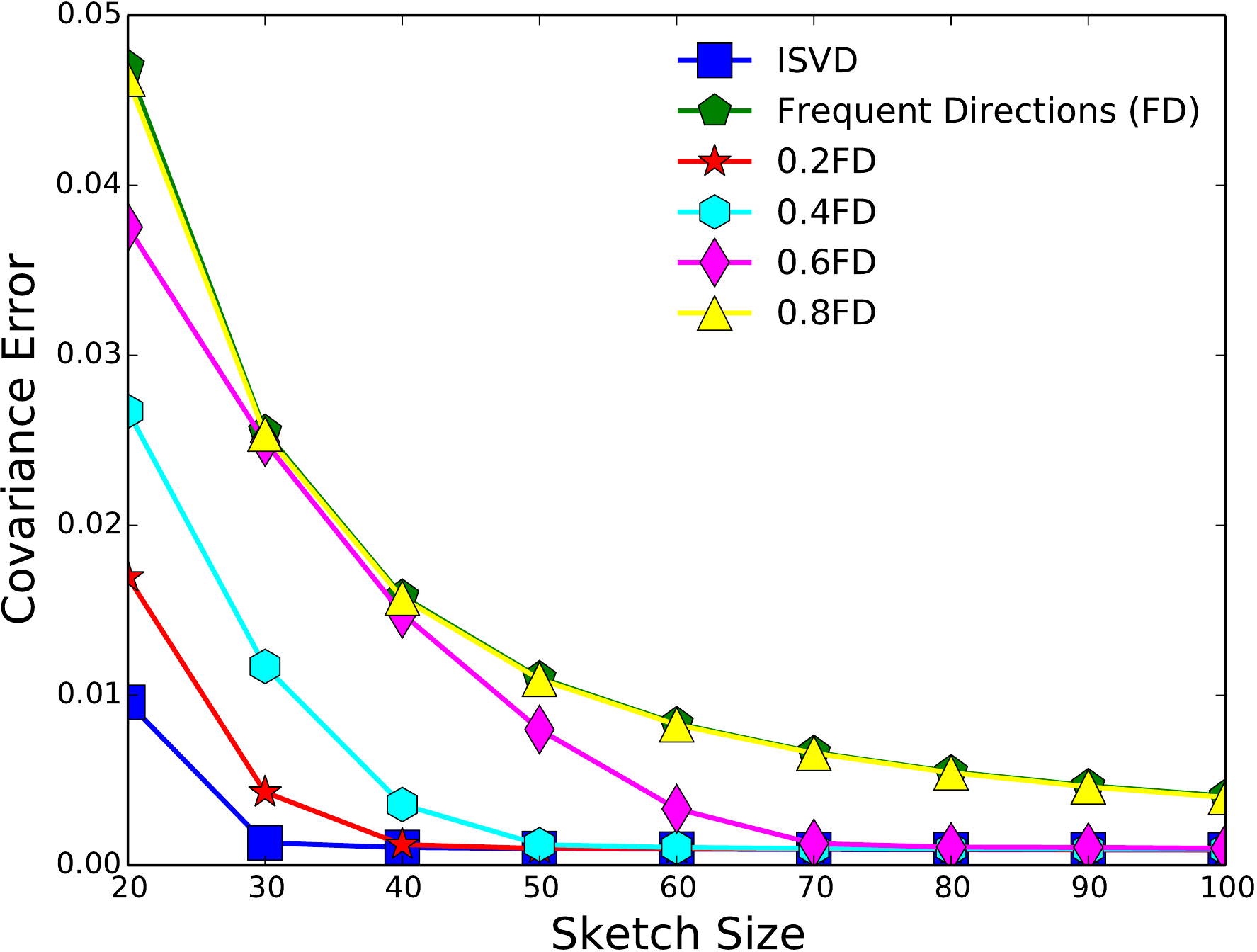}
\vspace{-4mm}
\caption{\label{fig:PFD}
Parametrized \FD on \s{Birds} (left), and \s{Spam} (middle), \s{Random Noisy}(30) (right).}  
\vspace{-.1in}
\end{centering}
\end{figure}

\begin{figure}[t!]
\begin{centering}
\includegraphics[width=\figsize]{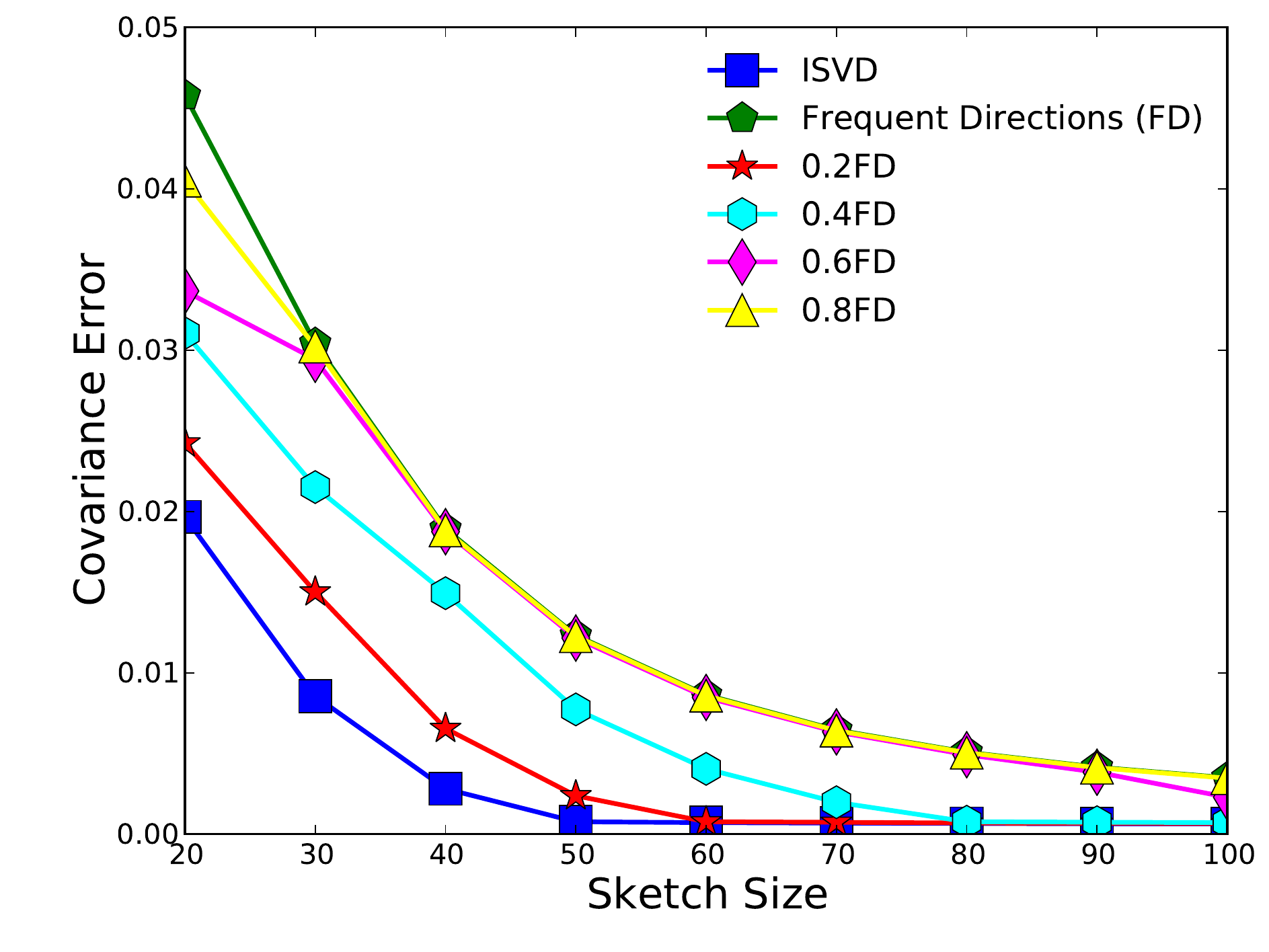}
\includegraphics[width=\figsize]{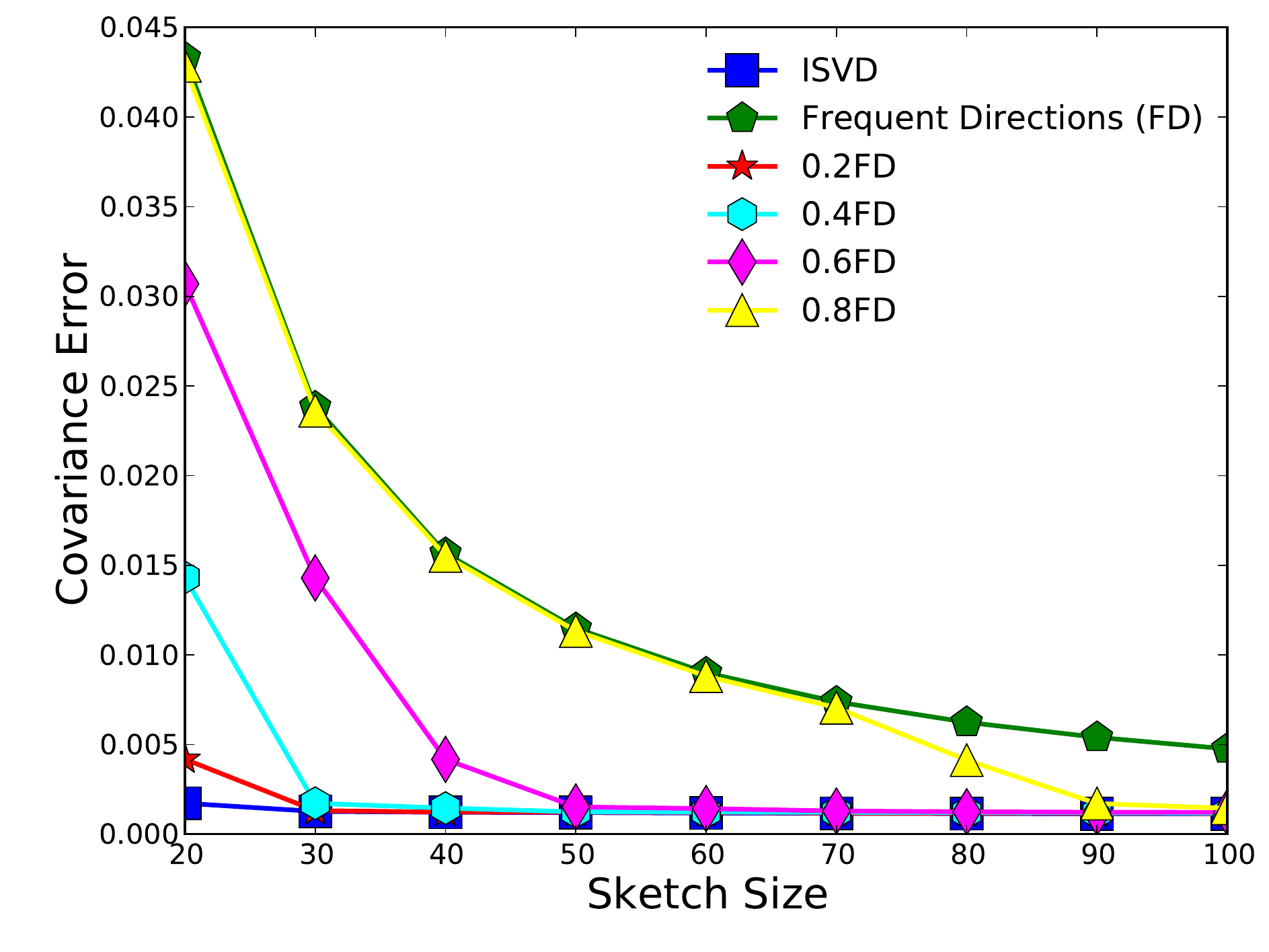}
\includegraphics[width=\figsize]{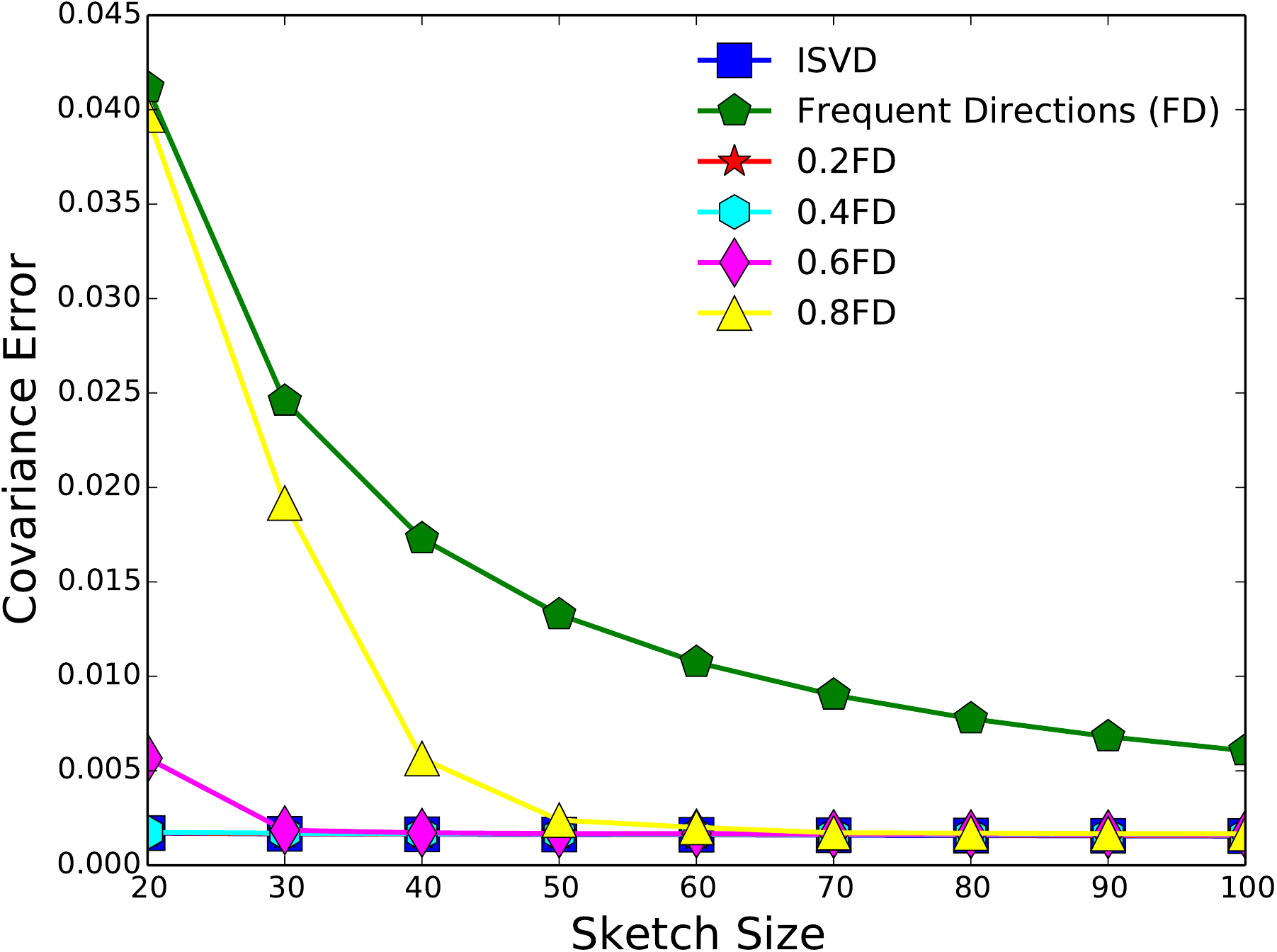}
\vspace{-4mm}
\caption{\label{fig:PFD-m}
Parametrized \FD on \s{Random Noisy} for $m=50$ (left), $20$ (middle), $10$ (right).}  
\vspace{-.1in}
\end{centering}
\end{figure}

Next in Figure \ref{fig:iter-alg}, we compare \iSVD, \FD, and $0.2$-\FD with two groups of variants: one based on SS streaming algorithm (\CFD and \SSD) and another based on Fast \FD.  We see that \CFD and \SSD typically perform slightly better than \FD in \s{cov-err} and same or worse in \s{proj-err}, but not nearly as good as $0.2$-\FD and \iSVD. 
Perhaps it is surprising that although SpaceSavings variants empirically improve upon MG variants for frequent items, $0.2$-\FD (based on MG) can largely outperform all the SS variants on matrix sketching.  


All variants achieve a very small error, but $0.2$-\s{FD}, \iSVD, and \s{Fast $0.2$-\FD} consistently matches or outperforms others in both \s{cov-err} and \s{proj-err} while \s{Fast \FD} incurs more error compare to other algorithms. 
We also observe that \s{Fast \FD} and \s{Fast $0.2$-\FD} are significantly (sometimes $10$ times) faster than \FD, \iSVD, and \s{$0.2$-\FD}.  \s{Fast \FD} takes less time, sometimes half as much compared to \s{Fast $0.2$-\FD}, however, given its much smaller error \s{Fast $0.2$-\FD} seems to have the best all-around performance.

\begin{figure}[t!]
\begin{centering}
{\tiny \textsf{Birds} \hspace{47mm} \textsf{Spam} \hspace{44mm} \textsf{Random Noisy}} 
\\
\includegraphics[width=\figsize]{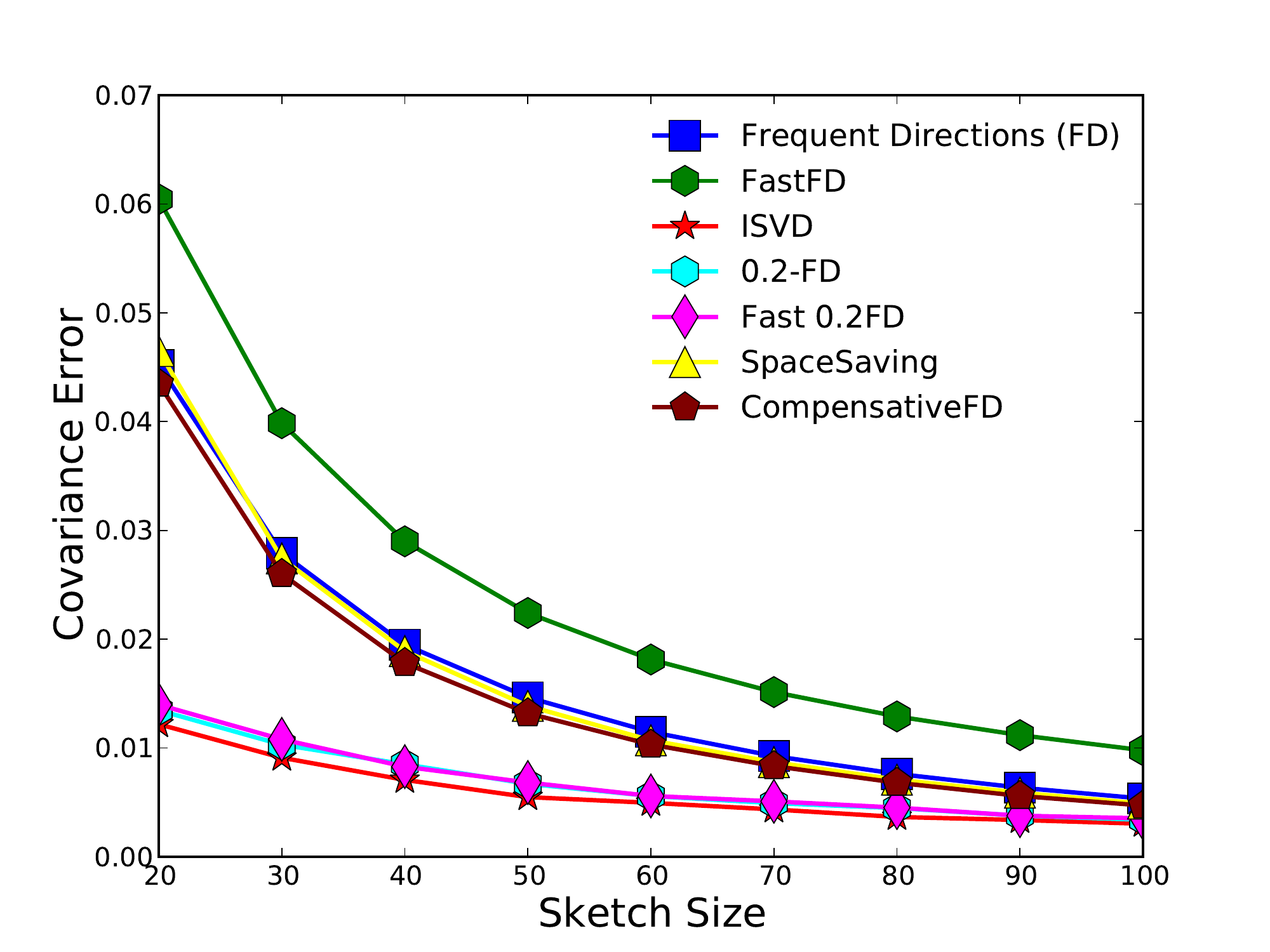}
\includegraphics[width=.33\linewidth]{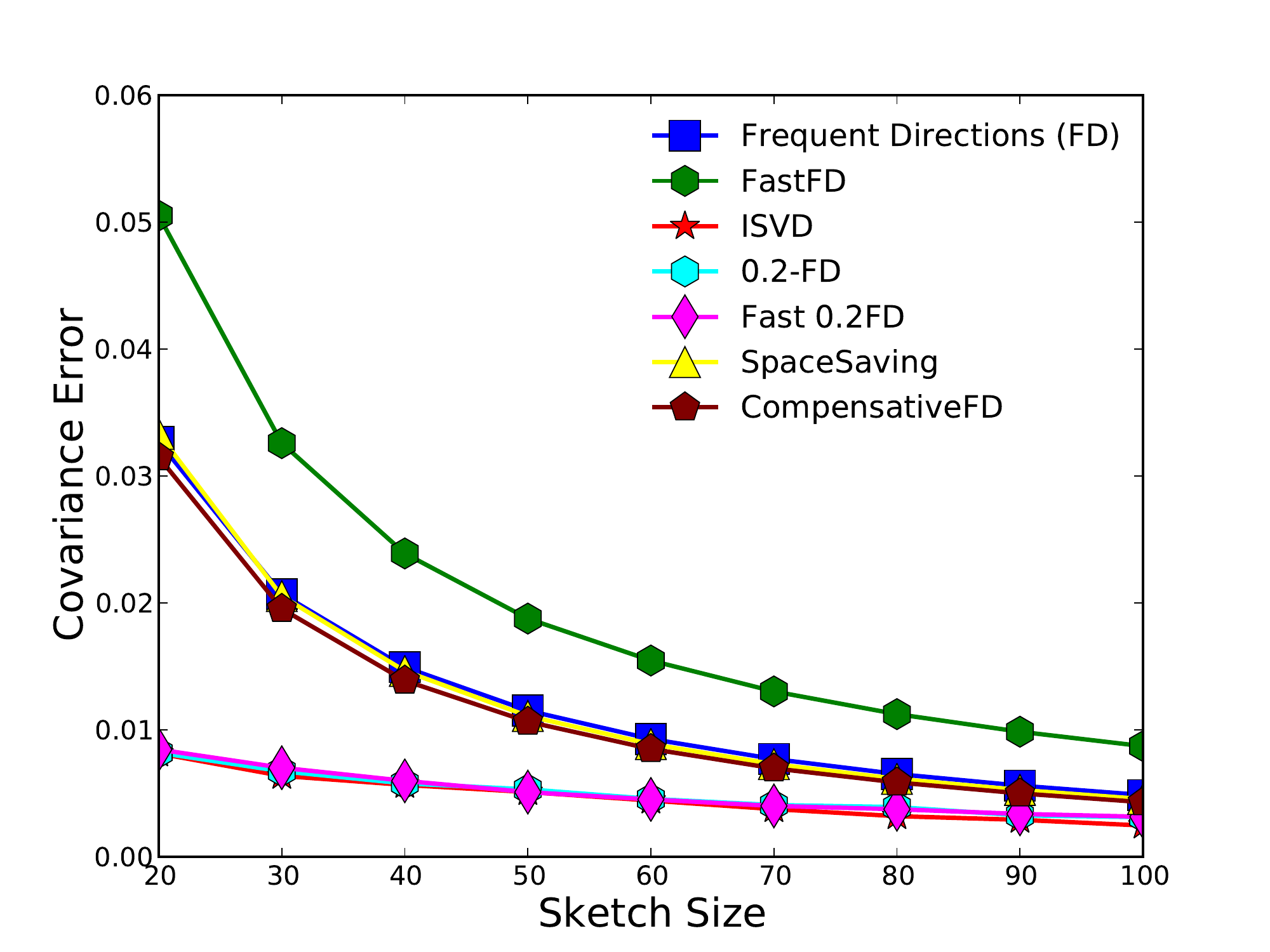}
\includegraphics[width=.32\linewidth]{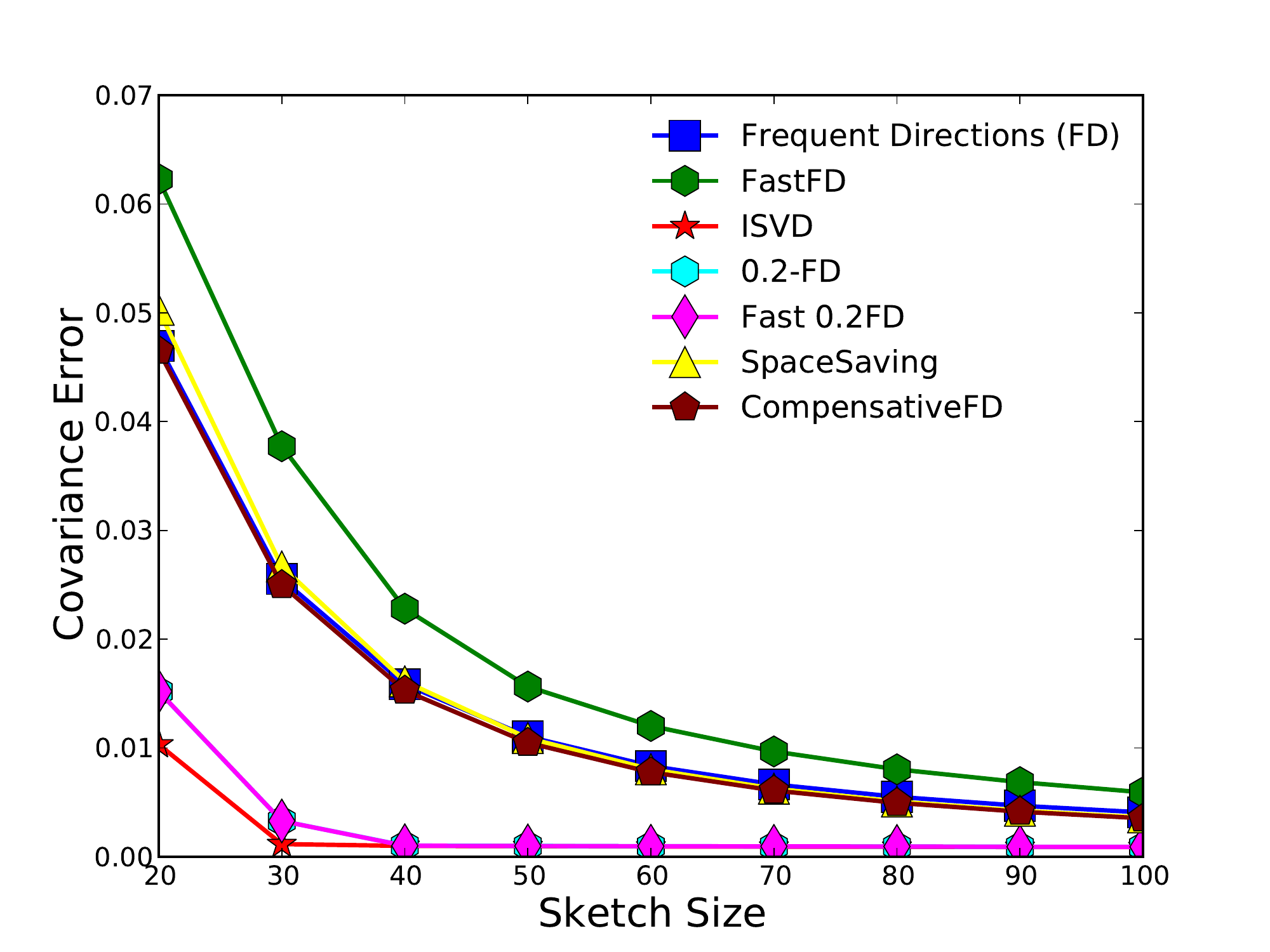}

\includegraphics[width=.32\linewidth]{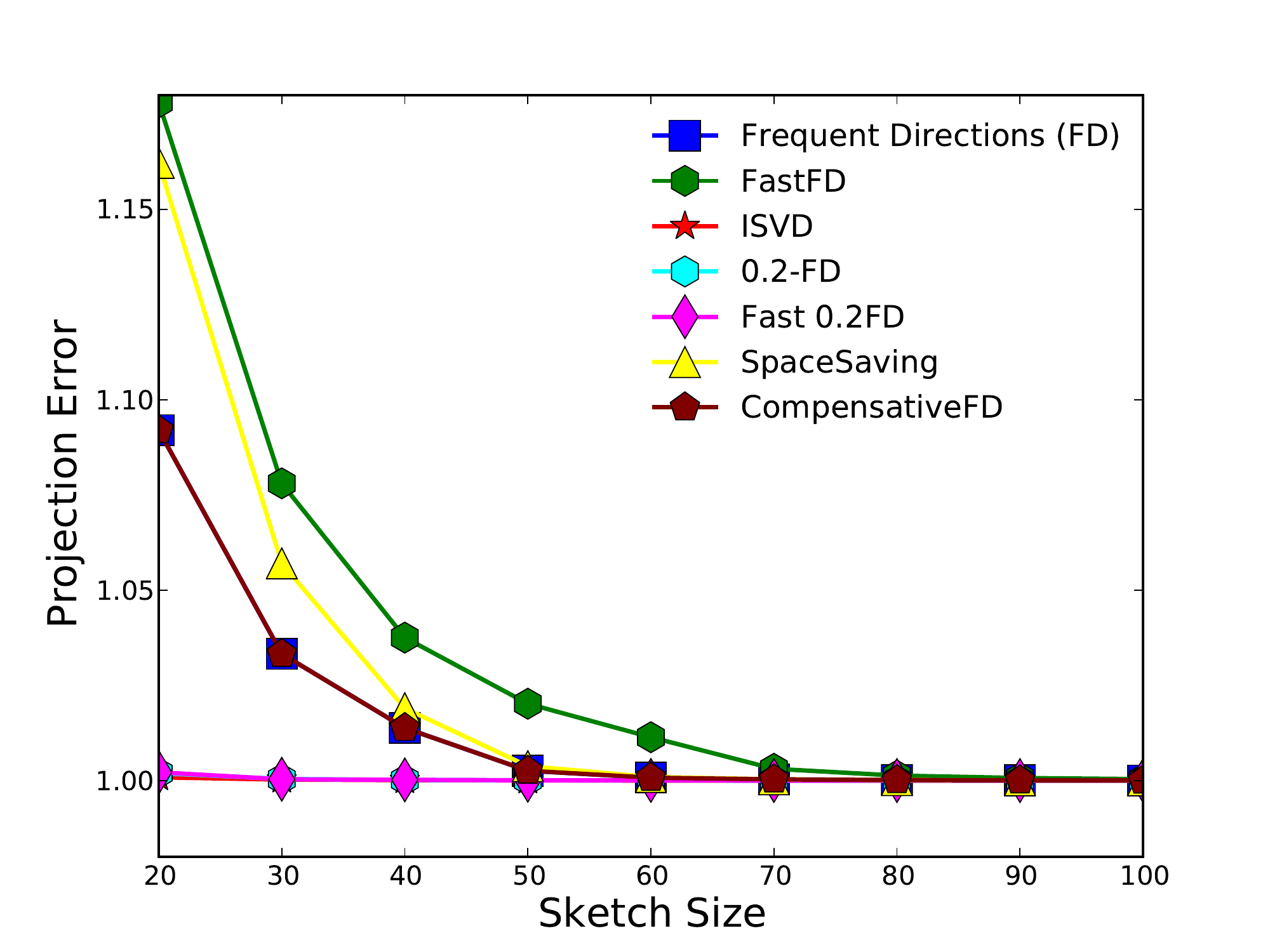}
\includegraphics[width=.33\linewidth]{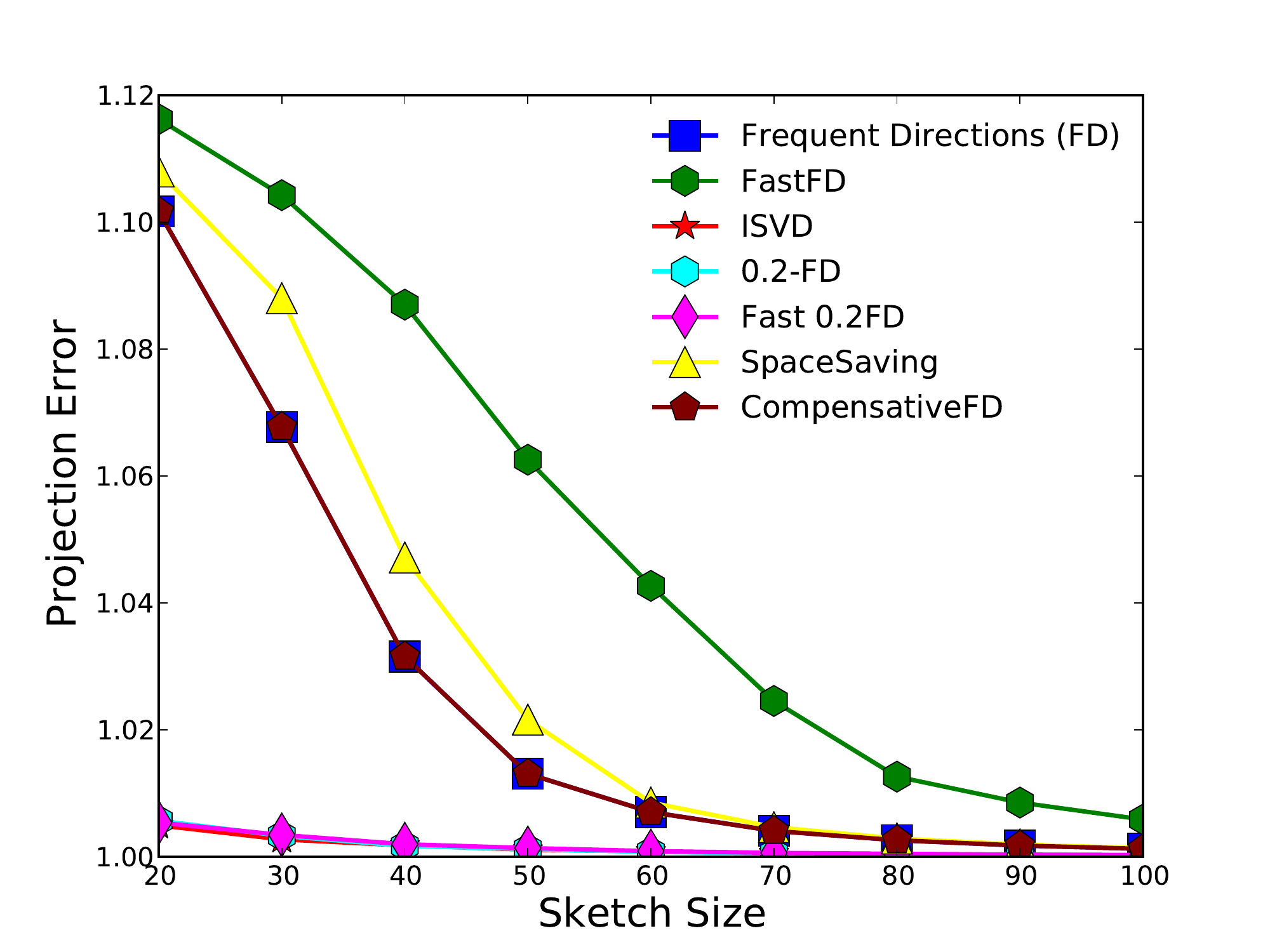}
\includegraphics[width=.32\linewidth]{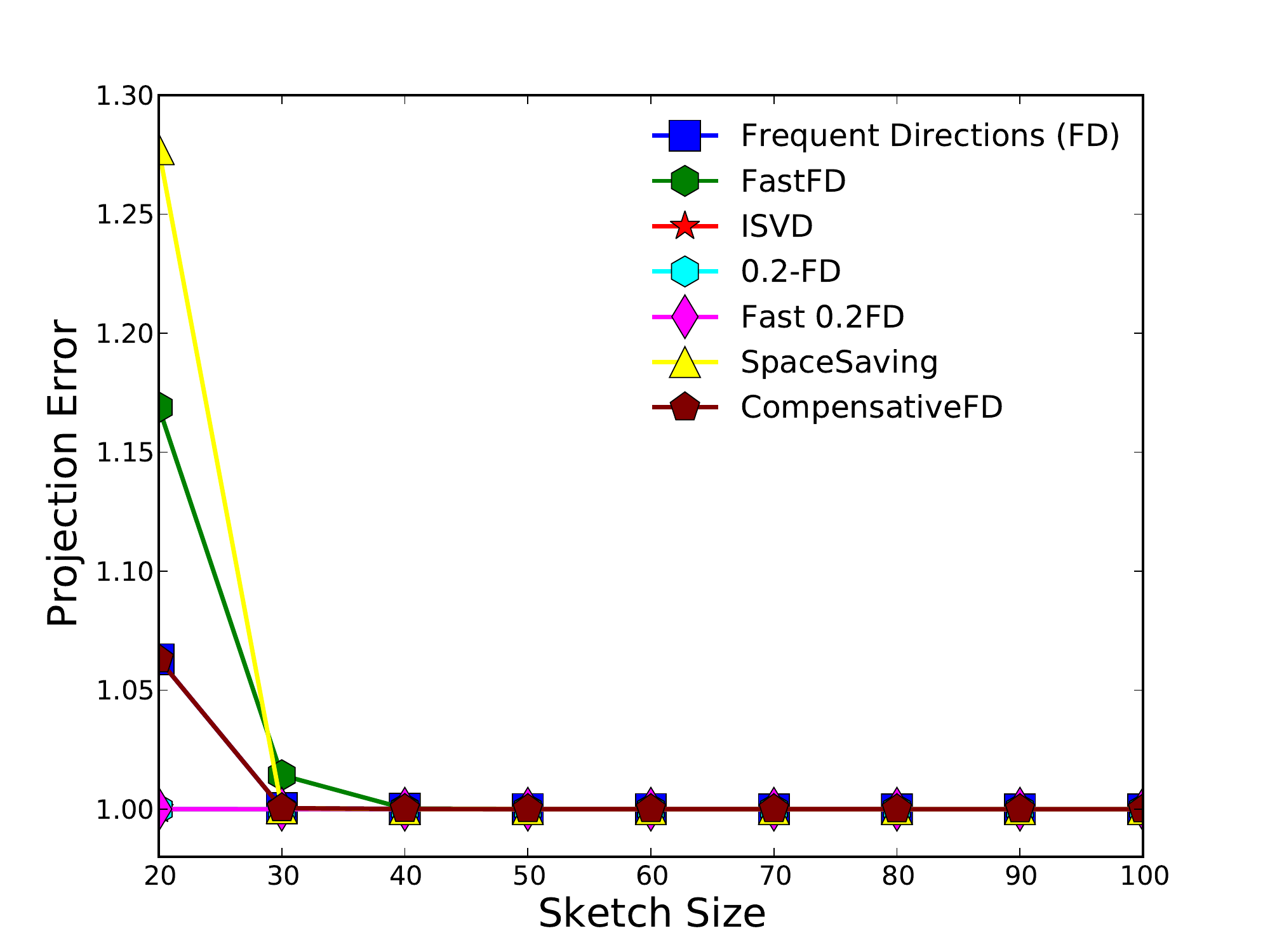}

\includegraphics[width=.32\linewidth]{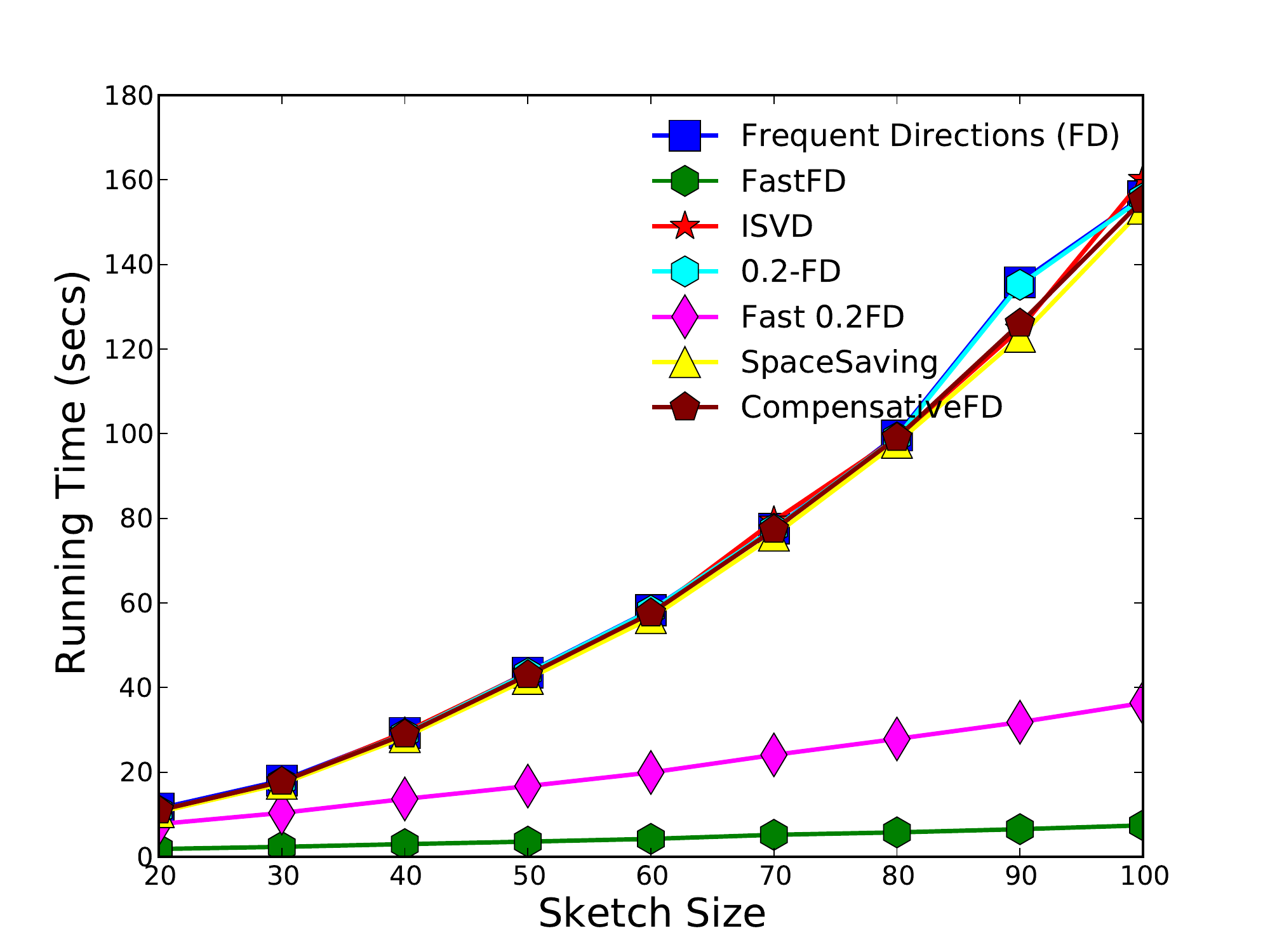}
\includegraphics[width=.32\linewidth]{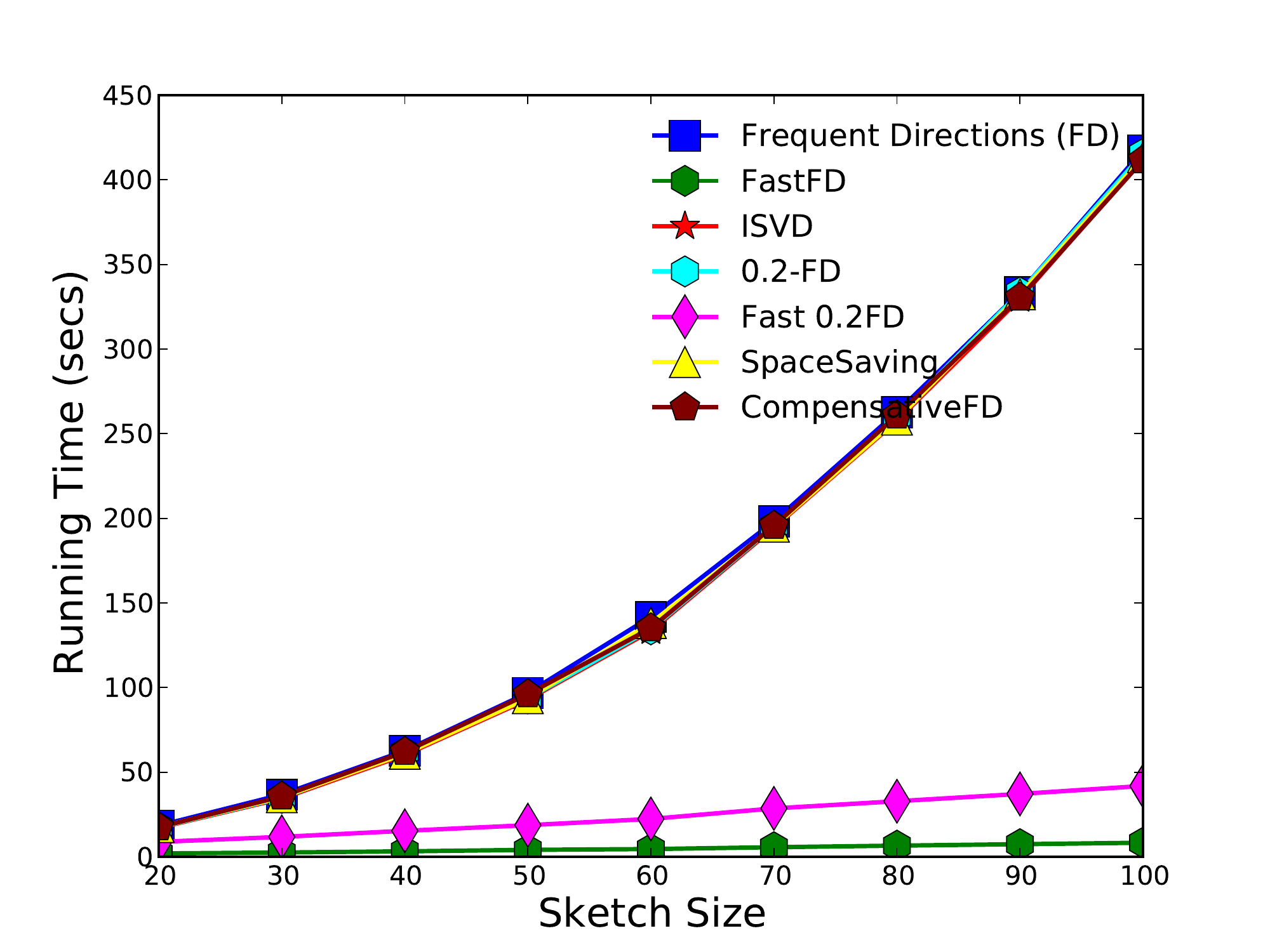}
\includegraphics[width=.32\linewidth]{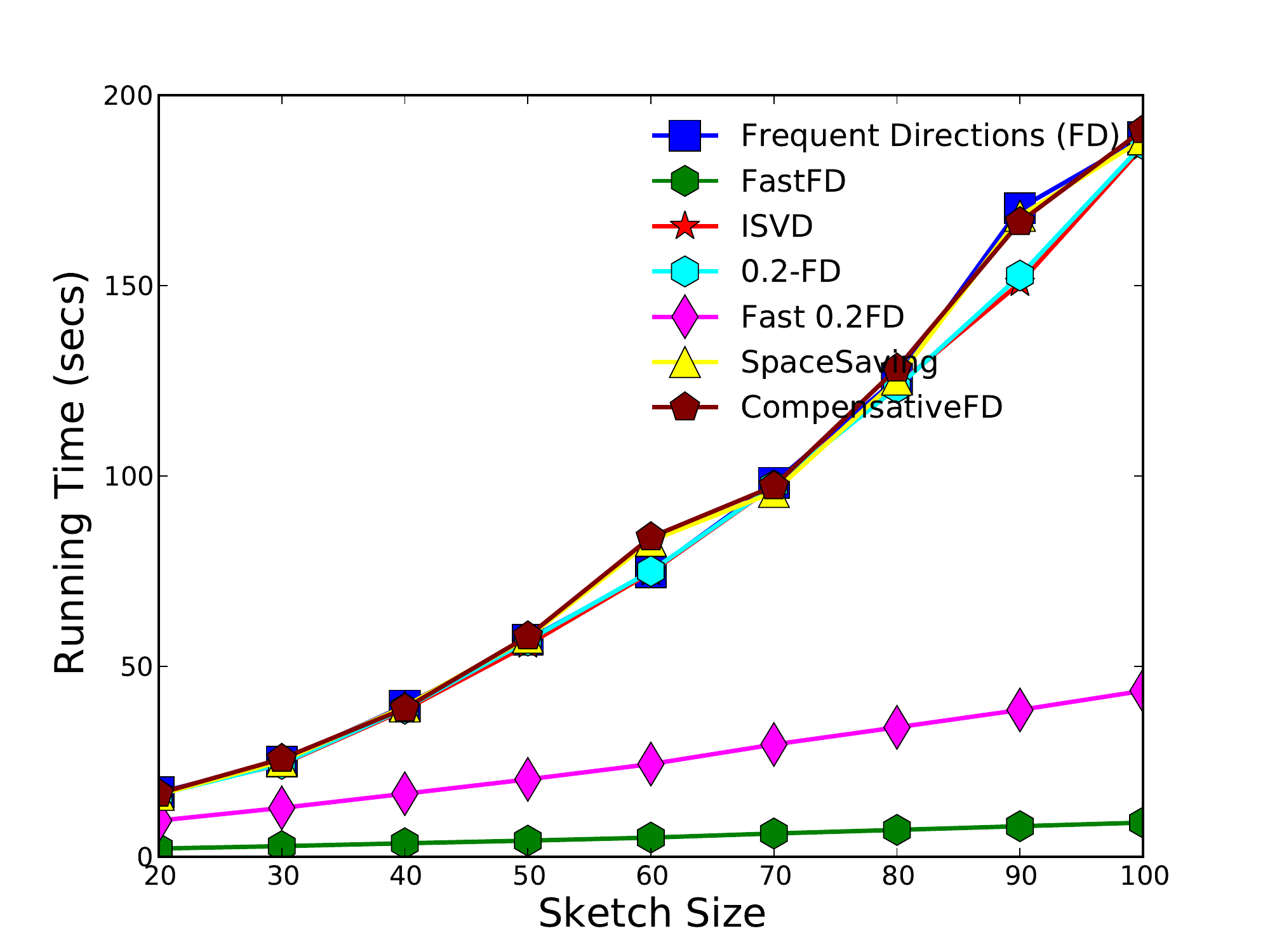}

\caption{\label{fig:iter-alg} Iterative algorithms on \s{Birds}(left), \s{Spam}(middle), and \s{Random Noisy}(30)(right).} 
\end{centering}
\end{figure}

\paragraph{Data adversarial to \iSVD.}
Next, using the \s{Adversarial} construction we show that \iSVD is not always better in practice. In Figure \ref{fig:adverse}, we see that \iSVD can perform much worse than other techniques.  Although at $\ell=20$, \iSVD and \FD roughly perform the same (with about \s{err} = 0.09), \iSVD does not improve much as $\ell$ increases, obtaining only \s{err} = 0.08 for $\ell=100$.  On the other hand, \FD (as well as \CFD and \SSD) decrease markedly and consistently to \s{err} = 0.02 for $\ell=100$.  
Moreover, all version of $\alpha$-\FD obtain roughly \s{err}=$0.005$ already for $\ell=20$.  
The large-norm directions are the first 4 singular vectors (from the second part of the stream) and once these directions are recognized as having the largest singular vectors, they are no longer decremented in any Parametrized \FD algorithm.

\begin{figure}[t!]
\begin{centering}
\includegraphics[width=\figsize]{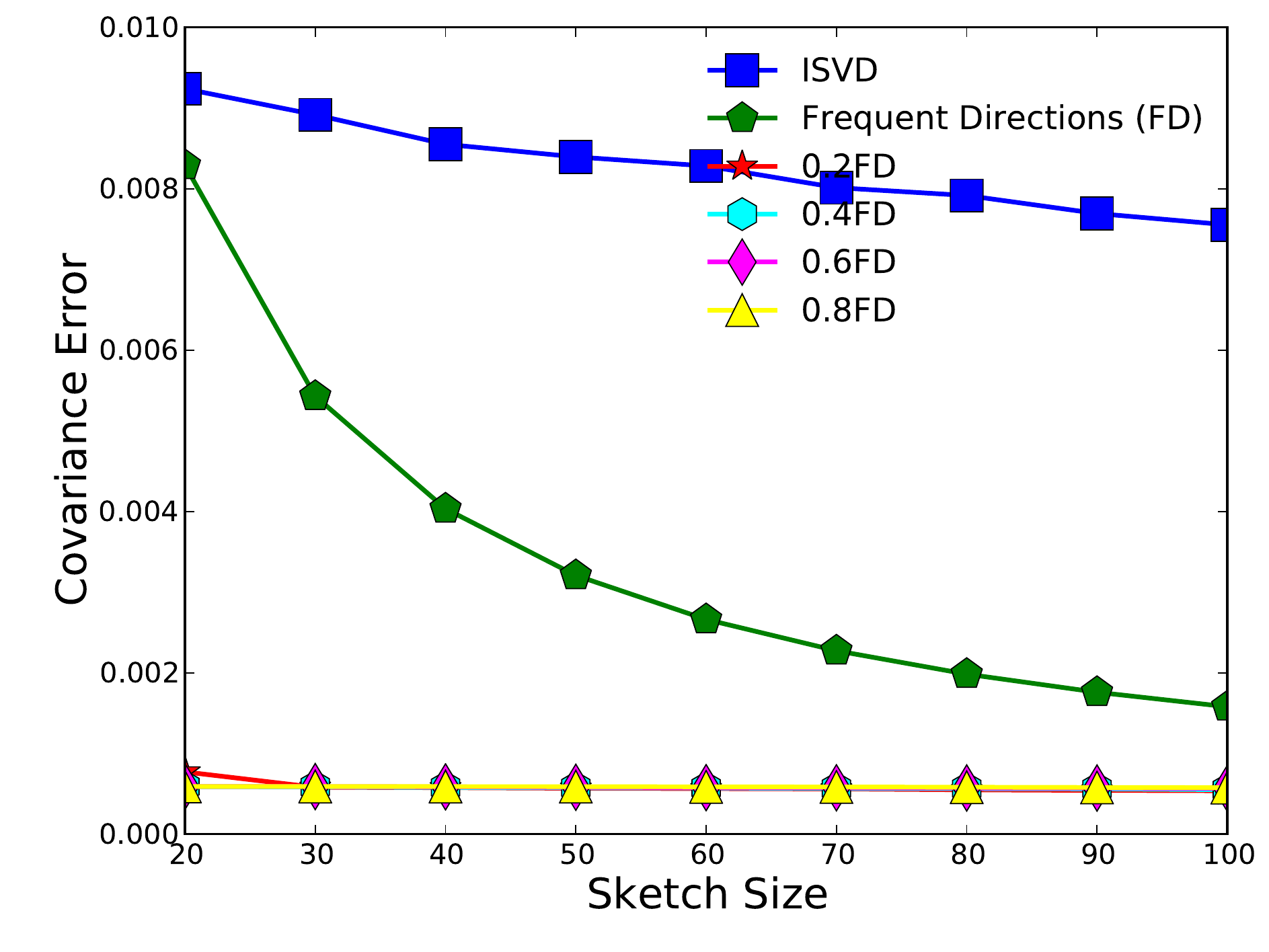}
\includegraphics[width=\figsize]{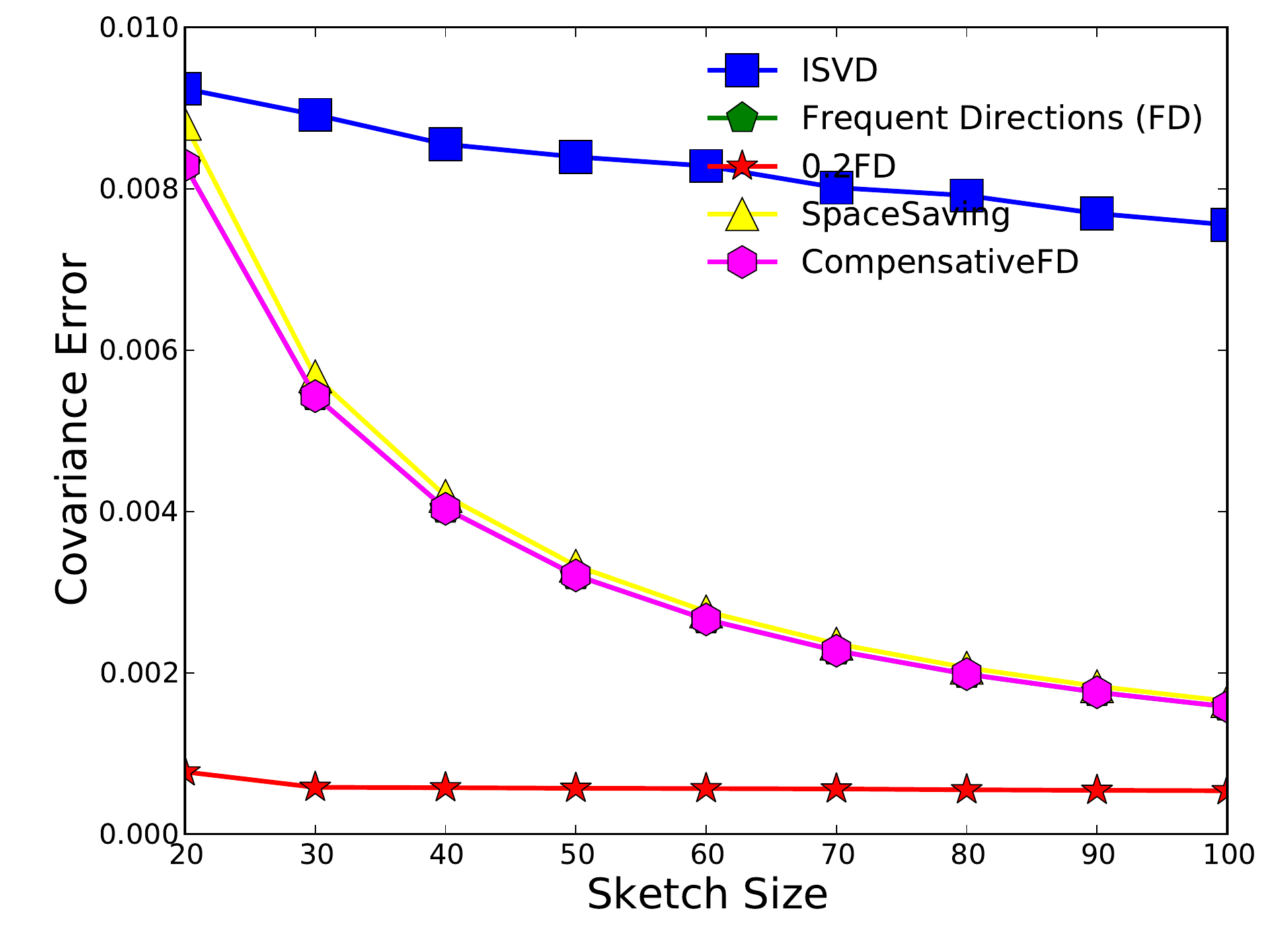}
\vspace{-4mm}
\caption{\label{fig:adverse}
Demonstrating dangers of \iSVD on \s{Adversarial} data.}  
\vspace{-.1in}
\end{centering}
\end{figure}

To wrap up this section, we demonstrate the scalability of these approaches on a much larger real data set \textsf{ConnectUS}.  Figure \ref{fig:connectus} shows variants of Parameterized \FD, and other iterative variants
on this dataset.  As the derived bounds on covariance error based on sketch size do not depend on $n$, the number of rows in $A$, it is not surprising that the performance of most algorithms is unchanged.  There are just a couple differences to point out.  
First, no algorithm converges as close to $0$ error as with the other smaller data sets; this is likely because with the much larger size, there is some variation that can not be captured even with $\ell = 100$ rows of a sketch.  
Second, \iSVD performs noticeably worse than the other \FD-based algorithms (although still significantly better than the leading randomized algorithms).  This likely has to do with the sparsity of \s{ConnectUS} combined with a data drift.  After building up a sketch on the first part of the matrix, sparse rows are observed orthogonal to existing directions.  The orthogonality, the same difficult property as in \s{Adversarial}, likely occurs here because the new rows have a small number of non-zero entrees, and all rows in the sketch have zeros in these locations; these correspond to the webpages marked by one of the unconnected users.

\begin{figure}[t!]
\begin{centering}
\includegraphics[width=\figsize]{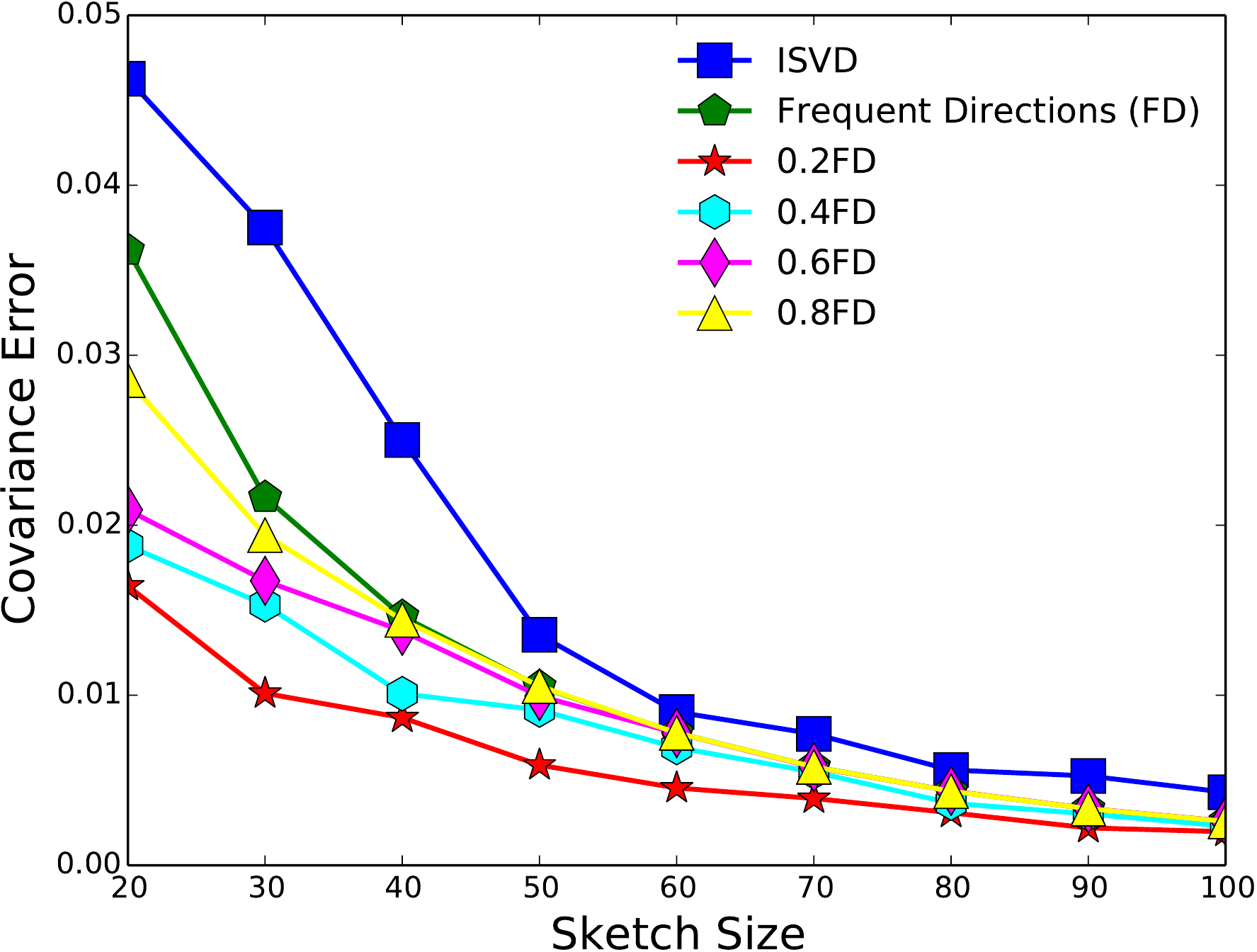}
\includegraphics[width=\figsize]{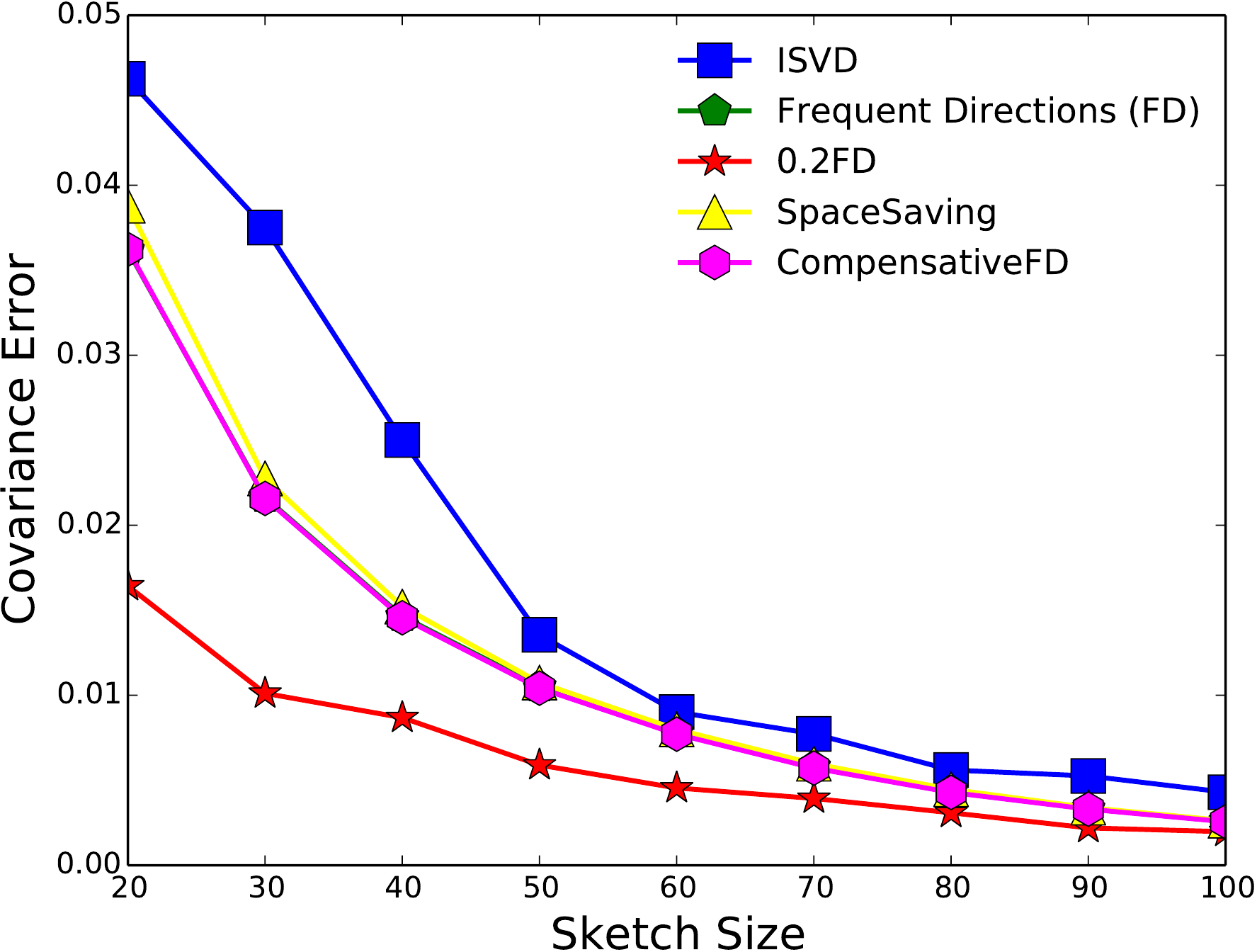}
\vspace{-4mm}
\caption{\label{fig:connectus}
Parameterized \FD (left), other iterative (right)
algorithms on \s{ConnectUS} dataset.}  
\end{centering}
\end{figure}

\subsection{Global Comparison}
\label{ssec:global-eval}

\begin{figure}[t!]
\begin{centering}
\includegraphics[width=.305\linewidth]{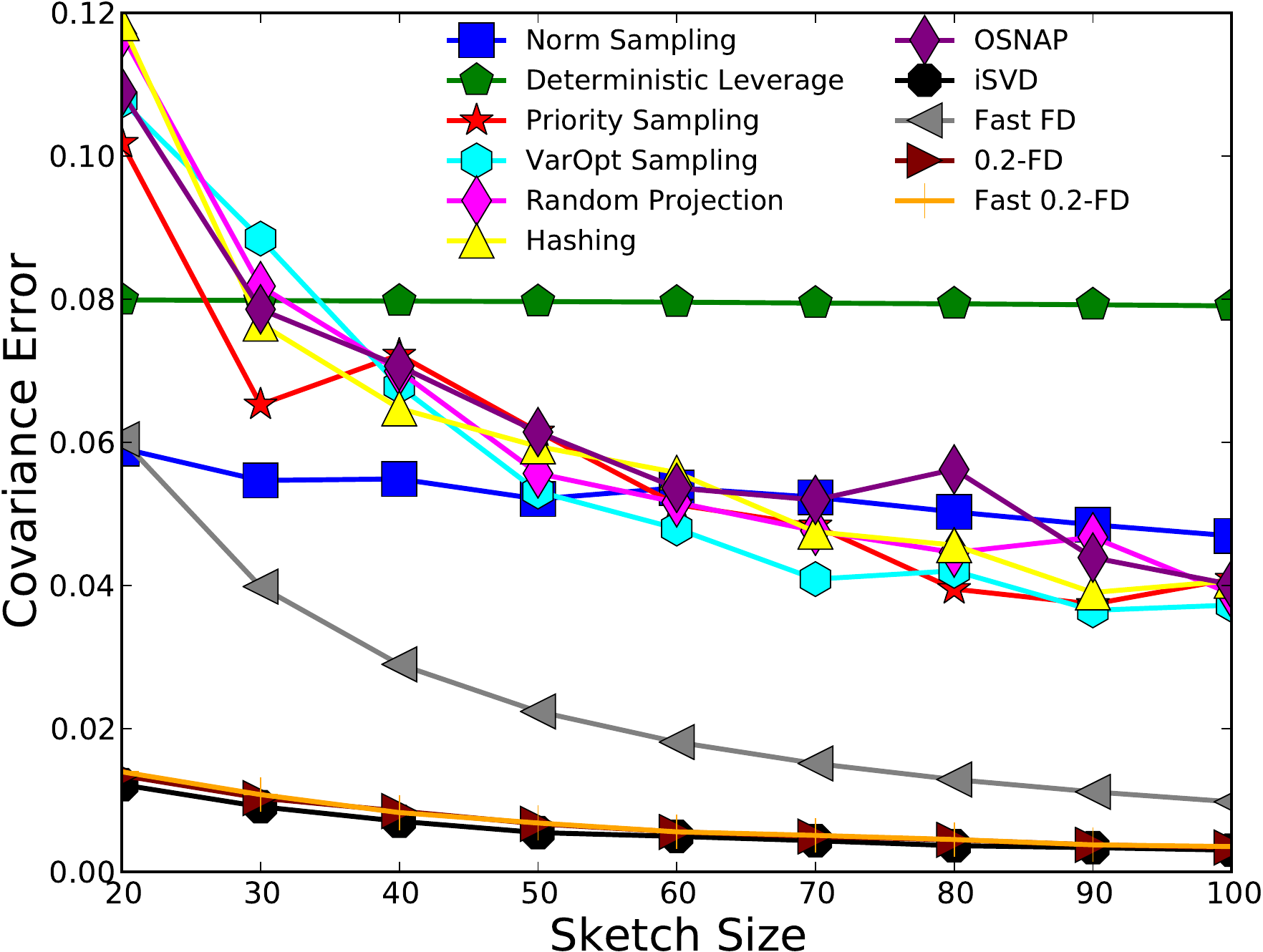}
\includegraphics[width=.34\linewidth]{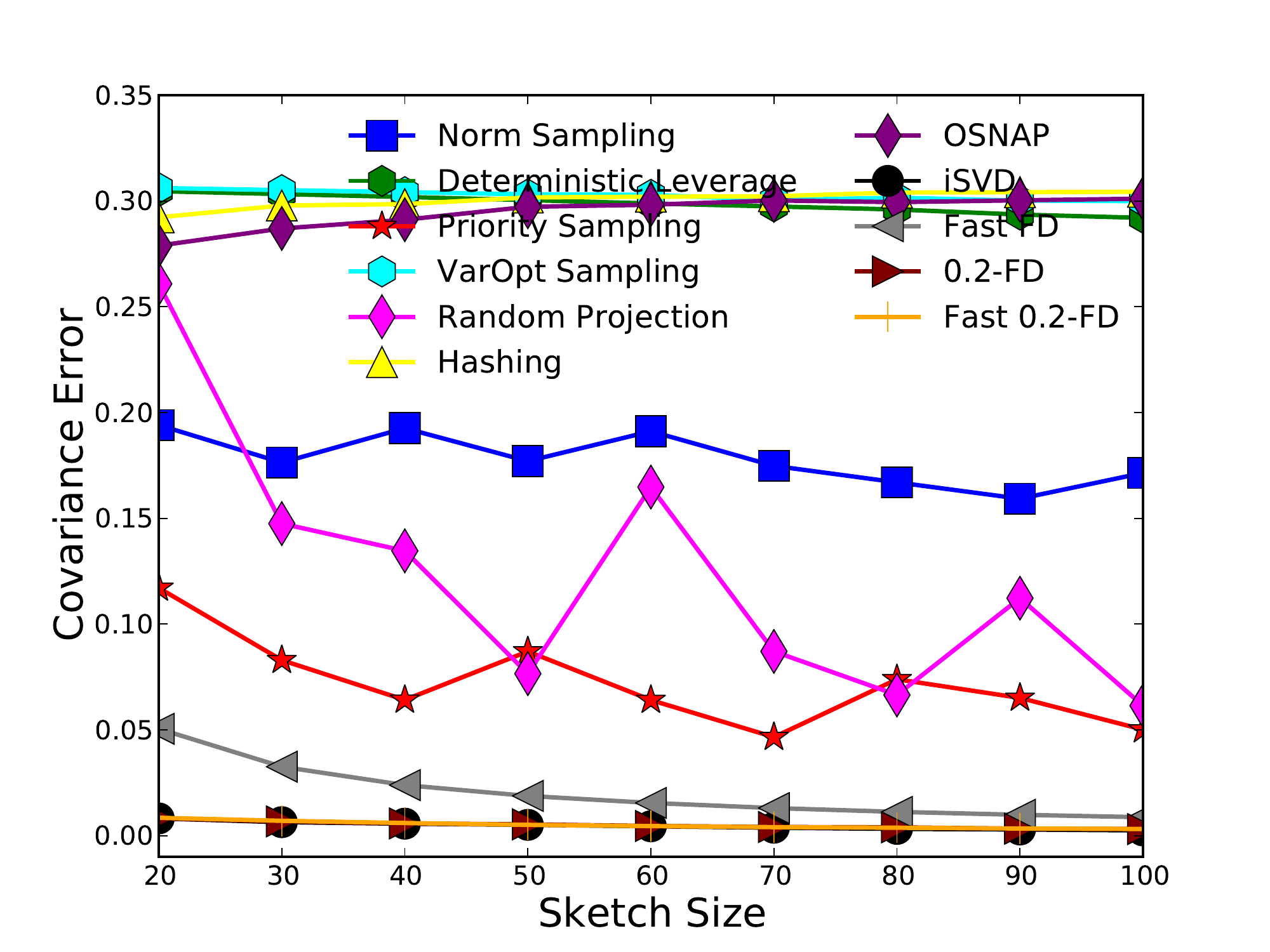}
\includegraphics[width=.305\linewidth]{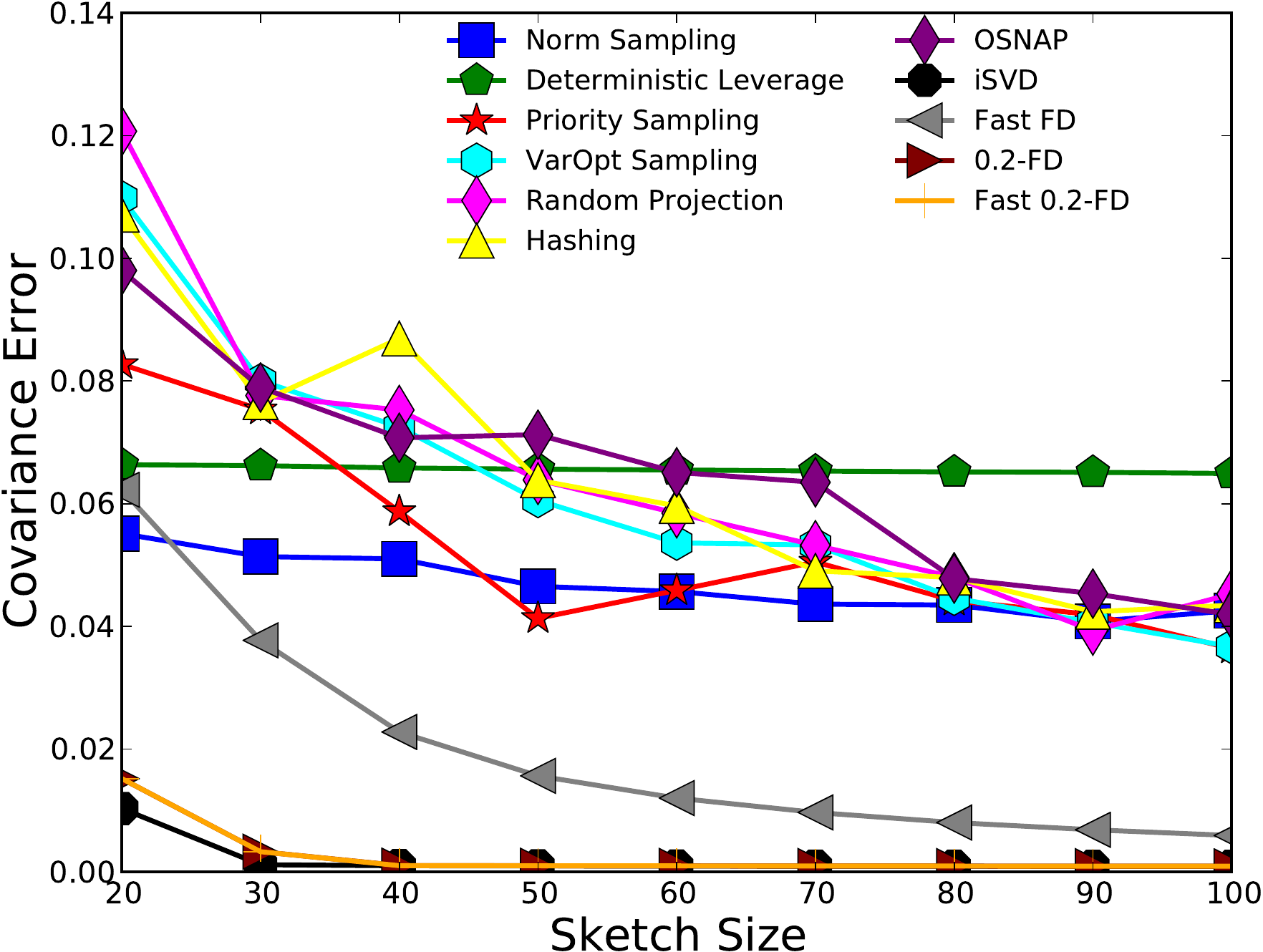}

\includegraphics[width=.305\linewidth]{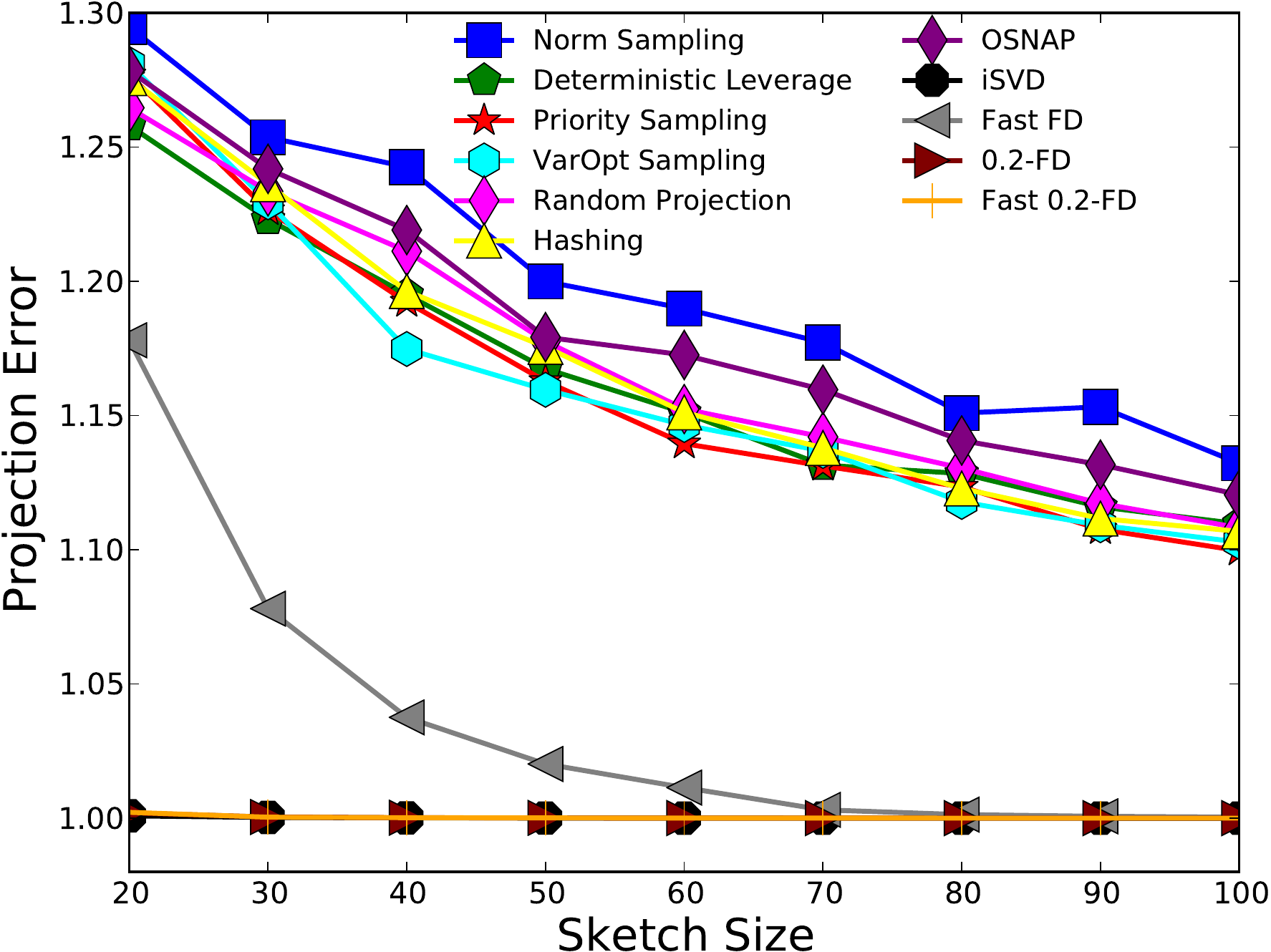}
\includegraphics[width=.34\linewidth]{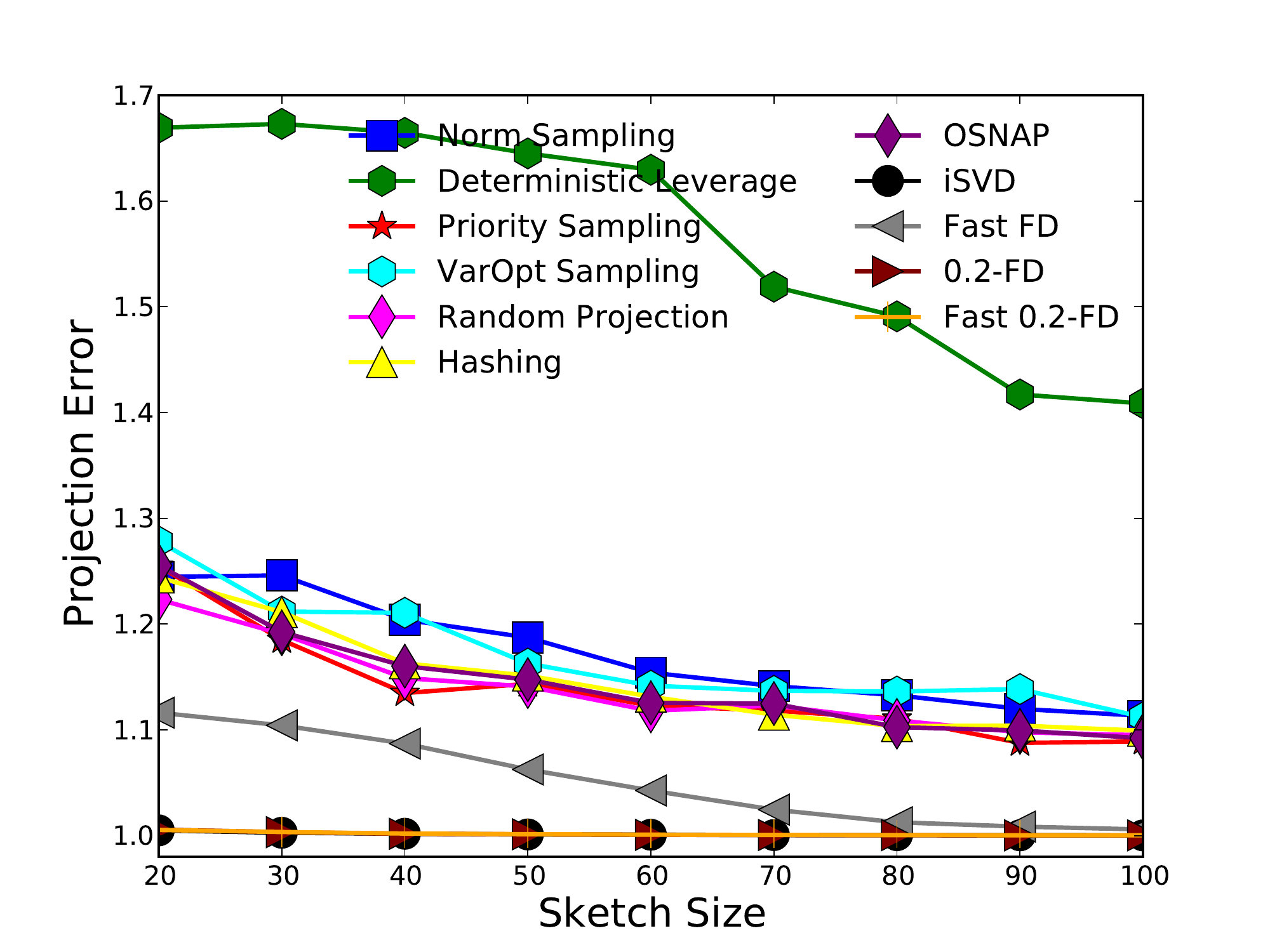}
\includegraphics[width=.305\linewidth]{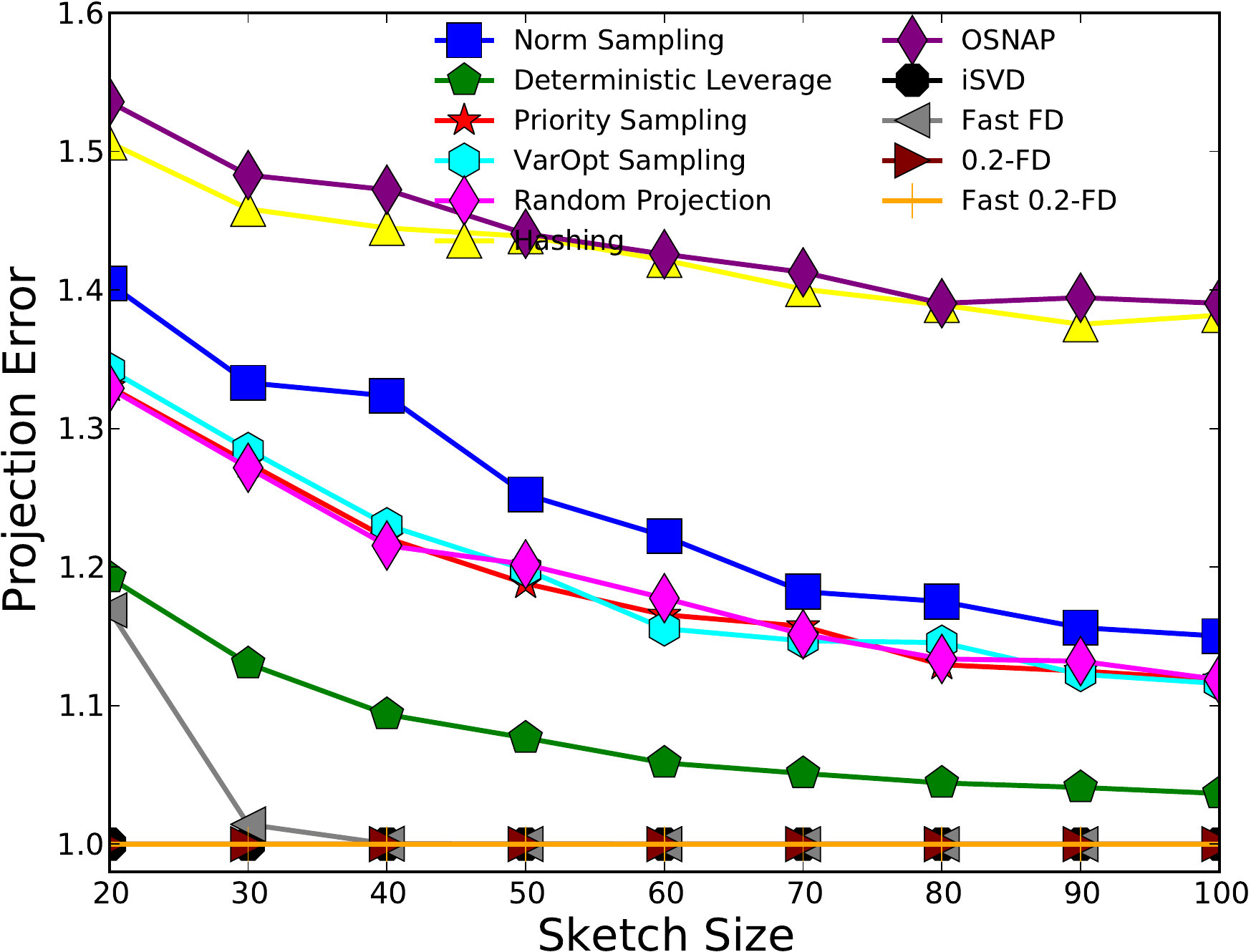}

\includegraphics[width=.305\linewidth]{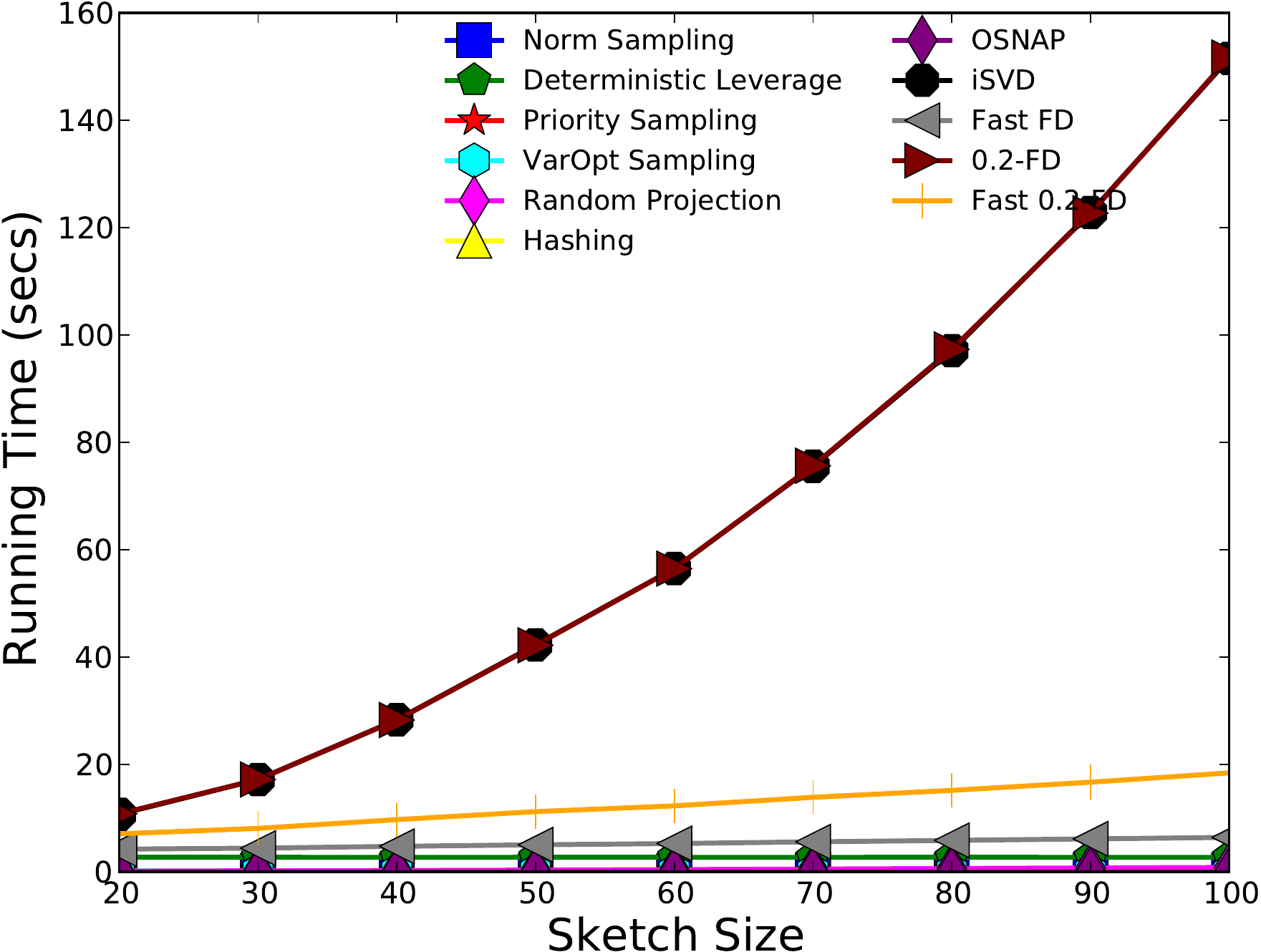}
\includegraphics[width=.34\linewidth]{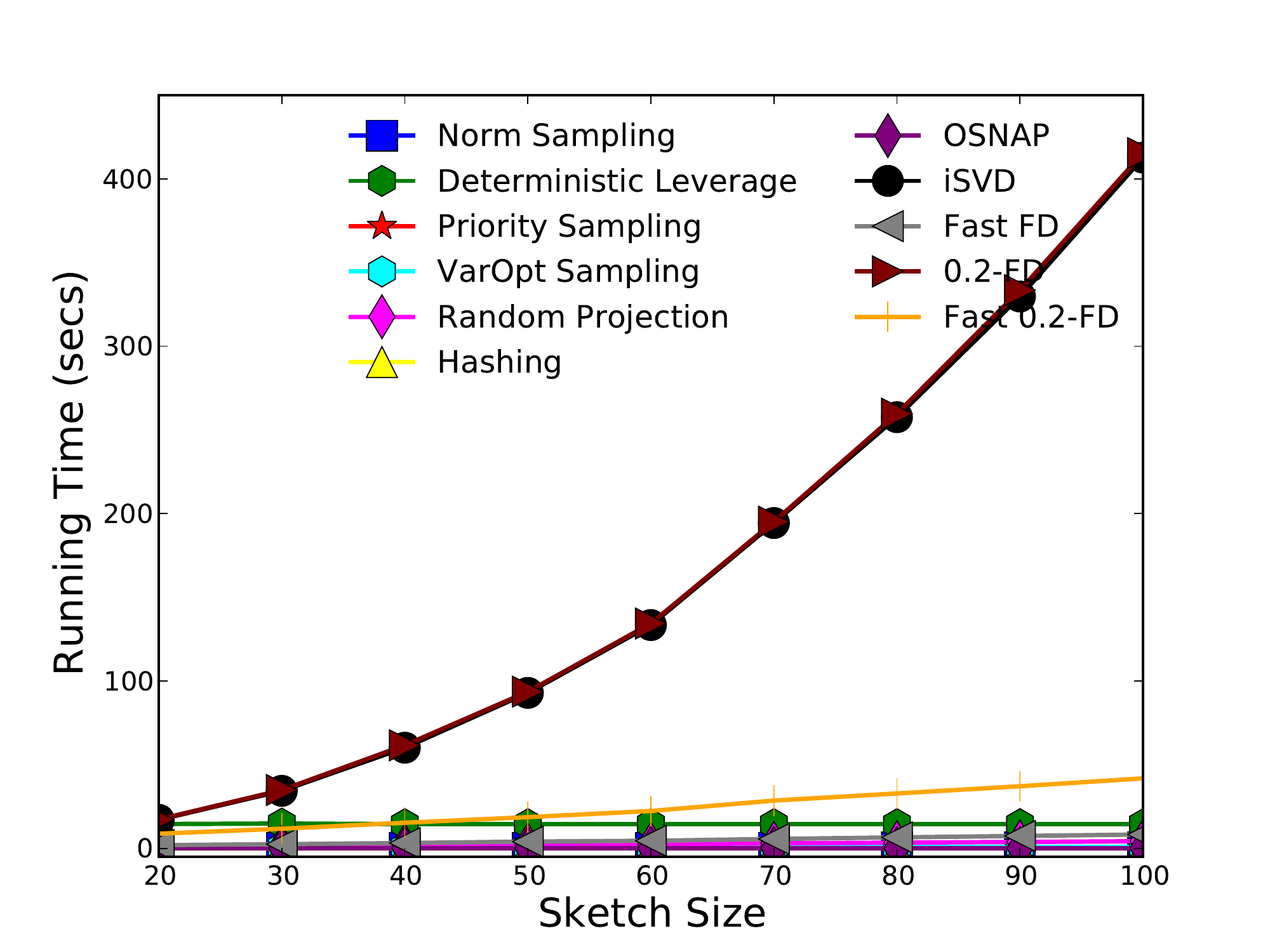}
\includegraphics[width=.305\linewidth]{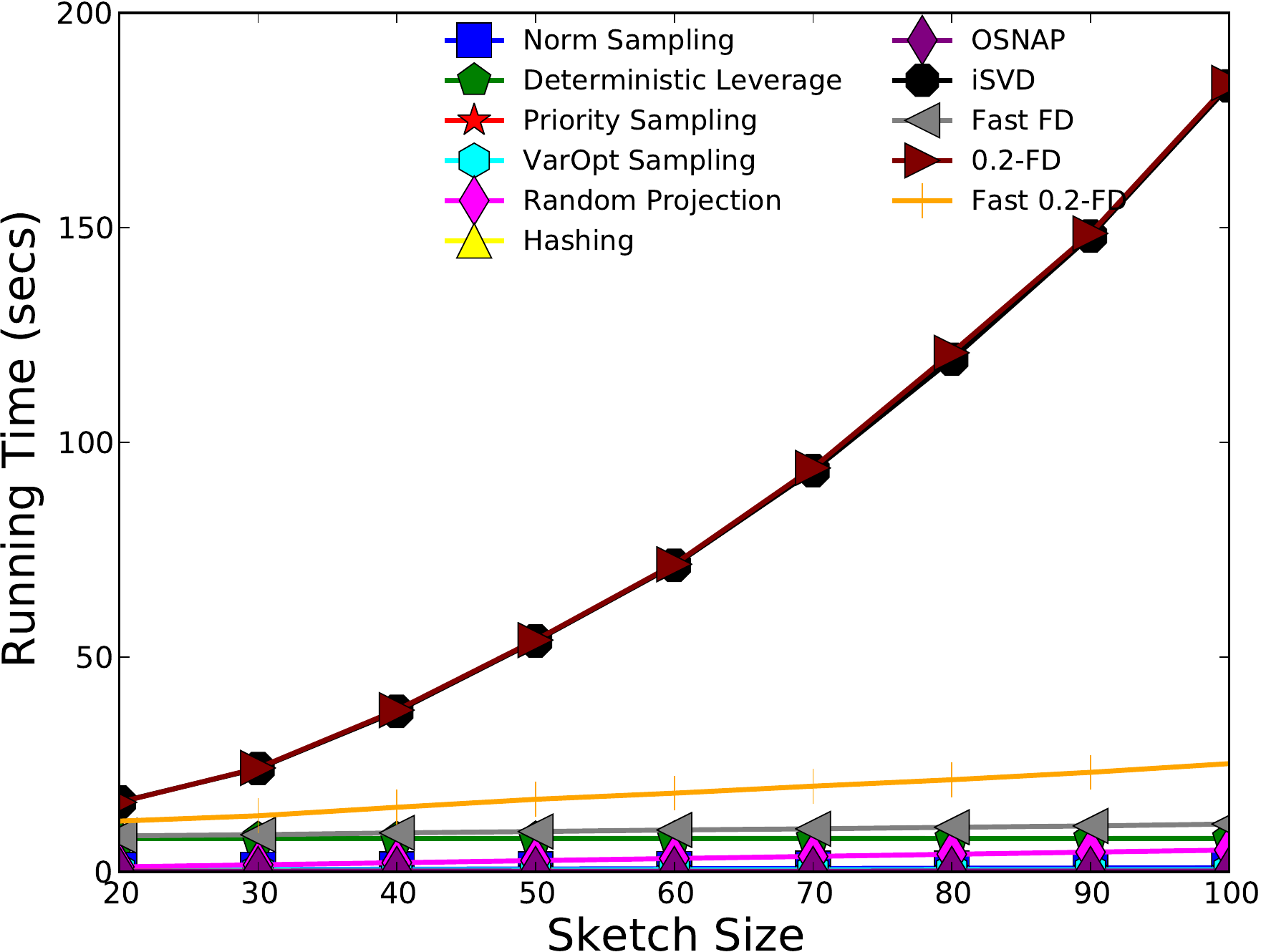}
\caption{\label{fig:global} Leading algorithms on \s{Birds}(left), \s{Spam}(middle), and \s{Random Noisy}(30)(right).} 
\end{centering}
\end{figure}

Figure \ref{fig:global} shows the covariance error, projection error, as well as the runtime for various sketch sizes of $\ell = 20$ to $100$ for the the leading algorithms from each category.  

We can observe that the iterative algorithms achieve much smaller errors, both covariance and projection, than all other algorithms, sometimes matched by \s{Deterministic Leverage}.  However, they are also significantly slower (sometimes a factor of $20$ or more) than the other algorithms.  The exception is \s{Fast \FD} and \s{Fast $0.2$-\FD}, which are slower than the other algorithms, but not significantly so.  

We also observe that for the most part, there is a negligible difference in the performance between the sampling algorithms and the projection algorithms, except for the \s{Random Noisy} dataset where \s{Hashing} and \s{OSNAP} result in worse projection error.

\begin{figure}[t!]
\begin{centering}
{\tiny \textsf{Birds} \hspace{47mm} \textsf{ConnectUS} \hspace{43mm} \textsf{CIFAR-10}} 
\\
\includegraphics[width=.32\linewidth]{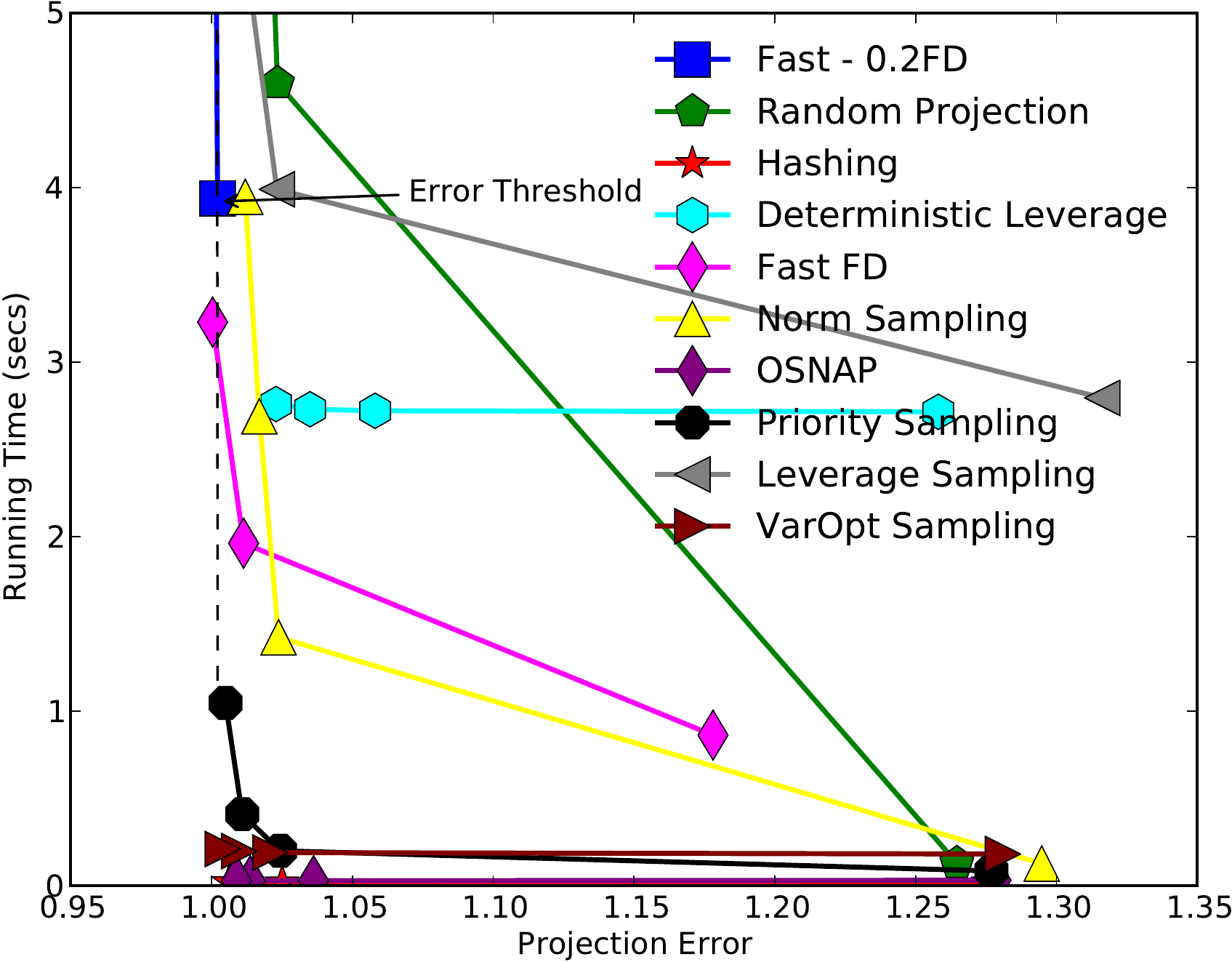}
\includegraphics[width=.32\linewidth]{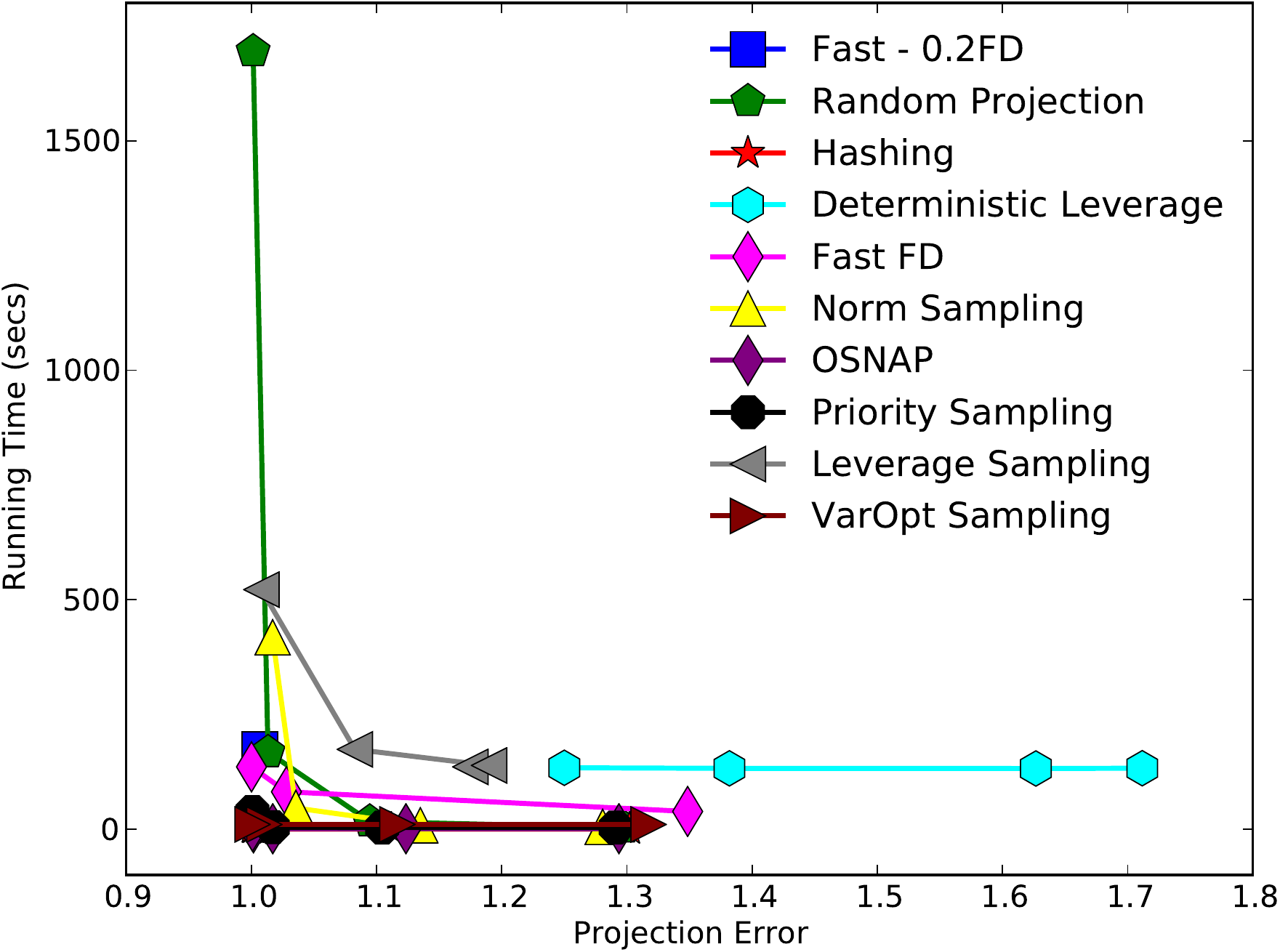}
\includegraphics[width=.32\linewidth]{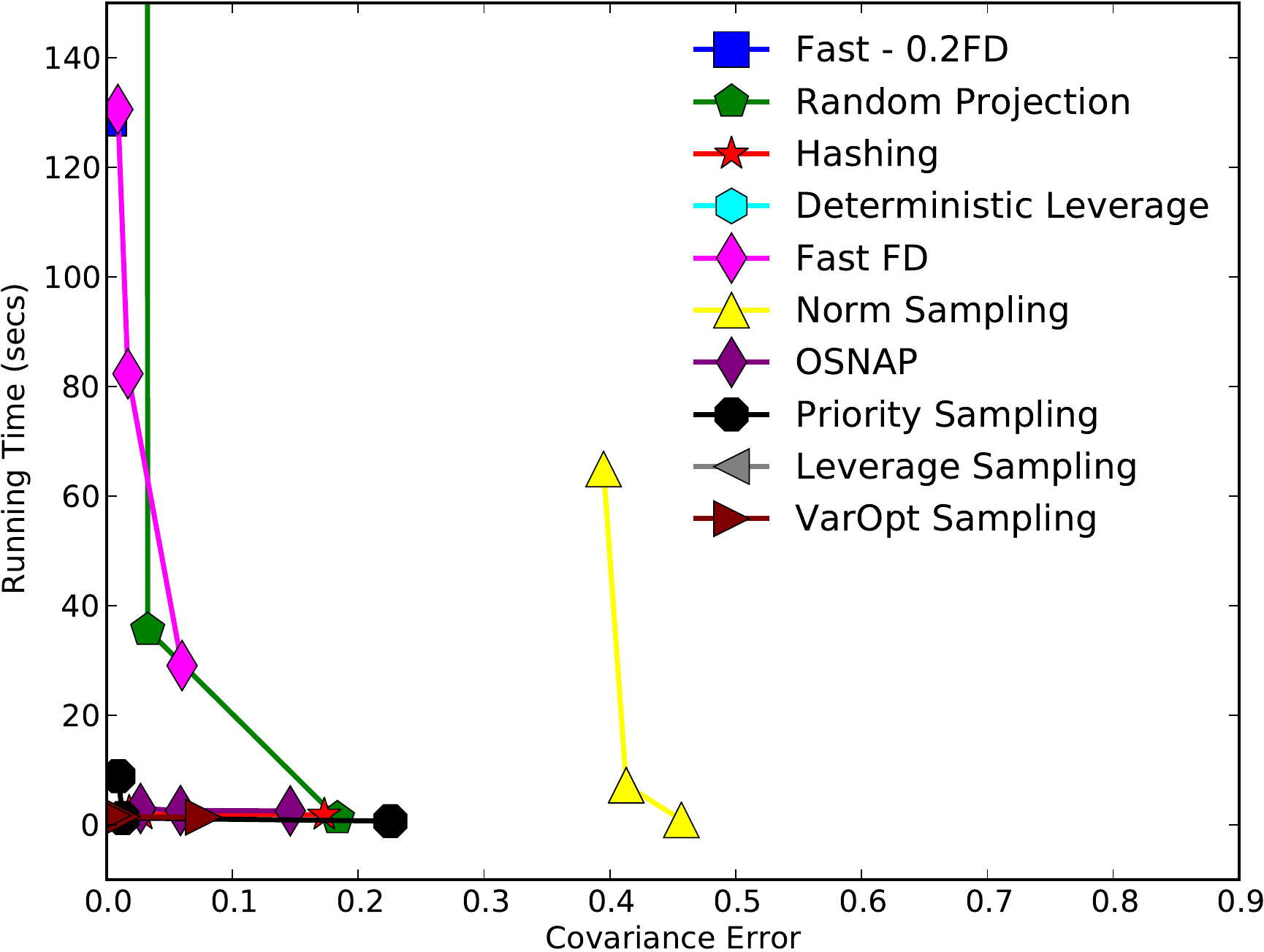}
\includegraphics[width=.32\linewidth]{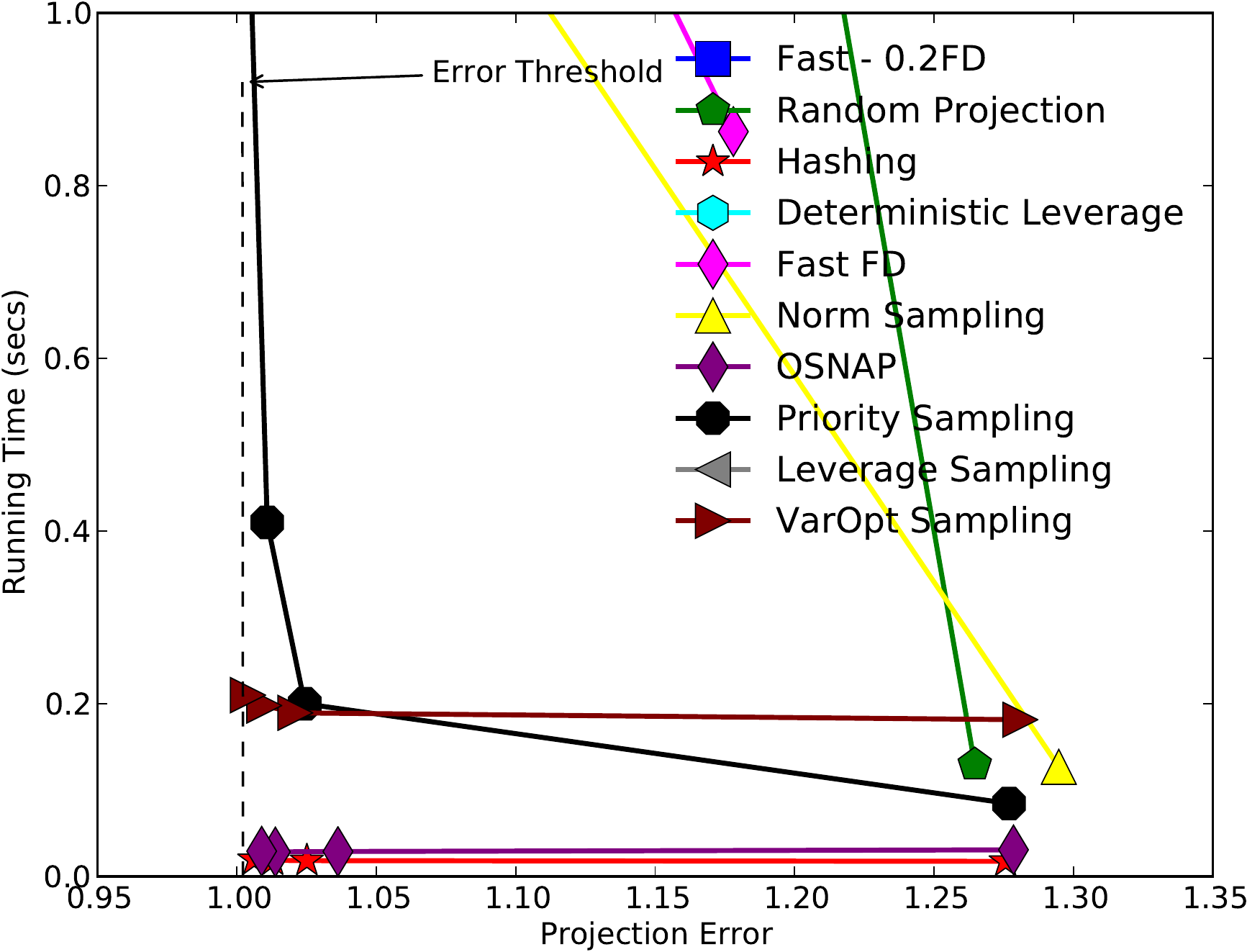}
\includegraphics[width=.32\linewidth]{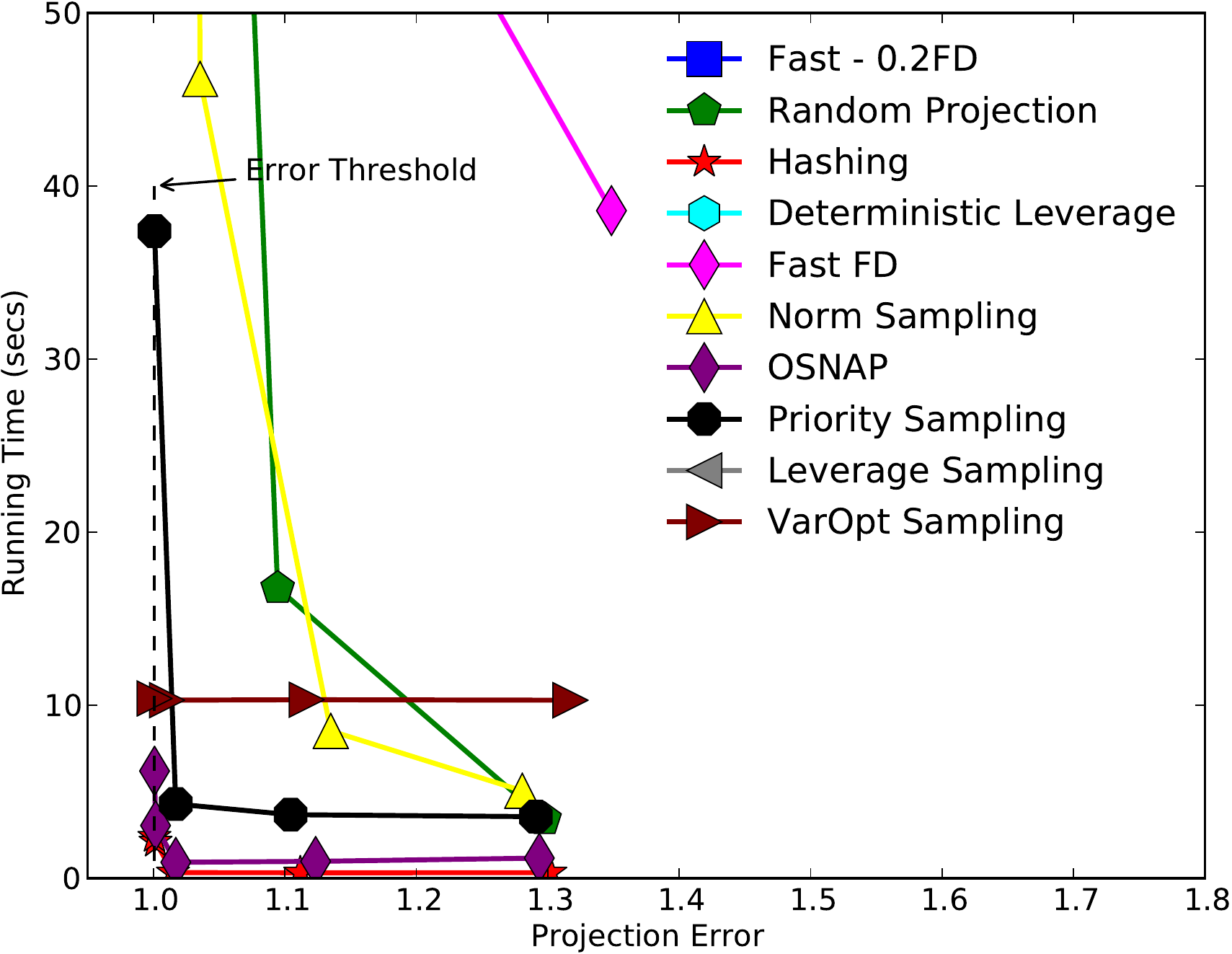}
\includegraphics[width=.32\linewidth]{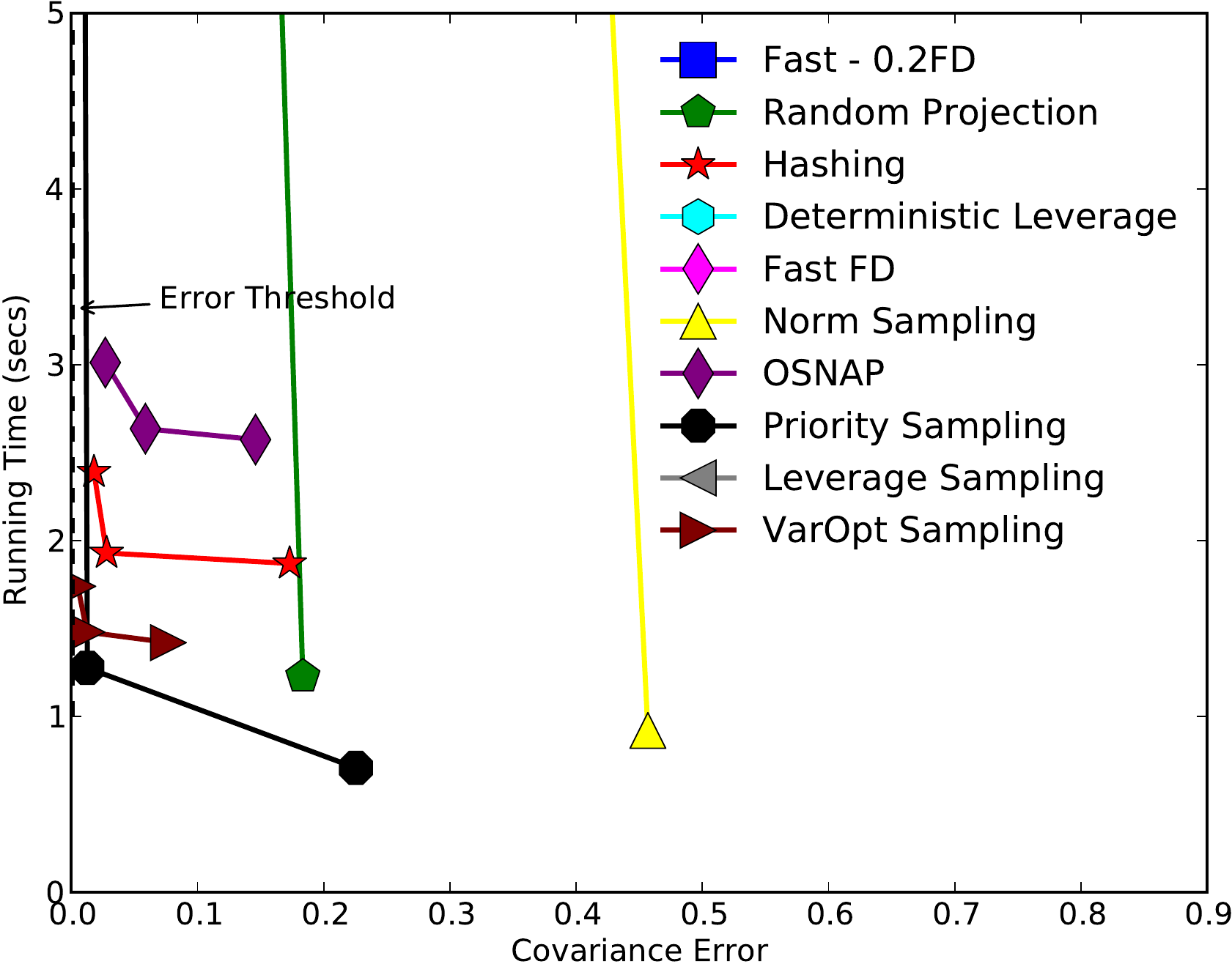}
\caption{\label{fig:et}
Projection error versus time on \s{Birds} and \s{ConnectUS} as well as Covariance error versus time on \s{CIFAR-10}. The second line shows close-ups}
\end{centering}
\end{figure}

However, if we allow a much large sketch size for faster runtime and small error, then these plots do not effectively demonstrate which algorithm performs best.  Thus in Figure \ref{fig:et} we run the leading algorithms on \s{Birds} as well as larger datasets, \s{ConnectUS} which is sparse and \s{CIFAR-10} which is dense.  We plot the error versus the runtime for various sketch sizes ranging up to $\ell = 10{,}000$.  The top row of the plots shows most data points to give a holistic view, and the second row zooms in on the relevant portion.  

For some plots, we draw an \emph{Error Threshold} vertical line corresponding to the error achieved by \s{Fast $0.2$-FD} using $\ell=20$.  Since this error is typically very low, but in comparison to the sampling or projection algorithms \s{Fast $0.2$-FD} is slow, this threshold is a useful target error rate for the other leading algorithms.  

We observe that \s{Fast FD} can sometimes match this error with slightly less time (see on \s{Birds}), but requires a larger sketch size of $\ell=100$.  
Additionally \s{VarOpt}, \s{Priority Sampling}, \s{Hashing}, and \s{OSNAP} can often meet this threshold.  Their runtimes can be roughly $100$ to $200$ times faster, but require sketch sizes on the order of $\ell = 10{,}000$ to match the error of \s{Fast $0.2$-FD} with $\ell = 20$.  

Among these fast algorithms requiring large sketch sizes we observe that \s{VarOpt} scales better than \s{Priority Sampling}, and that these two perform best on \s{CIFAR-10}, the large dense dataset.  They also noticeably outperform \s{Norm Sampling} both in runtime and error for the same sketch size.  
On the sparse dataset \s{ConnectUS}, algorithms \s{Hashing} and \s{OSNAP} seem to dominate \s{Priority Sampling} and \s{VarOpt}, and of those two \s{Hashing} performs slightly better.  

To put this space in perspective, on CIFAR-10 ($n$ = $60{,}000$ rows, 1.4GB memory footprint), to approximately reach the \emph{error threshold} 
\s{Hashing} needs $\ell = 10{,}000$ and 234MB in $2.4$ seconds, 
\s{VarOpt Sampling} requires $\ell=5{,}000$ and 117MB in $1.2$ seconds, 
\s{Fast FD} requires $\ell = 100$ and 2.3MB in $130$ seconds, and
\s{Fast $0.2$-FD} requires $\ell = 20$ and 0.48MB in $128$ seconds.  
All of these will easily fit in memory of most modern machines.  
The smaller sketch by \s{Fast $0.2$-FD} will allow expensive downstream applications (such as deep learning) to run much faster.  Alternatively, the output from \s{VarOpt Sampling} (which maintains interpretability of original rows) could be fed into \s{Fast $0.2$-FD} to get a compressed sketch in less time.


\paragraph{Reproducibility.}
We provide public access to our results using a testbed facility \textit{APT} \cite{apt}. APT is a platform where researchers can perform experiments and keep them public for verification and validation of the results. We provide our code, datasets, and experimental results in our APT profile with detailed description on how to reproduce, available at: \texttt{\url{http://aptlab.net/p/MatrixApx/MatrixApproxComparision}}.

\bibliographystyle{plain}
\bibliography{ipca-journal}

\appendix

\section{Generalized Error Bounds for \FD}.  
\label{app:aFD}

\begin{lemma}
In any such algorithm, for any unit vector $x$: 
\[
0 \leq \|Ax\|^2 - \|B x\|^2 \leq \|A -A_k\|_F^2 / (\alpha \ell-k)
\]
\end{lemma}
\begin{proof}
In the following, $y_i$ correspond to the singular vectors of $A$ ordered with respect to a decreasing corresponding singular value order.
\begin{align*}
\alpha \Delta \ell	& \le \|A\|_F^2 - \|B\|_F^2 										& \text{via Property 3} \\ 
		& = \sum_{i=1}^k \|A y_i\|^2 + \sum_{i=k+1}^d \|A y_i\|^2 - \|B\|_F^2 		&  \|A\|_F^2 = \sum_{i=1}^{d}\|Ay_i\|^2\\
		& = \sum_{i=1}^k \|A y_i\|^2 + \|A - A_k \|_F^2 - \|B\|_F^2 				& \\ 
		& \leq \|A - A_k \|_F^2 + \sum_{i=1}^k \left (\|A y_i\|^2 - \|B y_i\|^2 \right) 	& \text{$\sum_{i=1}^k \|B y_i\|^2 < \|B\|_F^2$} \\
		& \leq \|A - A_k \|_F^2 + k\Delta. 								& \text{via Property 2}			
\end{align*}  
Solving $\alpha \Delta \ell \leq \|A - A_k\|_F^2 + k \Delta$ for $\Delta$ to obtain $\Delta \leq \|A - A_k\|_F^2 / (\alpha \ell - k)$, which combined with Property 1 and Property 2 proves the lemma.
\end{proof}

\begin{lemma}
Any such algorithm described above, satisfies the following error bound
\[
\|A - \pi_{B_k}(A)\| \leq \alpha \ell / (\alpha \ell -k) \|A - A_k\|_F^2 
\]
Where $\pi_{B_k}(\cdot)$ represents the projection operator onto $B_k$, the top $k$ singular vectors of $B$.  
\end{lemma}
\begin{proof}
Here, $y_i$ correspond to the singular vectors of $A$ as above and $v_i$ to the singular vectors of $B$ in a similar fashion.
\begin{align*}
\|A - \pi_{B_k}(A)\|_F^2  	&=  \|A\|_F^2 - \|\pi_{B_k}(A)\|_F^2 = \|A\|_F^2 - \sum_{i=1}^k \|A v_i\|^2				& \text{Pythagorean theorem} \\ 
					& \leq  \|A\|_F^2 - \sum_{i=1}^k\|B v_i\|^2 										& \text{via Property 1} \\
					& \leq  \|A\|_F^2 - \sum_{i=1}^k \|B y_i\|^2    									& \text{since $\sum_{i=1}^j \|B v_i\|^2 \geq \sum_{i=1}^j \|B y_i\|^2$} \\ 
					&\leq \|A\|_F^2 - \sum_{i=1}^k (\|A y_i\|^2 - \Delta)  								& \text{via Property 2} \\ 
					& = \|A\|_F^2 - \|A_k\|_F^2 + k\Delta 											& \\ 
					&\leq \|A - A_k\|_F^2 + \frac{k}{\alpha \ell-k} \|A - A_k\|_F^2 								& \text{ by } \Delta \leq \|A - A_k\|_F^2 / (\alpha \ell - k) \\
					& = \frac{\alpha \ell}{\alpha \ell-k} \|A - A_k\|_F^2.  											&
\end{align*}
This completes the proof of lemma.
\end{proof}

\section{Error Bounds for \SSD}.  
\label{app:SSD}

To understand the error bounds for \RRSS, we will consider an arbitrary unit vector $x$.  We can decompose $x = \sum_{j=1}^d \beta_j v_j$ where $\beta^2_j = \langle x, v_j \rangle^2 > 0$ and $\sum_{j=1}^d \beta_j^2 = 1$.  
For notational convenience, without loss of generality, we assume that $\beta_j = 0$ for $j > \ell$.  Thus $v_{\ell-1}$ represents the entire component of $x$ in the null space of $B$ (or $B_{[i]}$ after processing row $i$).

To analyze this algorithm, at iteration $i \geq \ell$, we consider a $d \times d$ matrix $\bar{B}_{[i]}$ that has the following properties:  $\|B_{[i]} v_j\|^2 = \|\bar{B}_{[i]} v_j\|^2$ for $j < \ell-1$ and $j=\ell$, and $\|\bar{B}_{[i]} v_j\|^2 = \delta_i$ for $j=\ell-1$ and $j > \ell$.  
This matrix provides the constant but bounded overcount similar to the SS sketch.  
Also let $A_{[i]} = [a_1; a_2; \ldots; a_i]$.  

\begin{lemma} \label{lem:SS-bar}
For any unit vector $x$ we have $0 \leq \|\bar{B}_{[i]} x\|^2 - \|A_{[i]} x\|^2 \leq 2\delta_i$
\end{lemma}
\begin{proof}
We prove the first inequality by induction on $i$.  It holds for $i = \ell-1$, since $B_{[\ell-1]} = A_{[\ell-1]}$, and $\|\bar{B}_{[i]} x\|^2 \geq \|B_{[i]} x\|^2$.  
We now consider the inductive step at $i$.  Before the reduce-rank call, the property holds, since adding row $a_i$ to both $A_{[i]}$ (from $A_{[i-1]}$) and $C_{[i]}$ (from $B_{[i-1]}$) increases both squared norms equally (by $\langle a_i , x\rangle^2$) and the left rotation by $U^T$ also does not change norms on the right.  On the reduce-rank, norms only change in directions $v_\ell$ and $v_{\ell-1}$.  Direction $v_\ell$ increases by $\delta_i$, and in $\bar{B}_{[i]}$ the directions $v_{\ell-1}$ also does not change, since it is set back to $\delta_i$, which it was before the reduce-rank.  

We prove the second inequality also by induction, where it also trivially holds for the base case $i = \ell-1$.  Now we consider the inductive step, given it holds for $i-1$.  First obverse that $\delta_i \geq \delta_{i-1}$ since $\delta_i$ is at least the $(\ell-1)$st squared singular value of $B_{[i-1]}$, which is at least $\delta_{i-1}$.  Thus, the property holds up to the reduce rank step, since again, adding row $a_i$ and left-rotating does not affect the difference in norms.  
After the reduce rank, we again only need to consider the two directions changed $v_{\ell-1}$ and $v_\ell$.  By definition 
\[
\|A_{[i]} v_{\ell-1}\|^2 + 2 \delta_i \geq \delta_i = \|\bar B_{[i]} v_{\ell-1}\|^2,
\] 
so direction $v_{\ell-1}$ is satisfied.  Then 
\[
\|\bar B_{[i]} v_\ell\|^2 = \|B_{[i]} v_\ell\|^2 = \delta_i + \|C_{[i]} v_\ell\|^2 \leq 2 \delta_i
\] 
and $0 \leq \|A_{[i]} v_\ell\|^2 \leq \|\bar B_{[i]} v_\ell\|^2$.  Hence $\|\bar B_{[i]} v_\ell\|^2 - \|A_{[i]} v_\ell\|^2 \leq 2 \delta_i - 0$, satisfying the property for direction $v_\ell$, and completing the proof.  
\qed\end{proof}

Now we would like to prove the three Properties needed for relative error bounds for $B = B_{[n]}$.  But this does not hold since $\|B\|_F^2 = \|A\|_F^2$ (an otherwise nice property), and $\|\bar B\|_F^2 \gg \|A\|_F^2$.  Instead, we first consider yet another matrix $\hat B$ defined as follows with respect to $B$.  $B$ and $\hat B$ have the same right singular values $V$.  Let $\delta = \delta_n$, and for each singular value $\sigma_j$ of $B$, adjust the corresponding singular values of $\hat B$ to be $\hat \sigma_j = \max\{0, \sqrt{\sigma_j^2 - 2\delta}\}$.  

\begin{lemma}\label{lem:SS-hat}
For any unit vector $x$ we have 
$0 \leq \|A x\|^2 - \|\hat B x\|^2 \leq 2 \delta$
and
$\|A\|_F^2 - \|\hat B\|_F^2 \geq \delta (\ell-1)$. 
\end{lemma}
\begin{proof}
Directions $v_j$ for $j > \ell-1$, the squared singular values are shrunk by at least $\delta$.  The squared singular value is already $0$ for direction $v_{\ell-1}$.  And the singular value for direction $v_\ell$ is shrunk by $\delta$ to be exactly $0$.  Since before shrinking $\|B\|_F^2 = \|A\|_F^2$, the second expression in the lemma holds.  

The first expression follows by Lemma \ref{lem:SS-bar} since $\bar B$ only increases the squared singular values in directions $v_j$ for $j = \ell-1$ and $j > \ell$ by $\delta$, which are $0$ in $\hat B$.  And other directions $v_j$ are the same for $\bar B$ and $B$ and are at most $2 \delta$ larger than in $A$.  
\end{proof}

Thus $\hat B$ satisfies the three Properties.  We can now state the following property about $B$ directly, setting $\alpha=(1/2)$, adjusting $\ell$ to $\ell-1$, then adding back the at most $2 \delta = \Delta \leq \|A - A_k\|_F^2/ (\alpha \ell - \alpha -k)$ to each directional norm.  

\begin{theorem}\label{thm:SS}
After obtaining a matrix $B$ from \SSD on a matrix $A$ with parameter $\ell$, the following properties hold:
\begin{itemize} \vspace{-.1in}
\item[$\bullet$] $\|A\|_F^2 = \|B\|_F^2$.
\item[$\bullet$] for any unit vector $x$ and for $k < \ell/2-1/2$, we have $| \|A x\|^2 - \|B x \|^2 | \leq \|A - A_k\|_F^2 / (\ell/2 - 1/2 - k)$.
\item[$\bullet$] for $k < \ell/2-1$ we have $\|A - \pi_B^k(A)\|_F^2 \leq \|A - A_k\|_F^2 (\ell-1) / (\ell - 1 - 2k)$.  
\end{itemize}
\end{theorem}

\section{Adapting Known Error Bounds}
\label{app:cov-err}
Not many algorithms state bounds in terms of covariance error or the projection bounds in the particular setting we consider. Here we show using properties of these algorithms they achieve covariance error bound too.  

\subsection{Column Sampling Algorithms}
\label{app:rs-convert}
Almost all column sampling algorithms have a projection bound in terms of $\|A - \pi_{B}A\|^2_F \leq f(\eps) \|A-A_k\|_F^2$ where $A\in \R^n\times d$ is the input matrix and $B \in \R^{\ell \times d}$ is the output sketch. Here we derive the other type of bound, \s{cov-err}, for the two main algorithms in this regime; i.e. \s{Norm Sampling} and \s{Leverage Sampling}.  In our proof, we use a variant of Chernoff-Hoeffding inequality: 
Consider a set of $r$ independent random variables $\{X_1,\cdots, X_r\}$ where $0 \leq X_i \leq \Delta$. Let $M = \sum_{i=1}^r X_i$, then for any $\alpha \in (0,1/2)$
\[
\Pr\left[|M-\E[M]| > \alpha \right]\leq 2\exp\left(\frac{-2\alpha^2}{r \Delta^2}\right).
\]

\begin{lemma}
\label{lem:normsampling}
Let $B \in \R^{\ell \times d}$ with $\ell = O(d/\eps^2)$ be the output of \s{Norm Sampling}.   Then with probability $99/100$, for all unit vectors $x \in \R^d$
\[
\s{cov-err}(A,B) = \frac{|\|Bx\|^2 - \|Ax\|^2|}{\|A\|_F^2} \leq \eps.  
\]
\end{lemma}
\begin{proof}
Consider any unit vector $x \in \b{R}^d$.  
Define $\ell$ independent random variables $X_i = \langle b_i, x\rangle^2$ for $i=1,\ldots,\ell$. Recall that \s{Norm Sampling} selects row $a_j$ to be row $b_i$ in the sketch matrix $B$ with probability $\Pr(b_i \leftarrow a_j) = \|a_j\|^2/\|A\|_F^2$, and rescales the sampled row as $\|b_i\|^2 = \|a_j\|^2/(\ell \Pr(b_i \leftarrow a_j)) = \|A\|_F^2/\ell$.  Knowing these, we bound each $X_i$ as $0 \leq X_i \leq \|b_i\|^2 = \|A\|_F^2/\ell$ therefore $\Delta_i = \|A\|_F^2/\ell$ for all $X_i$s. 
Setting $M = \sum_{i=1}^{\ell} X_i = \|Bx\|^2$, we observe 
\[
\E[M] 
= 
\sum_{i=1}^{\ell} \E[X_i] 
= 
\sum_{i=1}^{\ell} \sum_{j=1}^n \Pr(b_i \leftarrow a_j) \left\langle \frac{a_j}{\sqrt{\ell \Pr(b_i \leftarrow a_j)}},x \right\rangle^2 
= 
\sum_{i=1}^{\ell} \sum_{j=1}^n \frac{1}{\ell}\langle a_j,x \rangle^2 
= 
\|Ax\|^2.  
\]
Finally using the Chernoff-Hoeffding bound and setting $\alpha = \eps \|A\|_F^2$ yields
\[
\Pr\left[|\|Bx\|^2 - \|Ax\|^2| > \eps \|A\|_F^2\right] 
\leq 
2 \exp\left(\frac{-2 (\eps \|A\|_F^2)^2}{\ell (\|A\|_F^2/\ell)^2}\right) 
= 
2\exp\left(-2\eps^2 \ell\right) 
\leq 
\delta.  
\]
Letting the probability of failure for that $x$ be $\delta = 1/100$, and solving for $\ell$ in the last inequality, we obtain $\ell \geq \frac{1}{2 \eps^2} \ln(2/\delta) = \frac{1}{2 \eps^2} \ln(200)$.  

However, this only holds for a single direction $x$; we need this to hold for all unit vectors $x$.  It can be shown, that allowing $\alpha = O(\eps)\cdot \|A\|_F^2$, we actually only need this to hold for a net $T$ of size $t =2^{O(d)}$ such directions $x$~\cite{Woo14}.  Then by the union bound, setting $\delta = 1/(100 t)$ this will hold for all unit vectors in $T$, and thus (after scaling $\eps$ by a constant) we can solve for $\ell = O((1/\eps^2) \ln(t)) = O(d/\eps^2)$.  
Hence with probability at least $99/100$ we have $\s{cov-err}(A,B) = \frac{| \|B x\|^2 - \|A x\|^2 |}{\|A\|_F^2} \leq \eps$ for all unit vectors $x$.  
\end{proof}

We would like to apply a similar proof for \s{Leverage Sampling}, but we do not obtain a good bound for $\Delta$ in the Chernoff-Hoeffding bound.  We rescale norm of each selected row $b_i$ as 
\[
\|b_i\|^2
= 
\frac{\|a_j\|^2}{\ell} \cdot \frac{1}{\Pr(b_i \leftarrow a_j)} 
=
\frac{\|a_j\|^2}{\ell} \cdot \frac{S_k}{s_j^{(k)}},
\]
where $s^{(k)}_j$ is the rank-$k$ leverage score of $a_j$, and $S_k = \sum_{j=1}^n s^{(k)}_j = k$.  Unfortunately, $s_j^{(k)}$ can be arbitrarily small compared to $S_k$ (e.g., if $a_j$ lies almost entirely outside the best rank-$k$ subspace).  And thus we do not have a finite bound on $\Delta$.  As such we assume that $S_k/s^{(k)}_j \leq \beta$ for an absolute constant $\beta$, and then obtain a bound based on $\beta$.  

\begin{lemma}
Let $B \in \R^{\ell \times d}$ with $\ell = O(d\beta^2/\eps^2)$ be the output of \s{Leverage Sampling}. Under the assumption that rank-$k$ leverage score of row $a_j$ is bounded as $s_j^{(k)} \geq \beta S_k$ for a fixed constant $\beta > 0$, then with probability $99/100$, for all unit vectors $x \in \R^d$
\[
\s{cov-err} = |\|Bx\|^2 - \|Ax\|^2|/\|A\|_F^2 \leq \eps
\]

\end{lemma}
\begin{proof}
Similar to Lemma \ref{lem:normsampling}, we define $\ell$ random variables $X_i = \langle b_i,x \rangle^2$ for $i=1,\cdots,\ell$. \s{Leverage Sampling} algorithm selects row $a_j$ to be row $b_i$ with probability $\Pr(b_i \leftarrow a_j) = s^{(k)}_j/S_k$ and rescales sampled rows as $\|b_i\|^2 = \|a_j\|^2/(\ell \Pr(b_i \leftarrow a_j)) = (\|a_j\|^2/\ell)(S_k/s^{(k)}_j) \leq \|a_j\|^2/(\beta \ell) \leq \|A\|_F^2/(\beta \ell)$. 
Therefore $\Delta_i = \|A\|_F^2/(\beta \ell)$ for all $X_i$s.
Setting $M = \sum_{i=1}^{\ell} X_i = \|Bx\|^2$ we observe:
\[
\E[M] = \sum_{i=1}^{\ell} \E{[X_i]} = \sum_{i=1}^{\ell} \sum_{j=1}^n \Pr(b_i \leftarrow a_j) \left\langle \frac{a_j}{\sqrt{\ell \Pr(b_i \leftarrow a_j)}},x \right\rangle^2 = \sum_{i=1}^{\ell} \sum_{j=1}^n \frac{1}{\ell}\langle a_j,x \rangle^2 = \|Ax\|^2
\]
Using Chernoff-Hoeffding bound with parameter $\alpha = \eps\|A\|_F^2$ gives:
\[
\Pr\left[|\|Bx\|^2 - \|Ax\|^2| > \eps \|A\|_F^2\right] \leq 2 \exp\left(\frac{-2\eps^2 \|A\|_F^4}{\sum_{j=1}^{\ell}(\beta^2/\ell^2) \|A\|_F^4}\right) = 2\exp\left(\frac{-2\ell \eps^{2}}{\beta^2}\right) \leq \delta
\]
Again letting the probability of failure for that $x$ be $\delta = 1/100$, we obtain $\ell \geq \frac{\beta^2}{2 \eps^2} \ln(2/\delta) = \frac{\beta^2}{2 \eps^2} \ln(200)$.  

In order to hold this for all unit vectors $x$, again we consider a net $T$ of size $t =2^{O(d)}$ unit directions $x$~\cite{Woo14}.  Then by the union bound, setting $\delta = 1/(100 t)$ this will hold for all such vectors in $T$, and thus we can solve for $\ell = O((\beta^2/\eps^2) \ln(t)) = O(d\beta^2/\eps^2)$. 
Hence with probability at least $99/100$ we have $\s{cov-err}(A,B) = \frac{| \|B x\|^2 - \|A x\|^2 |}{\|A\|_F^2} \leq \eps$ for all unit vectors $x$.  
\end{proof}

\subsection{Random Projection Algorithms}
\label{app:rp-convert}
Most of random projection algorithms state a bound (with constant probability) that they can create a matrix $B = SA$ such that $\|Bx\| = (1\pm\eps)\|Ax\|$ (e.g. $(1-\eps)\|Ax\| \leq \|Bx\| \leq (1+\eps)\|A x\|$) for all $x \in \b{R}^d$.  

Here we relate this to \s{cov-err} and \s{proj-err}.  We first show that whether this bound is squared only affects $\eps$ by a constant factor.  
\begin{lemma}
\label{lem:unsq}
For $\eps \in (0,\frac{1}{4})$, the inequality $\|Bx\| = (1 \pm \eps) \|Ax\|$ implies $\|Bx\|^2 = (1 \pm 3\eps) \|Ax\|^2$.  
\end{lemma}
\begin{proof} 
The upper bound follows since $(1+\eps)^2 = 1 + 2\eps + \eps^2 \leq 1+3\eps$ for $\eps \in (0,\frac{1}{4})$.  
The lower bound follows since $(1-\eps)^2 = 1 + \eps^2 - 2\eps \geq 1-2\eps$ for $\eps \in (0,\frac{1}{4})$.  
\end{proof}

Now to relate this bound to \s{cov-err}, we will use $\rho(A) = \|A\|_F^2 / \|A\|_2^2$, the \emph{numeric rank} of $A$, which is always at least $1$.  

\begin{lemma}
Given a matrix $B$ such that $\|Bx\| = (1 \pm \eps) \|Ax\|$ for all $x \in \b{R}^d$, then when $\eps \in (0,\frac{1}{4})$ 
\[
\s{\em cov-err} = \|A^T A - B^T B\|_2 / \|A\|_F^2 \leq 3\eps / \rho(A).
\] 
\end{lemma}
\begin{proof}

Using Lemma \ref{lem:unsq}, $| \|A x\|^2 -\|B x\|^2 | \leq 3\eps \|Ax\|^2$.  Now restrict $\|x\| =1$, so then 
\[
 x^T (A^T A - B^T B) x 
=
 x^T A^T A x - x^T B^T B x  
=
| \|A x\|^2 - \|B x\|^2 | 
\leq 
3\eps \|Ax\|^2.
\]
Since this holds for all $x$ such that $\|x\|=1$, then it holds for the $x = x^*$ which maximizes the left hand size so ${x^*}^T (A^T A - B^T B)x^* = \|A^T A - B^T B\|_2$ so 
\[
\|A^T A - B^T B\|_2 = x^{*^{T}} (A^T A - B^T B) x^* \leq 3 \eps \|A x^*\|^2 \leq 3 \eps \|A\|_2^2.
\]
Dividing both sides by $\|A\|_F^2 = \rho(A) / \|A\|_2^2$ completes the proof.  
\end{proof}

We next relate this to the \s{proj-err}.    
We note that it may be possible that the full property $\|Bx\| = (1 \pm \eps) \|Ax\|$ may not be necessary to obtain a bound on \s{proj-err}, but we are not aware of an explicit statement otherwise.  There are bounds (see \cite{Woo14}) where one reconstructs a matrix $\hat A$ which obtains the bounds below in place of $\pi_{B_k}(A)$, and these have roughly $1/\eps$ dependence on $\eps$; however, they also have a factor $n$ in their size.  

\begin{lemma}
Given a matrix $B$ such that $\|Bx\| = (1 \pm \eps) \|Ax\|$ for all $x \in \b{R}^d$, then when $\eps \in (0,\frac{1}{4})$ 
\[
\s{\em proj-err} = \|A - \pi_{B_k}(A)\|_F^2 / \|A-A_k\|_F^2 \leq (1+24\eps).
\] 
\end{lemma}
\begin{proof}
Let $V = [v_1, v_2, \ldots, v_d]$ be the right singular vectors of $A$, so that $\|A_k\|_F^2 = \sum_{i=1}^k \|A v_i\|^2$ and $\|A - A_k\|_F^2 = \sum_{i=k+1}^d \|A v_i\|^2$.  
It follows from $\|B_k\|_F^2 \geq \sum_{i=1}^k \|B v_i\|^2$ and Lemma \ref{lem:unsq} that 
\[
\|B-B_k\|_F^2 
\leq 
\sum_{i=k+1}^d \|B v_i\|^2 
\leq 
\sum_{i=k+1}^d \frac{1}{1-3\eps} \|A v_i\|^2
\leq
\frac{1}{1-3 \eps} \|A-A_k\|_F^2.
\]
Let $R = [r_1, r_2, \ldots, r_d]$ be the right singular vectors of $B$.  
Then by matrix Pythagorean and Lemma \ref{lem:unsq}
\begin{align*}
\|A - \pi_{B_k}(A)\|^2 
&= 
\sum_{i=k+1}^d \|A r_i\|^2 
\leq 
\sum_{i=k+1}^d (1+3\eps) \|B r_i\|^2
= 
(1+3\eps) \|B-B_k\|_F^2
\\ & \leq
\frac{1+3 \eps}{1-3\eps} \|A-A_k\|_F^2
\leq
(1+24 \eps) \|A - A_k\|_F^2. \qedhere
\end{align*}
\end{proof}

\end{document}